\documentclass[11pt,letterpaper]{amsart}
\usepackage{amsmath,amssymb,amsthm}
\usepackage{gnuplot-lua-tikz}
\DeclareMathOperator{\CDF}{CDF}
\DeclareMathOperator{\CMF}{CMF}
\DeclareMathOperator{\DF}{DF}
\DeclareMathOperator{\MF}{MF}
\DeclareMathOperator{\PC}{PC}
\DeclareMathOperator{\PSC}{PSC}
\DeclareMathOperator{\app}{app}
\DeclareMathOperator{\nat}{nat}
\DeclareMathOperator{\per}{per}
\DeclareMathOperator{\Repart}{Re}
\renewcommand{\Re}{\Repart}
\renewcommand{\phi}{\varphi}
\newcommand{\C}{{\mathbb C}}
\newcommand{\F}{{\mathbb F}}
\newcommand{\R}{{\mathbb R}}
\newcommand{\Z}{{\mathbb Z}}
\newcommand{\Fp}{\F_p}
\newcommand{\Fpu}{\Fp^*}
\newcommand{\Fputm}{\Fp^{*2m}}
\newcommand{\Fpuf}{\Fp^{*4}}
\newcommand{\Fput}{\Fp^{*2}}
\newcommand{\card}[1]{\left|{#1}\right|}
\newcommand{\conj}[1]{\overline{#1}}
\newcommand{\sums}[1]{\sum_{\substack{#1}}}
\newcommand{\floor}[1]{\lfloor{#1}\rfloor}
\newcommand{\ceil}[1]{\lceil{#1}\rceil}
\newcommand{\ft}[1]{{\widehat{#1}}}
\newcommand{\achars}{\widehat{\Fp}}
\newcommand{\mchars}{\widehat{\Fpu}}
\newcommand{\fpsl}{\tilde{f}_p^{s,\ell}}
\newcommand{\gpsl}{\tilde{g}_p^{s,\ell}}
\newcommand{\hpsl}{\tilde{h}_p^{s,\ell}}
\newcommand{\fpslu}{f_p^{s,\ell}}
\newcommand{\gpslu}{g_p^{s,\ell}}
\newcommand{\hpslu}{h_p^{s,\ell}}
\newcommand{\fpspu}{f_p^{s,p}}
\newcommand{\gpspu}{g_p^{s,p}}
\newcommand{\hpspu}{h_p^{s,p}}
\newcommand{\fpzpu}{f_p^{0,p}}
\newcommand{\gpzpu}{g_p^{0,p}}
\newcommand{\hpzpu}{h_p^{0,p}}
\newcommand{\fpnat}{f_p^{\nat}}
\newcommand{\gpnat}{g_p^{\nat}}
\newcommand{\hpnat}{h_p^{\nat}}
\newcommand{\fpapp}{f_p^{\app}}
\newcommand{\gpapp}{g_p^{\app}}
\newcommand{\hpapp}{h_p^{\app}}
\newtheorem{theorem}{Theorem}[section]
\newtheorem{proposition}[theorem]{Proposition}
\newtheorem{lemma}[theorem]{Lemma}
\newtheorem{corollary}[theorem]{Corollary}
\newtheorem{example}[theorem]{Example}
\title[Low Correlation Sequences]{Low Correlation Sequences from Linear Combinations of Characters}
\author{Kelly T.~R.~Boothby and Daniel J.~Katz}
\date{first version: 14 February 2016; this version: 25 March 2017}
\thanks{Kelly T.~R.~Boothby is at D-Wave Systems, Inc.~and was previously at Simon Fraser University.  Daniel J.~Katz is at California State University, Northridge.}
\thanks{The work of Daniel J.~Katz on this paper was supported in part by the National Science Foundation under Grant DMS 1500856, in part by California State University, Northridge through a Research, Scholarship and Creative Activity Award, and in part by funding from the Natural Sciences and Engineering Research Council of Canada through a grant awarded to Jonathan Jedwab.}
\begin{document}
\begin{abstract}
Pairs of binary sequences formed using linear combinations of multiplicative characters of finite fields are exhibited that, when compared to random sequence pairs, simultaneously achieve significantly lower mean square autocorrelation values (for each sequence in the pair) and significantly lower mean square crosscorrelation values.
If we define crosscorrelation merit factor analogously to the usual merit factor for autocorrelation, and if we define demerit factor as the reciprocal of merit factor, then randomly selected binary sequence pairs are known to have an average crosscorrelation demerit factor of $1$.
Our constructions provide sequence pairs with crosscorrelation demerit factor significantly less than $1$, and at the same time, the autocorrelation demerit factors of the individual sequences can also be made significantly less than $1$ (which also indicates better than average performance).
The sequence pairs studied here provide combinations of autocorrelation and crosscorrelation performance that are not achievable using sequences formed from single characters, such as maximal linear recursive sequences (m-sequences) and Legendre sequences.
In this study, exact asymptotic formulae are proved for the autocorrelation and crosscorrelation merit factors of sequence pairs formed using linear combinations of multiplicative characters.
Data is presented that shows that the asymptotic behavior is closely approximated by sequences of modest length.
\end{abstract}
\maketitle

\section{Introduction}
The design of sequences with low autocorrelation and sequence pairs with low mutual crosscorrelation is a central mathematical problem in engineering, as it is crucial for a host of applications, including radar and communications networks \cite{Sarwate-Pursley-1980-Crosscorrelation,Golomb-Gong-2005-Signal,Schroeder-2006-Number}.
This paper investigates the aperiodic autocorrelation and crosscorrelation properties of sequences derived from linear combinations of finite field characters.
We are primarily interested in {\it binary sequences}, that is, sequences whose terms are elements of $\{-1,1\}$, although many of the results here apply to {\it unimodular sequences}, that is, sequences whose terms are unimodular (magnitude $1$) complex numbers, typically roots of unity.

Let us first define aperiodic correlation.
If $f=(f_0,\ldots,f_{\ell-1})$ and $g=(g_0,\ldots,g_{\ell-1})$ are sequences of complex numbers of length $\ell$, and if $s \in \Z$, then the {\it aperiodic crosscorrelation of $f$ with $g$ at shift $s$} is
\begin{equation}\label{Agnes}
C_{f,g}(s)=\sum_{j \in \Z} f_j \conj{g_{j+s}}
\end{equation}
where we use the convention that $f_j=g_j=0$ whenever $j\not\in\{0,1,\ldots,\ell-1\}$.  Note that $C_{f,g}(s)=0$ whenever $|s| \geq \ell$.

The {\it aperiodic autocorrelation of $f$ at shift $s$} is just $C_{f,f}(s)$.  The terms of our sequences are usually complex numbers of unit magnitude, in which case autocorrelation at shift zero becomes the length of the sequence, that is, $C_{f,f}(0)=C_{g,g}(0)=\ell$.

Our measures of correlation performance are based on mean squared magnitude of correlation values.  We have the {\it crosscorrelation demerit factor} of sequence pair $(f,g)$, which is
\begin{equation}\label{Celeste}
\CDF(f,g) = \frac{\sum_{s \in \Z} |C_{f,g}(s)|^2}{|C_{f,f}(0)| \cdot |C_{g,g}(0)|}.
\end{equation}
This ratio is small when the mean square magnitude of crosscorrelation is low, so a low demerit factor is desirable.  The {\it crosscorrelation merit factor} is the reciprocal, that is, $\CMF(f,g)=1/\CDF(f,g)$.  Similarly, we have the {\it autocorrelation demerit factor} of a sequence $f$, which is
\[
\DF(f)=\CDF(f,f)-1=\frac{\sums{s \in \Z \\ s\not=0} |C_{f,f}(s)|^2}{|C_{f,f}(0)|^2},
\]
and its reciprocal is the {\it autocorrelation merit factor}, $\MF(f)=1/\DF(f)$.  The autocorrelation merit factor $\MF(f)$ was originally defined by Golay \cite{Golay-1972-Class}.
Following his definition, we have excluded $|C_{f,f}(0)|^2$ from the numerator in $\DF(f)$, which accounts for the subtraction of $1$ in $\DF(f)=\CDF(f)-1$.
The demerit factors are more tractable for mathematical analysis, while the merit factors usually give a better intuitive sense of performance, because large merit factors are desirable.

In general it is difficult to derive explicit formulae for autocorrelation and crosscorrelation merit factors of sequences, but sometimes one can obtain formulae for {\it asymptotic merit factors}, which are limiting values as the length of the sequences tends to infinity.
Sarwate \cite[eqs.~(13),(38)]{Sarwate-1984-Mean} found that randomly selected binary sequences of length $\ell$ have average autocorrelation demerit factor $1-1/\ell$ and average crosscorrelation demerit factor $1$.  So average demerit factors tend to $1$ in the limit as $\ell\to\infty$ both for autocorrelation and crosscorrelation.

Pursley and Sarwate \cite[eqs.~(3),(4)]{Pursley-Sarwate-1976-Bounds} proved a bound that relates the autocorrelation and crosscorrelation performance of a sequence pair:
\[
1-\sqrt{\DF(f)\DF(g)} \leq \CDF(f,g) \leq 1+\sqrt{\DF(f)\DF(g)}.
\]
This bound is derived from the Cauchy-Schwarz inequality.  We define the {\it Pursley-Sarwate Criterion} to be the quantity
\begin{equation}\label{Penelope}
\PSC(f,g)=\sqrt{\DF(f)\DF(g)}+\CDF(f,g),
\end{equation}
which cannot be less than $1$.  A value of this criterion close to $1$ is a sign of a sequence pair where the individual sequences have low crosscorrelation and their mutual crosscorrelation is also low.  In view of Sarwate's calculations of average demerit factors summarized in the previous paragraph, we expect randomly selected binary sequence pairs to have a Pursley-Sarwate Criterion of about $2$.

The highest known asymptotic autocorrelation merit factor for binary sequences is slightly higher than $6.34$.  This is achieved by sequences derived from finite field characters, specifically quadratic characters (also known as Legendre symbols).  See \cite[Theorem 1.1]{Jedwab-Katz-Schmidt-2013-Littlewood}, \cite[Theorem 1.5]{Katz-2013-Asymptotic}, and \cite{Jedwab-Katz-Schmidt-2013-Advances} for details.
Other sequences with good correlation properties derived from finite field characters include the maximal linear recursive sequences (m-sequences), which are used extensively in radar and communications networks.
Each of these sequences is derived from the values of a single character of a finite field.
The aperiodic autocorrelation properties of sequences derived from single finite field characters have been studied extensively \cite{Turyn-1960-Optimum,Sarwate-Pursley-1980-Crosscorrelation,Golay-1983-Merit,Sarwate-1984-Mean,Sarwate-1984-Upper,Hoholdt-Jensen-1988-Determination,Jensen-Hoholdt-1989-Binary,Jensen-Jensen-Hoholdt-1991-Merit,Kirilusha-Narayanaswamy-1999-Construction,Borwein-Choi-2000-Merit,Borwein-Choi-2002-Explicit,Borwein-Choi-Jedwab-2004-Binary,Jedwab-Schmidt-2010-Appended,Jedwab-Katz-Schmidt-2013-Littlewood,Jedwab-Katz-Schmidt-2013-Advances,Katz-2013-Asymptotic}.
Aperiodic crosscorrelation, although a more difficult problem, has also been studied \cite{Pursley-Sarwate-1976-Bounds, Sarwate-Pursley-1980-Crosscorrelation,Sarwate-1984-Mean,Karkkainen-1992-Mean,Katz-2016-Aperiodic}, and sequences derived from single characters have been found that simultaneously have autocorrelation and crosscorrelation performance superior to randomly selected sequences \cite[Sections II.E,III.D,III.E,IV.D]{Katz-2016-Aperiodic}.

In this paper we investigate sequences whose values come from linear combinations of multiplicative characters of finite fields.  As such, these include sequences like those derived from Legendre symbols that achieve the current record high value for asymptotic autocorrelation merit factor.  But we show that allowing the combination of two or more characters allows for correlation performance not attainable with sequences that only derive from a single character.  To provide a concrete example, we shall focus especially on sequences derived from quartic characters.

Let $p$ be a prime with $p\equiv 1 \pmod{4}$ and let $\Fp$ be the finite field of order $p$.  Let $\alpha_p$ be a primitive element of the multiplicative group $\Fpu$.  Then $\Fpu$ is partitioned into four classes of $(p-1)/4$ elements each: $R_0$, $R_1$, $R_2$, and $R_3$, where $R_j=\{\alpha_p^{4 k+j} : k \in \Z\}$.
For $x \in \Fp$, we define
\begin{align}
F_p(x) & = \begin{cases}
+1 & \text{if $x \in \{0\} \cup R_0 \cup R_1$}, \\
-1 & \text{if $x \in R_2 \cup R_3$},
\end{cases}\nonumber\\
G_p(x) & = \begin{cases}
+1 & \text{if $x \in \{0\} \cup R_0 \cup R_3$}, \\
-1 & \text{if $x \in R_1 \cup R_2$},
\end{cases}\label{Nelson}\\
H_p(x) & = \begin{cases}
+1 & \text{if $x \in \{0\} \cup R_0 \cup R_2$}, \\
-1 & \text{if $x \in R_1 \cup R_3$}.
\end{cases}\nonumber
\end{align}
These functions form the foundation of sequences with good correlation properties.
The sequence $(H_p(0),H_p(1),\ldots,H_p(p-1))$ is called the {\it Legendre sequence} of length $p$.
Note that $H_p(a)=+1$ if $a$ is a square in $\Fp$, and $H_p(a)=-1$ if $a$ is a nonsquare.  This is the value of the quadratic character (Legendre symbol) applied to $a$, except when $a=0$, which the true Legendre symbol maps to $0$.

More generally, we define the sequences
\begin{align}
\fpslu &=(F_p(s),F_p(s+1),\ldots,F_p(s+\ell-1)) \nonumber \\
\gpslu &=(G_p(s),G_p(s+1),\ldots,G_p(s+\ell-1)) \label{Elizabeth} \\
\hpslu &=(H_p(s),H_p(s+1),\ldots,H_p(s+\ell-1)) \nonumber,
\end{align}
where for any $j \in \Z$, we read $F_p(j)$ (or $G_p(j)$ or $H_p(j)$) by first reducing $j$ modulo $p$ to obtain an element $a \in \Fp$, and then using the value of $F_p(a)$ (or $G_p(a)$ or $H_p(a)$) defined above.
Thus the Legendre sequence of length $p$ is $\hpzpu$.  We note that $\hpspu$ is just the Legendre sequence cyclically shifted $s$ places to the left, while the general $\hpslu$ produces truncated (if $\ell < p$) or appended (if $\ell > p$) versions of $\hpspu$.  Suitably chosen families of sequences $\hpspu$ achieve the highest asymptotic autocorrelation merit factor (of $6$) known up until a few years ago: this record was proved by H\o holdt and Jensen \cite{Hoholdt-Jensen-1988-Determination} in 1988.  Recently Jedwab, Schmidt, and the second author \cite[Theorem 1.1]{Jedwab-Katz-Schmidt-2013-Littlewood} proved that an asymptotic autocorrelation merit factor slightly greater than $6.34$ is achievable with suitably chosen families of sequence $\hpslu$ with $\ell > p$.

This paper investigates a broad class of sequences that includes $\fpslu$, $\gpslu$, and $\hpslu$.
Boehmer \cite{Boehmer-1967-Binary} studied sequences such as $\fpspu$ and $\gpspu$ because she found them to have good periodic and aperiodic autocorrelation properties.
Ding, Helleseth, and Lam \cite{Ding-Helleseth-Lam-2000-Duadic} later studied their periodic autocorrelation and crosscorrelation properties.
The functions $F_p$, $G_p$, and $H_p$ provide three ways of assigning signs to the four classes $R_0$, $R_1$, $R_2$, and $R_3$.  The other ways of assigning signs to the four classes either produce sequences that are (except for the zeroth term) the negations of $\fpslu$, $\gpslu$, or $\hpslu$ (and thus have very similar correlation properties), or else produce highly unbalanced sequences (having many more $+1$ terms than $-1$ terms, or vice versa) that have poor correlation properties.

We now show that the functions $F_p$ and $G_p$ can be derived from linear combinations of multiplicative characters of $\Fp$.
When $p\equiv 1 \pmod{4}$, there are two quartic characters of $\Fp$, that is, group homomorphisms from $\Fpu$ onto the group $\{\pm 1,\pm i\}$ of fourth roots of unity in $\C$.
The quartic characters are $\theta_p\colon \Fpu \to \C$ with $\theta_p(\alpha_p^k)=i^k$  and $\conj{\theta_p} \colon\Fpu \to \C$ with $\conj{\theta_p}(\alpha_p^k)=\conj{i^k}=i^{-k}$.
Then one can check that for any $x \in \Fpu$, we have
\begin{align*}
F_p(x) & = \frac{1-i}{2} \theta_p(x)+\frac{1+i}{2} \conj{\theta_p}(x) \\
G_p(x) & = \frac{1+i}{2} \theta_p(x)+\frac{1-i}{2} \conj{\theta_p}(x),
\end{align*}
while we have decreed that $F_p(0)=G_p(0)=1$.  Although it is customary to decree that multiplicative characters like $\theta_p$ and $\conj{\theta_p}$ take $0$ to $0$, this would fail to produce a binary sequence.

Theorems \ref{Anne} and \ref{Rebecca} of this paper show that our sequences $\fpspu$ and $\gpspu$ behave similarly to shifted Legendre sequences $\hpspu$ in that their autocorrelation performance improves if they are cyclically shifted by approximately $1/4$ (or approximately $3/4$) of their length, as was observed in \cite{Golay-1983-Merit} and proved asymptotically in \cite{Hoholdt-Jensen-1988-Determination}.
Therefore, we use sequences $\fpnat=f_p^{(p-1)/4,p}$ and $\gpnat=g_p^{(p-1)/4,p}$ to obtain low autocorrelation similar to that of $\hpnat=h_p^{(p-1)/4,p}$.
The ``$\nat$'' superscript is to remind us that these sequences are of ``natural'' length: neither truncated nor appended.
There is an additional issue in the autocorrelation of $\fpnat$ and $\gpnat$ that does not appear for the shifted Legendre sequences $\hpnat$.
In Theorems \ref{Anne} and \ref{Christopher} we show that the asymptotic autocorrelation demerit factors for $\fpnat$ and $\gpnat$, and the asymptotic crosscorrelation demerit factor for the pair $(\fpnat,\gpnat)$ depend intimately on number-theoretic properties of the prime $p$, that is, the order of the finite field upon which these sequences are derived.

Every prime $p$ with $p\equiv 1 \pmod{4}$ can be expressed uniquely as $p=a^2+b^2$ where $a, b$ are positive integers with $a$ odd and $b$ even: equivalently we can say that $p$ factors uniquely in the ring $\Z[i]$ of Gaussian integers as $p=(a+b i)(a-b i)$ with $a+b i$ in the first quadrant of the complex plane.
This means that there is a unique $\gamma_p \in (0,\pi/2)$ such that $a=\sqrt{p}\cos\gamma_p$ and $b=\sqrt{p}\sin\gamma_p$.
Then we show that the autocorrelation demerit factors of both $\fpnat$ and $\gpnat$ are approximately $\frac{1}{2} -\frac{\cos(2\gamma_p)}{3}$ and the crosscorrelation demerit factor of the pair $(\fpnat,\gpnat)$ is approximately $\frac{2}{3}+\frac{\cos(2\gamma_p)}{3}$.

Lemma \ref{Theresa} below shows that the values of $\gamma_p$ are equidistributed in the interval $(0,\pi/2)$.  So for any $\gamma \in (0,\pi/2)$, there is an increasing sequence of primes $p$ such that $\gamma_p \to \gamma$ as $p\to\infty$, and then in this limit
\begin{align}
\DF(\fpnat) & \to \frac{1}{2}-\frac{\cos(2\gamma)}{3} \nonumber \\
\DF(\gpnat) & \to \frac{1}{2}-\frac{\cos(2\gamma)}{3} \label{Janice} \\
\CDF(\fpnat,\gpnat) & \to \frac{2}{3}+\frac{\cos(2\gamma)}{3}, \nonumber \\
\PSC(\fpnat,\gpnat) & \to \frac{7}{6}, \nonumber
\end{align}
where one should recall the definition of the Pursley-Sarwate Criterion, $\PSC(f,g)$ from \eqref{Penelope}, and remember that this criterion cannot be less than $1$, and will typically be about $2$ for a randomly selected sequence pair.
The limits in \eqref{Janice} are proved as specific cases of Theorems \ref{Anne} (using the parameters $\Lambda=1$ and $R=1/4$) and \ref{Christopher} (using the parameter $\Lambda=1$).

Note that there is a tradeoff between autocorrelation and crosscorrelation: the asymptotic Pursley-Sarwate Criterion is always $7/6$, so lowering crosscorrelation requires a concomitant rise in autocorrelation.  On one extreme, we can obtain asymptotic autocorrelation merit factor of $6$ (demerit factor $1/6$) and crosscorrelation merit factor of $1$ (demerit factor $1$).  And on the other extreme, we can obtain asymptotic autocorrelation merit factor $6/5$ (demerit factor $5/6$) and asymptotic crosscorrelation merit factor $3$ (demerit factor $1/3$).
In the middle, we can obtain sequence pairs with asymptotic autocorrelation merit factor $3$ and crosscorrelation merit factor $6/5$, equal to the best that can be achieved with m-sequences (see \cite[Section II.E]{Katz-2016-Aperiodic}), but the extreme of autocorrelation merit factor $6$ and crosscorrelation merit factor $1$ (and much of the range in between) is inaccessible to any previously known sequence pair construction.  Throughout the range, our asymptotic autocorrelation merit factors and crosscorrelation merit factor are always better than $1$ (except at the one extreme where the asymptotic crosscorrelation merit factor is $1$).  So these sequence pairs have superior correlation performance to pairs of randomly selected sequences (which have average autocorrelation and crosscorrelation demerit factors that tend to $1$ as length tends to infinity per Sarwate \cite[eqs.~(13),(38)]{Sarwate-1984-Mean}).

In Figure \ref{Aaron}, we show the dependence of autocorrelation and crosscorrelation demerit factors of $\fpnat$ and $\gpnat$ on $\cos(2\gamma_p)$.  The lines indicate the asymptotic values calculated in Theorems \ref{Anne} and \ref{Christopher}, while the data points show the actual values for the sequence pairs $(\fpnat,\gpnat)$ for all primes $p$ with $p\equiv 1 \pmod{4}$ and $p < 2000$.  To avoid clutter, we have only plotted the autocorrelation demerit factors of $\fpnat$: the values for $\gpnat$ are similar.
Note that the data points are close to the asymptotic values.
The few exceptions come from very short sequences.

\begin{center}
\begin{figure}
\begin{center}
\caption{Demerit factors of quartic residue sequences as a function of $\cos(2\gamma_p)$: autocorrelation (plusses, dashed line) and crosscorrelation (filled circles, solid line)}\label{Aaron}
\begin{tikzpicture}[gnuplot]
\path (0.000,0.000) rectangle (12.446,9.398);
\gpcolor{color=gp lt color border}
\gpsetlinetype{gp lt border}
\gpsetlinewidth{1.00}
\draw[gp path] (1.504,0.985)--(1.684,0.985);
\draw[gp path] (11.893,0.985)--(11.713,0.985);
\node[gp node right] at (1.320,0.985) { 0};
\draw[gp path] (1.504,2.326)--(1.684,2.326);
\draw[gp path] (11.893,2.326)--(11.713,2.326);
\node[gp node right] at (1.320,2.326) { 0.2};
\draw[gp path] (1.504,3.666)--(1.684,3.666);
\draw[gp path] (11.893,3.666)--(11.713,3.666);
\node[gp node right] at (1.320,3.666) { 0.4};
\draw[gp path] (1.504,5.007)--(1.684,5.007);
\draw[gp path] (11.893,5.007)--(11.713,5.007);
\node[gp node right] at (1.320,5.007) { 0.6};
\draw[gp path] (1.504,6.348)--(1.684,6.348);
\draw[gp path] (11.893,6.348)--(11.713,6.348);
\node[gp node right] at (1.320,6.348) { 0.8};
\draw[gp path] (1.504,7.688)--(1.684,7.688);
\draw[gp path] (11.893,7.688)--(11.713,7.688);
\node[gp node right] at (1.320,7.688) { 1};
\draw[gp path] (1.504,9.029)--(1.684,9.029);
\draw[gp path] (11.893,9.029)--(11.713,9.029);
\node[gp node right] at (1.320,9.029) { 1.2};
\draw[gp path] (1.504,0.985)--(1.504,1.165);
\draw[gp path] (1.504,9.029)--(1.504,8.849);
\node[gp node center] at (1.504,0.677) {-1};
\draw[gp path] (4.101,0.985)--(4.101,1.165);
\draw[gp path] (4.101,9.029)--(4.101,8.849);
\node[gp node center] at (4.101,0.677) {-0.5};
\draw[gp path] (6.699,0.985)--(6.699,1.165);
\draw[gp path] (6.699,9.029)--(6.699,8.849);
\node[gp node center] at (6.699,0.677) { 0};
\draw[gp path] (9.296,0.985)--(9.296,1.165);
\draw[gp path] (9.296,9.029)--(9.296,8.849);
\node[gp node center] at (9.296,0.677) { 0.5};
\draw[gp path] (11.893,0.985)--(11.893,1.165);
\draw[gp path] (11.893,9.029)--(11.893,8.849);
\node[gp node center] at (11.893,0.677) { 1};
\draw[gp path] (1.504,9.029)--(1.504,0.985)--(11.893,0.985)--(11.893,9.029)--cycle;
\node[gp node center,rotate=-270] at (0.246,5.007) {Demerit Factor};
\node[gp node center] at (6.698,0.215) {$\cos(2\gamma_p)$};
\gpcolor{rgb color={0.000,0.000,0.000}}
\gpsetlinetype{gp lt plot 2}
\gpsetlinewidth{2.00}
\draw[gp path] (1.504,6.571)--(1.609,6.526)--(1.714,6.481)--(1.819,6.436)--(1.924,6.391)%
  --(2.029,6.345)--(2.134,6.300)--(2.239,6.255)--(2.344,6.210)--(2.448,6.165)--(2.553,6.120)%
  --(2.658,6.075)--(2.763,6.029)--(2.868,5.984)--(2.973,5.939)--(3.078,5.894)--(3.183,5.849)%
  --(3.288,5.804)--(3.393,5.759)--(3.498,5.713)--(3.603,5.668)--(3.708,5.623)--(3.813,5.578)%
  --(3.918,5.533)--(4.023,5.488)--(4.127,5.443)--(4.232,5.397)--(4.337,5.352)--(4.442,5.307)%
  --(4.547,5.262)--(4.652,5.217)--(4.757,5.172)--(4.862,5.127)--(4.967,5.081)--(5.072,5.036)%
  --(5.177,4.991)--(5.282,4.946)--(5.387,4.901)--(5.492,4.856)--(5.597,4.811)--(5.702,4.765)%
  --(5.807,4.720)--(5.911,4.675)--(6.016,4.630)--(6.121,4.585)--(6.226,4.540)--(6.331,4.495)%
  --(6.436,4.450)--(6.541,4.404)--(6.646,4.359)--(6.751,4.314)--(6.856,4.269)--(6.961,4.224)%
  --(7.066,4.179)--(7.171,4.134)--(7.276,4.088)--(7.381,4.043)--(7.486,3.998)--(7.590,3.953)%
  --(7.695,3.908)--(7.800,3.863)--(7.905,3.818)--(8.010,3.772)--(8.115,3.727)--(8.220,3.682)%
  --(8.325,3.637)--(8.430,3.592)--(8.535,3.547)--(8.640,3.502)--(8.745,3.456)--(8.850,3.411)%
  --(8.955,3.366)--(9.060,3.321)--(9.165,3.276)--(9.270,3.231)--(9.374,3.186)--(9.479,3.140)%
  --(9.584,3.095)--(9.689,3.050)--(9.794,3.005)--(9.899,2.960)--(10.004,2.915)--(10.109,2.870)%
  --(10.214,2.824)--(10.319,2.779)--(10.424,2.734)--(10.529,2.689)--(10.634,2.644)--(10.739,2.599)%
  --(10.844,2.554)--(10.949,2.508)--(11.053,2.463)--(11.158,2.418)--(11.263,2.373)--(11.368,2.328)%
  --(11.473,2.283)--(11.578,2.238)--(11.683,2.193)--(11.788,2.147)--(11.893,2.102);
\gpsetlinetype{gp lt plot 0}
\draw[gp path] (1.504,3.219)--(1.609,3.265)--(1.714,3.310)--(1.819,3.355)--(1.924,3.400)%
  --(2.029,3.445)--(2.134,3.490)--(2.239,3.535)--(2.344,3.581)--(2.448,3.626)--(2.553,3.671)%
  --(2.658,3.716)--(2.763,3.761)--(2.868,3.806)--(2.973,3.851)--(3.078,3.897)--(3.183,3.942)%
  --(3.288,3.987)--(3.393,4.032)--(3.498,4.077)--(3.603,4.122)--(3.708,4.167)--(3.813,4.213)%
  --(3.918,4.258)--(4.023,4.303)--(4.127,4.348)--(4.232,4.393)--(4.337,4.438)--(4.442,4.483)%
  --(4.547,4.529)--(4.652,4.574)--(4.757,4.619)--(4.862,4.664)--(4.967,4.709)--(5.072,4.754)%
  --(5.177,4.799)--(5.282,4.844)--(5.387,4.890)--(5.492,4.935)--(5.597,4.980)--(5.702,5.025)%
  --(5.807,5.070)--(5.911,5.115)--(6.016,5.160)--(6.121,5.206)--(6.226,5.251)--(6.331,5.296)%
  --(6.436,5.341)--(6.541,5.386)--(6.646,5.431)--(6.751,5.476)--(6.856,5.522)--(6.961,5.567)%
  --(7.066,5.612)--(7.171,5.657)--(7.276,5.702)--(7.381,5.747)--(7.486,5.792)--(7.590,5.838)%
  --(7.695,5.883)--(7.800,5.928)--(7.905,5.973)--(8.010,6.018)--(8.115,6.063)--(8.220,6.108)%
  --(8.325,6.154)--(8.430,6.199)--(8.535,6.244)--(8.640,6.289)--(8.745,6.334)--(8.850,6.379)%
  --(8.955,6.424)--(9.060,6.470)--(9.165,6.515)--(9.270,6.560)--(9.374,6.605)--(9.479,6.650)%
  --(9.584,6.695)--(9.689,6.740)--(9.794,6.786)--(9.899,6.831)--(10.004,6.876)--(10.109,6.921)%
  --(10.214,6.966)--(10.319,7.011)--(10.424,7.056)--(10.529,7.102)--(10.634,7.147)--(10.739,7.192)%
  --(10.844,7.237)--(10.949,7.282)--(11.053,7.327)--(11.158,7.372)--(11.263,7.417)--(11.368,7.463)%
  --(11.473,7.508)--(11.578,7.553)--(11.683,7.598)--(11.788,7.643)--(11.893,7.688);
\gpsetlinewidth{1.00}
\gpsetpointsize{2.00}
\gppoint{gp mark 1}{(3.582,6.348)}
\gppoint{gp mark 1}{(8.696,3.682)}
\gppoint{gp mark 1}{(2.115,5.995)}
\gppoint{gp mark 1}{(10.460,3.185)}
\gppoint{gp mark 1}{(1.785,6.920)}
\gppoint{gp mark 1}{(7.839,4.143)}
\gppoint{gp mark 1}{(11.109,2.426)}
\gppoint{gp mark 1}{(5.762,5.273)}
\gppoint{gp mark 1}{(2.785,6.208)}
\gppoint{gp mark 1}{(4.422,5.392)}
\gppoint{gp mark 1}{(10.179,2.997)}
\gppoint{gp mark 1}{(1.607,6.655)}
\gppoint{gp mark 1}{(2.362,5.894)}
\gppoint{gp mark 1}{(6.009,4.580)}
\gppoint{gp mark 1}{(10.680,2.708)}
\gppoint{gp mark 1}{(4.921,4.940)}
\gppoint{gp mark 1}{(9.511,3.273)}
\gppoint{gp mark 1}{(11.653,2.231)}
\gppoint{gp mark 1}{(6.153,4.639)}
\gppoint{gp mark 1}{(4.142,5.681)}
\gppoint{gp mark 1}{(1.557,6.402)}
\gppoint{gp mark 1}{(11.712,2.243)}
\gppoint{gp mark 1}{(9.039,3.566)}
\gppoint{gp mark 1}{(11.203,2.346)}
\gppoint{gp mark 1}{(1.544,6.443)}
\gppoint{gp mark 1}{(8.031,4.041)}
\gppoint{gp mark 1}{(4.542,5.509)}
\gppoint{gp mark 1}{(2.428,6.089)}
\gppoint{gp mark 1}{(11.751,2.192)}
\gppoint{gp mark 1}{(7.113,4.194)}
\gppoint{gp mark 1}{(5.470,4.718)}
\gppoint{gp mark 1}{(4.001,5.729)}
\gppoint{gp mark 1}{(2.248,6.103)}
\gppoint{gp mark 1}{(10.009,2.904)}
\gppoint{gp mark 1}{(2.869,6.087)}
\gppoint{gp mark 1}{(9.222,3.152)}
\gppoint{gp mark 1}{(10.951,2.551)}
\gppoint{gp mark 1}{(1.530,6.920)}
\gppoint{gp mark 1}{(1.733,6.544)}
\gppoint{gp mark 1}{(7.056,4.192)}
\gppoint{gp mark 1}{(8.438,3.674)}
\gppoint{gp mark 1}{(2.638,5.945)}
\gppoint{gp mark 1}{(11.529,2.299)}
\gppoint{gp mark 1}{(9.639,3.054)}
\gppoint{gp mark 1}{(2.014,6.169)}
\gppoint{gp mark 1}{(3.917,5.694)}
\gppoint{gp mark 1}{(9.973,2.834)}
\gppoint{gp mark 1}{(8.237,3.730)}
\gppoint{gp mark 1}{(4.590,5.405)}
\gppoint{gp mark 1}{(1.522,6.698)}
\gppoint{gp mark 1}{(10.772,2.663)}
\gppoint{gp mark 1}{(1.936,6.543)}
\gppoint{gp mark 1}{(6.402,4.221)}
\gppoint{gp mark 1}{(7.582,4.094)}
\gppoint{gp mark 1}{(11.634,2.240)}
\gppoint{gp mark 1}{(4.193,5.496)}
\gppoint{gp mark 1}{(11.327,2.384)}
\gppoint{gp mark 1}{(9.670,3.148)}
\gppoint{gp mark 1}{(1.519,6.622)}
\gppoint{gp mark 1}{(1.875,6.615)}
\gppoint{gp mark 1}{(4.801,5.065)}
\gppoint{gp mark 1}{(11.836,2.129)}
\gppoint{gp mark 1}{(2.616,5.978)}
\gppoint{gp mark 1}{(6.432,4.382)}
\gppoint{gp mark 1}{(9.948,2.973)}
\gppoint{gp mark 1}{(5.388,4.888)}
\gppoint{gp mark 1}{(3.081,5.889)}
\gppoint{gp mark 1}{(1.825,6.555)}
\gppoint{gp mark 1}{(9.413,3.122)}
\gppoint{gp mark 1}{(10.640,2.678)}
\gppoint{gp mark 1}{(7.947,3.748)}
\gppoint{gp mark 1}{(11.699,2.204)}
\gppoint{gp mark 1}{(11.467,2.276)}
\gppoint{gp mark 1}{(8.874,3.433)}
\gppoint{gp mark 1}{(7.420,4.118)}
\gppoint{gp mark 1}{(5.507,4.924)}
\gppoint{gp mark 1}{(10.789,2.555)}
\gppoint{gp mark 1}{(3.346,5.679)}
\gppoint{gp mark 1}{(11.723,2.201)}
\gppoint{gp mark 1}{(11.518,2.261)}
\gppoint{gp mark 1}{(3.821,5.633)}
\gppoint{gp mark 1}{(6.929,4.276)}
\gppoint{gp mark 1}{(2.735,5.988)}
\gppoint{gp mark 1}{(1.595,6.585)}
\gppoint{gp mark 1}{(1.752,6.433)}
\gppoint{gp mark 1}{(10.914,2.493)}
\gppoint{gp mark 1}{(3.146,5.894)}
\gppoint{gp mark 1}{(11.855,2.102)}
\gppoint{gp mark 1}{(9.469,3.249)}
\gppoint{gp mark 1}{(7.359,4.061)}
\gppoint{gp mark 1}{(5.606,4.830)}
\gppoint{gp mark 1}{(8.212,3.664)}
\gppoint{gp mark 1}{(11.316,2.325)}
\gppoint{gp mark 1}{(1.724,6.490)}
\gppoint{gp mark 1}{(2.976,6.036)}
\gppoint{gp mark 1}{(6.910,4.209)}
\gppoint{gp mark 1}{(7.748,3.894)}
\gppoint{gp mark 1}{(9.708,3.099)}
\gppoint{gp mark 1}{(11.859,2.119)}
\gppoint{gp mark 1}{(2.184,6.209)}
\gppoint{gp mark 1}{(3.376,5.836)}
\gppoint{gp mark 1}{(2.488,6.227)}
\gppoint{gp mark 1}{(11.377,2.336)}
\gppoint{gp mark 1}{(1.512,6.551)}
\gppoint{gp mark 1}{(6.495,4.359)}
\gppoint{gp mark 1}{(1.701,6.509)}
\gppoint{gp mark 1}{(8.840,3.435)}
\gppoint{gp mark 1}{(11.863,2.126)}
\gppoint{gp mark 1}{(3.197,5.876)}
\gppoint{gp mark 1}{(6.112,4.636)}
\gppoint{gp mark 1}{(5.350,4.951)}
\gppoint{gp mark 1}{(11.429,2.364)}
\gppoint{gp mark 1}{(1.568,6.632)}
\gppoint{gp mark 1}{(10.097,2.884)}
\gppoint{gp mark 1}{(9.102,3.319)}
\gppoint{gp mark 1}{(1.845,6.533)}
\gppoint{gp mark 1}{(9.720,3.035)}
\gppoint{gp mark 1}{(5.043,5.056)}
\gppoint{gp mark 1}{(4.373,5.108)}
\gppoint{gp mark 1}{(1.510,6.533)}
\gppoint{gp mark 1}{(1.562,6.399)}
\gppoint{gp mark 1}{(2.592,6.058)}
\gppoint{gp mark 1}{(11.252,2.362)}
\gppoint{gp mark 1}{(7.603,3.967)}
\gppoint{gp mark 1}{(3.767,5.622)}
\gppoint{gp mark 1}{(2.905,6.070)}
\gppoint{gp mark 1}{(9.905,2.908)}
\gppoint{gp mark 1}{(11.795,2.163)}
\gppoint{gp mark 1}{(8.951,3.402)}
\gppoint{gp mark 1}{(2.234,6.258)}
\gppoint{gp mark 1}{(3.237,5.814)}
\gppoint{gp mark 1}{(6.522,4.448)}
\gppoint{gp mark 1}{(5.824,4.814)}
\gppoint{gp mark 1}{(10.396,2.725)}
\gppoint{gp mark 1}{(1.649,6.558)}
\gppoint{gp mark 1}{(8.570,3.506)}
\gppoint{gp mark 1}{(6.869,4.260)}
\gppoint{gp mark 1}{(7.544,3.977)}
\gppoint{gp mark 1}{(10.808,2.598)}
\gppoint{gp mark 1}{(3.093,5.981)}
\gppoint{gp mark 1}{(8.199,3.710)}
\gppoint{gp mark 1}{(11.545,2.280)}
\gppoint{gp mark 1}{(2.412,6.125)}
\gppoint{gp mark 1}{(11.360,2.333)}
\gppoint{gp mark 1}{(4.289,5.292)}
\gppoint{gp mark 1}{(11.142,2.414)}
\gppoint{gp mark 1}{(5.879,4.665)}
\gppoint{gp mark 7}{(3.582,3.398)}
\gppoint{gp mark 7}{(8.696,7.212)}
\gppoint{gp mark 7}{(2.115,4.627)}
\gppoint{gp mark 7}{(10.460,6.796)}
\gppoint{gp mark 7}{(1.785,3.536)}
\gppoint{gp mark 7}{(7.839,5.886)}
\gppoint{gp mark 7}{(11.109,7.641)}
\gppoint{gp mark 7}{(5.762,4.964)}
\gppoint{gp mark 7}{(2.785,3.824)}
\gppoint{gp mark 7}{(4.422,4.906)}
\gppoint{gp mark 7}{(10.179,7.201)}
\gppoint{gp mark 7}{(1.607,3.333)}
\gppoint{gp mark 7}{(2.362,3.969)}
\gppoint{gp mark 7}{(6.009,5.387)}
\gppoint{gp mark 7}{(10.680,7.407)}
\gppoint{gp mark 7}{(4.921,4.892)}
\gppoint{gp mark 7}{(9.511,6.818)}
\gppoint{gp mark 7}{(11.653,7.527)}
\gppoint{gp mark 7}{(6.153,5.336)}
\gppoint{gp mark 7}{(4.142,4.261)}
\gppoint{gp mark 7}{(1.557,3.413)}
\gppoint{gp mark 7}{(11.712,7.665)}
\gppoint{gp mark 7}{(9.039,6.504)}
\gppoint{gp mark 7}{(11.203,7.269)}
\gppoint{gp mark 7}{(1.544,3.439)}
\gppoint{gp mark 7}{(8.031,5.909)}
\gppoint{gp mark 7}{(4.542,4.408)}
\gppoint{gp mark 7}{(2.428,3.624)}
\gppoint{gp mark 7}{(11.751,7.484)}
\gppoint{gp mark 7}{(7.113,5.788)}
\gppoint{gp mark 7}{(5.470,5.070)}
\gppoint{gp mark 7}{(4.001,4.208)}
\gppoint{gp mark 7}{(2.248,3.547)}
\gppoint{gp mark 7}{(10.009,7.036)}
\gppoint{gp mark 7}{(2.869,3.775)}
\gppoint{gp mark 7}{(9.222,6.558)}
\gppoint{gp mark 7}{(10.951,7.309)}
\gppoint{gp mark 7}{(1.530,3.026)}
\gppoint{gp mark 7}{(1.733,3.252)}
\gppoint{gp mark 7}{(7.056,5.631)}
\gppoint{gp mark 7}{(8.438,6.293)}
\gppoint{gp mark 7}{(2.638,3.871)}
\gppoint{gp mark 7}{(11.529,7.574)}
\gppoint{gp mark 7}{(9.639,6.733)}
\gppoint{gp mark 7}{(2.014,3.613)}
\gppoint{gp mark 7}{(3.917,4.184)}
\gppoint{gp mark 7}{(9.973,6.859)}
\gppoint{gp mark 7}{(8.237,6.157)}
\gppoint{gp mark 7}{(4.590,4.439)}
\gppoint{gp mark 7}{(1.522,3.248)}
\gppoint{gp mark 7}{(10.772,7.096)}
\gppoint{gp mark 7}{(1.936,3.460)}
\gppoint{gp mark 7}{(6.402,5.316)}
\gppoint{gp mark 7}{(7.582,5.697)}
\gppoint{gp mark 7}{(11.634,7.577)}
\gppoint{gp mark 7}{(4.193,4.298)}
\gppoint{gp mark 7}{(11.327,7.428)}
\gppoint{gp mark 7}{(9.670,6.690)}
\gppoint{gp mark 7}{(1.519,3.260)}
\gppoint{gp mark 7}{(1.875,3.281)}
\gppoint{gp mark 7}{(4.801,4.666)}
\gppoint{gp mark 7}{(11.836,7.732)}
\gppoint{gp mark 7}{(2.616,3.796)}
\gppoint{gp mark 7}{(6.432,5.470)}
\gppoint{gp mark 7}{(9.948,6.904)}
\gppoint{gp mark 7}{(5.388,4.893)}
\gppoint{gp mark 7}{(3.081,3.968)}
\gppoint{gp mark 7}{(1.825,3.327)}
\gppoint{gp mark 7}{(9.413,6.629)}
\gppoint{gp mark 7}{(10.640,7.176)}
\gppoint{gp mark 7}{(7.947,6.072)}
\gppoint{gp mark 7}{(11.699,7.630)}
\gppoint{gp mark 7}{(11.467,7.468)}
\gppoint{gp mark 7}{(8.874,6.500)}
\gppoint{gp mark 7}{(7.420,5.734)}
\gppoint{gp mark 7}{(5.507,4.885)}
\gppoint{gp mark 7}{(10.789,7.254)}
\gppoint{gp mark 7}{(3.346,4.149)}
\gppoint{gp mark 7}{(11.723,7.596)}
\gppoint{gp mark 7}{(11.518,7.583)}
\gppoint{gp mark 7}{(3.821,4.210)}
\gppoint{gp mark 7}{(6.929,5.570)}
\gppoint{gp mark 7}{(2.735,3.764)}
\gppoint{gp mark 7}{(1.595,3.322)}
\gppoint{gp mark 7}{(1.752,3.397)}
\gppoint{gp mark 7}{(10.914,7.309)}
\gppoint{gp mark 7}{(3.146,3.948)}
\gppoint{gp mark 7}{(11.855,7.643)}
\gppoint{gp mark 7}{(9.469,6.663)}
\gppoint{gp mark 7}{(7.359,5.715)}
\gppoint{gp mark 7}{(5.606,4.982)}
\gppoint{gp mark 7}{(8.212,6.128)}
\gppoint{gp mark 7}{(11.316,7.454)}
\gppoint{gp mark 7}{(1.724,3.353)}
\gppoint{gp mark 7}{(2.976,3.801)}
\gppoint{gp mark 7}{(6.910,5.562)}
\gppoint{gp mark 7}{(7.748,5.926)}
\gppoint{gp mark 7}{(9.708,6.642)}
\gppoint{gp mark 7}{(11.859,7.745)}
\gppoint{gp mark 7}{(2.184,3.471)}
\gppoint{gp mark 7}{(3.376,4.019)}
\gppoint{gp mark 7}{(2.488,3.603)}
\gppoint{gp mark 7}{(11.377,7.384)}
\gppoint{gp mark 7}{(1.512,3.251)}
\gppoint{gp mark 7}{(6.495,5.443)}
\gppoint{gp mark 7}{(1.701,3.355)}
\gppoint{gp mark 7}{(8.840,6.428)}
\gppoint{gp mark 7}{(11.863,7.652)}
\gppoint{gp mark 7}{(3.197,3.956)}
\gppoint{gp mark 7}{(6.112,5.240)}
\gppoint{gp mark 7}{(5.350,4.939)}
\gppoint{gp mark 7}{(11.429,7.557)}
\gppoint{gp mark 7}{(1.568,3.307)}
\gppoint{gp mark 7}{(10.097,6.834)}
\gppoint{gp mark 7}{(9.102,6.475)}
\gppoint{gp mark 7}{(1.845,3.278)}
\gppoint{gp mark 7}{(9.720,6.760)}
\gppoint{gp mark 7}{(5.043,4.680)}
\gppoint{gp mark 7}{(4.373,4.520)}
\gppoint{gp mark 7}{(1.510,3.346)}
\gppoint{gp mark 7}{(1.562,3.299)}
\gppoint{gp mark 7}{(2.592,3.695)}
\gppoint{gp mark 7}{(11.252,7.445)}
\gppoint{gp mark 7}{(7.603,5.876)}
\gppoint{gp mark 7}{(3.767,4.170)}
\gppoint{gp mark 7}{(2.905,3.812)}
\gppoint{gp mark 7}{(9.905,6.849)}
\gppoint{gp mark 7}{(11.795,7.647)}
\gppoint{gp mark 7}{(8.951,6.443)}
\gppoint{gp mark 7}{(2.234,3.610)}
\gppoint{gp mark 7}{(3.237,3.934)}
\gppoint{gp mark 7}{(6.522,5.366)}
\gppoint{gp mark 7}{(5.824,5.035)}
\gppoint{gp mark 7}{(10.396,7.044)}
\gppoint{gp mark 7}{(1.649,3.261)}
\gppoint{gp mark 7}{(8.570,6.269)}
\gppoint{gp mark 7}{(6.869,5.561)}
\gppoint{gp mark 7}{(7.544,5.885)}
\gppoint{gp mark 7}{(10.808,7.191)}
\gppoint{gp mark 7}{(3.093,3.849)}
\gppoint{gp mark 7}{(8.199,6.167)}
\gppoint{gp mark 7}{(11.545,7.481)}
\gppoint{gp mark 7}{(2.412,3.638)}
\gppoint{gp mark 7}{(11.360,7.468)}
\gppoint{gp mark 7}{(4.289,4.545)}
\gppoint{gp mark 7}{(11.142,7.349)}
\gppoint{gp mark 7}{(5.879,5.090)}
\gpcolor{color=gp lt color border}
\gpsetlinetype{gp lt border}
\draw[gp path] (1.504,9.029)--(1.504,0.985)--(11.893,0.985)--(11.893,9.029)--cycle;
\gpdefrectangularnode{gp plot 1}{\pgfpoint{1.504cm}{0.985cm}}{\pgfpoint{11.893cm}{9.029cm}}
\end{tikzpicture}
\end{center}
\end{figure}
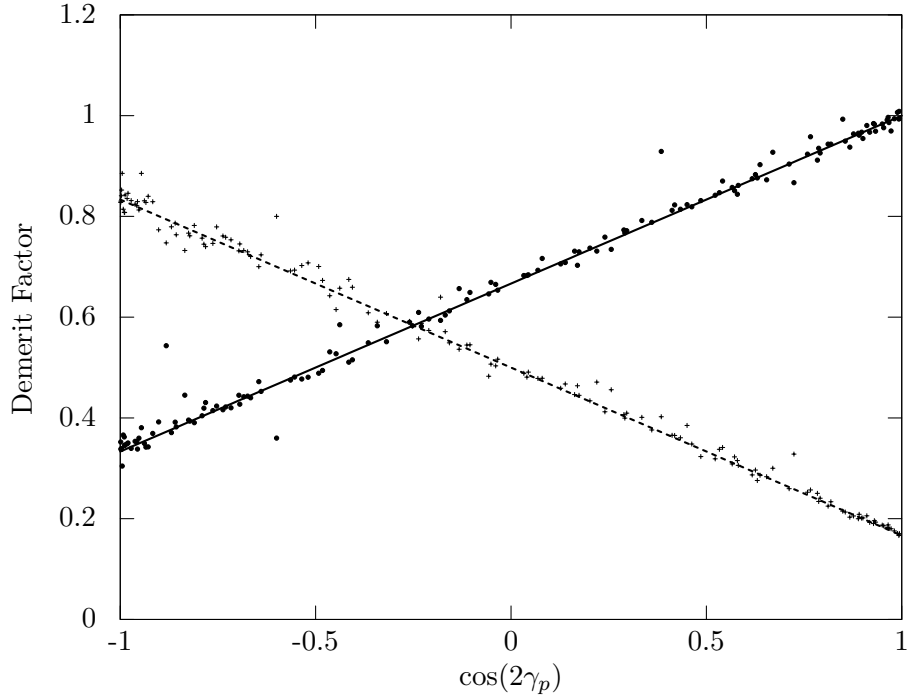
\end{center}

Interestingly, if we crosscorrelate our sequences $\fpnat$ and $\gpnat$ derived from quartic characters with shifted Legendre sequences $\hpnat$, the performance does not depend appreciably on the prime $p$.  This is also in accord with our proof in Theorem \ref{Clarence} (setting $\Lambda=1$) that the asymptotic crosscorrelation demerit factor for $(\fpnat,\hpnat)$ or $(\gpnat,\hpnat)$ should always tend to $1$ as $p\to\infty$.
In Figure \ref{Edward}, we show the dependence of autocorrelation and crosscorrelation demerit factors of $\fpnat$ and $\hpnat$ on $\cos(2\gamma_p)$.  The lines indicate the asymptotic values calculated in Theorems \ref{Anne}, \ref{Rebecca}, and \ref{Clarence}.
The data points show the actual values for the sequence pairs $(\fpnat,\hpnat)$ for all primes $p$ with $p\equiv 1 \pmod{4}$ and $p < 2000$.  Note that the autocorrelation demerit factor for $\hpnat$ shows little dependence on $p$ (it has asymptotic value of $1/6$ regardless of $\gamma_p$), while the autocorrelation demerit factor of $\fpnat$ varies considerably with $\cos(2\gamma_p)$ as already seen in Figure \ref{Aaron}.  The crosscorrelation demerit factor for $(\fpnat,\hpnat)$ is always around $1$, the asymptotic value proved in Theorem \ref{Clarence}.
One would get a similar plot if $\fpnat$ were replaced with $\gpnat$.
Again, note that the data points are close to the asymptotic values.

\begin{center}
\begin{figure}
\begin{center}
\caption{Demerit factors as a function of $\cos(2\gamma_p)$: autocorrelation of quartic residue sequence (plusses, dashed line); autocorrelation of Legendre sequence (diamonds, dot-dash line); crosscorrelation of these paired (filled circles, solid line)}\label{Edward}
\begin{tikzpicture}[gnuplot]
\path (0.000,0.000) rectangle (12.446,9.398);
\gpcolor{color=gp lt color border}
\gpsetlinetype{gp lt border}
\gpsetlinewidth{1.00}
\draw[gp path] (1.504,0.985)--(1.684,0.985);
\draw[gp path] (11.893,0.985)--(11.713,0.985);
\node[gp node right] at (1.320,0.985) { 0};
\draw[gp path] (1.504,2.326)--(1.684,2.326);
\draw[gp path] (11.893,2.326)--(11.713,2.326);
\node[gp node right] at (1.320,2.326) { 0.2};
\draw[gp path] (1.504,3.666)--(1.684,3.666);
\draw[gp path] (11.893,3.666)--(11.713,3.666);
\node[gp node right] at (1.320,3.666) { 0.4};
\draw[gp path] (1.504,5.007)--(1.684,5.007);
\draw[gp path] (11.893,5.007)--(11.713,5.007);
\node[gp node right] at (1.320,5.007) { 0.6};
\draw[gp path] (1.504,6.348)--(1.684,6.348);
\draw[gp path] (11.893,6.348)--(11.713,6.348);
\node[gp node right] at (1.320,6.348) { 0.8};
\draw[gp path] (1.504,7.688)--(1.684,7.688);
\draw[gp path] (11.893,7.688)--(11.713,7.688);
\node[gp node right] at (1.320,7.688) { 1};
\draw[gp path] (1.504,9.029)--(1.684,9.029);
\draw[gp path] (11.893,9.029)--(11.713,9.029);
\node[gp node right] at (1.320,9.029) { 1.2};
\draw[gp path] (1.504,0.985)--(1.504,1.165);
\draw[gp path] (1.504,9.029)--(1.504,8.849);
\node[gp node center] at (1.504,0.677) {-1};
\draw[gp path] (4.101,0.985)--(4.101,1.165);
\draw[gp path] (4.101,9.029)--(4.101,8.849);
\node[gp node center] at (4.101,0.677) {-0.5};
\draw[gp path] (6.699,0.985)--(6.699,1.165);
\draw[gp path] (6.699,9.029)--(6.699,8.849);
\node[gp node center] at (6.699,0.677) { 0};
\draw[gp path] (9.296,0.985)--(9.296,1.165);
\draw[gp path] (9.296,9.029)--(9.296,8.849);
\node[gp node center] at (9.296,0.677) { 0.5};
\draw[gp path] (11.893,0.985)--(11.893,1.165);
\draw[gp path] (11.893,9.029)--(11.893,8.849);
\node[gp node center] at (11.893,0.677) { 1};
\draw[gp path] (1.504,9.029)--(1.504,0.985)--(11.893,0.985)--(11.893,9.029)--cycle;
\node[gp node center,rotate=-270] at (0.246,5.007) {Demerit Factor};
\node[gp node center] at (6.698,0.215) {$\cos(2\gamma_p)$};
\gpcolor{rgb color={0.000,0.000,0.000}}
\gpsetlinetype{gp lt plot 2}
\gpsetlinewidth{2.00}
\draw[gp path] (1.504,6.571)--(1.609,6.526)--(1.714,6.481)--(1.819,6.436)--(1.924,6.391)%
  --(2.029,6.345)--(2.134,6.300)--(2.239,6.255)--(2.344,6.210)--(2.448,6.165)--(2.553,6.120)%
  --(2.658,6.075)--(2.763,6.029)--(2.868,5.984)--(2.973,5.939)--(3.078,5.894)--(3.183,5.849)%
  --(3.288,5.804)--(3.393,5.759)--(3.498,5.713)--(3.603,5.668)--(3.708,5.623)--(3.813,5.578)%
  --(3.918,5.533)--(4.023,5.488)--(4.127,5.443)--(4.232,5.397)--(4.337,5.352)--(4.442,5.307)%
  --(4.547,5.262)--(4.652,5.217)--(4.757,5.172)--(4.862,5.127)--(4.967,5.081)--(5.072,5.036)%
  --(5.177,4.991)--(5.282,4.946)--(5.387,4.901)--(5.492,4.856)--(5.597,4.811)--(5.702,4.765)%
  --(5.807,4.720)--(5.911,4.675)--(6.016,4.630)--(6.121,4.585)--(6.226,4.540)--(6.331,4.495)%
  --(6.436,4.450)--(6.541,4.404)--(6.646,4.359)--(6.751,4.314)--(6.856,4.269)--(6.961,4.224)%
  --(7.066,4.179)--(7.171,4.134)--(7.276,4.088)--(7.381,4.043)--(7.486,3.998)--(7.590,3.953)%
  --(7.695,3.908)--(7.800,3.863)--(7.905,3.818)--(8.010,3.772)--(8.115,3.727)--(8.220,3.682)%
  --(8.325,3.637)--(8.430,3.592)--(8.535,3.547)--(8.640,3.502)--(8.745,3.456)--(8.850,3.411)%
  --(8.955,3.366)--(9.060,3.321)--(9.165,3.276)--(9.270,3.231)--(9.374,3.186)--(9.479,3.140)%
  --(9.584,3.095)--(9.689,3.050)--(9.794,3.005)--(9.899,2.960)--(10.004,2.915)--(10.109,2.870)%
  --(10.214,2.824)--(10.319,2.779)--(10.424,2.734)--(10.529,2.689)--(10.634,2.644)--(10.739,2.599)%
  --(10.844,2.554)--(10.949,2.508)--(11.053,2.463)--(11.158,2.418)--(11.263,2.373)--(11.368,2.328)%
  --(11.473,2.283)--(11.578,2.238)--(11.683,2.193)--(11.788,2.147)--(11.893,2.102);
\gpsetlinetype{gp lt plot 4}
\draw[gp path] (1.504,2.102)--(1.609,2.102)--(1.714,2.102)--(1.819,2.102)--(1.924,2.102)%
  --(2.029,2.102)--(2.134,2.102)--(2.239,2.102)--(2.344,2.102)--(2.448,2.102)--(2.553,2.102)%
  --(2.658,2.102)--(2.763,2.102)--(2.868,2.102)--(2.973,2.102)--(3.078,2.102)--(3.183,2.102)%
  --(3.288,2.102)--(3.393,2.102)--(3.498,2.102)--(3.603,2.102)--(3.708,2.102)--(3.813,2.102)%
  --(3.918,2.102)--(4.023,2.102)--(4.127,2.102)--(4.232,2.102)--(4.337,2.102)--(4.442,2.102)%
  --(4.547,2.102)--(4.652,2.102)--(4.757,2.102)--(4.862,2.102)--(4.967,2.102)--(5.072,2.102)%
  --(5.177,2.102)--(5.282,2.102)--(5.387,2.102)--(5.492,2.102)--(5.597,2.102)--(5.702,2.102)%
  --(5.807,2.102)--(5.911,2.102)--(6.016,2.102)--(6.121,2.102)--(6.226,2.102)--(6.331,2.102)%
  --(6.436,2.102)--(6.541,2.102)--(6.646,2.102)--(6.751,2.102)--(6.856,2.102)--(6.961,2.102)%
  --(7.066,2.102)--(7.171,2.102)--(7.276,2.102)--(7.381,2.102)--(7.486,2.102)--(7.590,2.102)%
  --(7.695,2.102)--(7.800,2.102)--(7.905,2.102)--(8.010,2.102)--(8.115,2.102)--(8.220,2.102)%
  --(8.325,2.102)--(8.430,2.102)--(8.535,2.102)--(8.640,2.102)--(8.745,2.102)--(8.850,2.102)%
  --(8.955,2.102)--(9.060,2.102)--(9.165,2.102)--(9.270,2.102)--(9.374,2.102)--(9.479,2.102)%
  --(9.584,2.102)--(9.689,2.102)--(9.794,2.102)--(9.899,2.102)--(10.004,2.102)--(10.109,2.102)%
  --(10.214,2.102)--(10.319,2.102)--(10.424,2.102)--(10.529,2.102)--(10.634,2.102)--(10.739,2.102)%
  --(10.844,2.102)--(10.949,2.102)--(11.053,2.102)--(11.158,2.102)--(11.263,2.102)--(11.368,2.102)%
  --(11.473,2.102)--(11.578,2.102)--(11.683,2.102)--(11.788,2.102)--(11.893,2.102);
\gpsetlinetype{gp lt plot 0}
\draw[gp path] (1.504,7.688)--(1.609,7.688)--(1.714,7.688)--(1.819,7.688)--(1.924,7.688)%
  --(2.029,7.688)--(2.134,7.688)--(2.239,7.688)--(2.344,7.688)--(2.448,7.688)--(2.553,7.688)%
  --(2.658,7.688)--(2.763,7.688)--(2.868,7.688)--(2.973,7.688)--(3.078,7.688)--(3.183,7.688)%
  --(3.288,7.688)--(3.393,7.688)--(3.498,7.688)--(3.603,7.688)--(3.708,7.688)--(3.813,7.688)%
  --(3.918,7.688)--(4.023,7.688)--(4.127,7.688)--(4.232,7.688)--(4.337,7.688)--(4.442,7.688)%
  --(4.547,7.688)--(4.652,7.688)--(4.757,7.688)--(4.862,7.688)--(4.967,7.688)--(5.072,7.688)%
  --(5.177,7.688)--(5.282,7.688)--(5.387,7.688)--(5.492,7.688)--(5.597,7.688)--(5.702,7.688)%
  --(5.807,7.688)--(5.911,7.688)--(6.016,7.688)--(6.121,7.688)--(6.226,7.688)--(6.331,7.688)%
  --(6.436,7.688)--(6.541,7.688)--(6.646,7.688)--(6.751,7.688)--(6.856,7.688)--(6.961,7.688)%
  --(7.066,7.688)--(7.171,7.688)--(7.276,7.688)--(7.381,7.688)--(7.486,7.688)--(7.590,7.688)%
  --(7.695,7.688)--(7.800,7.688)--(7.905,7.688)--(8.010,7.688)--(8.115,7.688)--(8.220,7.688)%
  --(8.325,7.688)--(8.430,7.688)--(8.535,7.688)--(8.640,7.688)--(8.745,7.688)--(8.850,7.688)%
  --(8.955,7.688)--(9.060,7.688)--(9.165,7.688)--(9.270,7.688)--(9.374,7.688)--(9.479,7.688)%
  --(9.584,7.688)--(9.689,7.688)--(9.794,7.688)--(9.899,7.688)--(10.004,7.688)--(10.109,7.688)%
  --(10.214,7.688)--(10.319,7.688)--(10.424,7.688)--(10.529,7.688)--(10.634,7.688)--(10.739,7.688)%
  --(10.844,7.688)--(10.949,7.688)--(11.053,7.688)--(11.158,7.688)--(11.263,7.688)--(11.368,7.688)%
  --(11.473,7.688)--(11.578,7.688)--(11.683,7.688)--(11.788,7.688)--(11.893,7.688);
\gpsetlinewidth{1.00}
\gpsetpointsize{2.00}
\gppoint{gp mark 1}{(3.582,6.348)}
\gppoint{gp mark 1}{(8.696,3.682)}
\gppoint{gp mark 1}{(2.115,5.995)}
\gppoint{gp mark 1}{(10.460,3.185)}
\gppoint{gp mark 1}{(1.785,6.920)}
\gppoint{gp mark 1}{(7.839,4.143)}
\gppoint{gp mark 1}{(11.109,2.426)}
\gppoint{gp mark 1}{(5.762,5.273)}
\gppoint{gp mark 1}{(2.785,6.208)}
\gppoint{gp mark 1}{(4.422,5.392)}
\gppoint{gp mark 1}{(10.179,2.997)}
\gppoint{gp mark 1}{(1.607,6.655)}
\gppoint{gp mark 1}{(2.362,5.894)}
\gppoint{gp mark 1}{(6.009,4.580)}
\gppoint{gp mark 1}{(10.680,2.708)}
\gppoint{gp mark 1}{(4.921,4.940)}
\gppoint{gp mark 1}{(9.511,3.273)}
\gppoint{gp mark 1}{(11.653,2.231)}
\gppoint{gp mark 1}{(6.153,4.639)}
\gppoint{gp mark 1}{(4.142,5.681)}
\gppoint{gp mark 1}{(1.557,6.402)}
\gppoint{gp mark 1}{(11.712,2.243)}
\gppoint{gp mark 1}{(9.039,3.566)}
\gppoint{gp mark 1}{(11.203,2.346)}
\gppoint{gp mark 1}{(1.544,6.443)}
\gppoint{gp mark 1}{(8.031,4.041)}
\gppoint{gp mark 1}{(4.542,5.509)}
\gppoint{gp mark 1}{(2.428,6.089)}
\gppoint{gp mark 1}{(11.751,2.192)}
\gppoint{gp mark 1}{(7.113,4.194)}
\gppoint{gp mark 1}{(5.470,4.718)}
\gppoint{gp mark 1}{(4.001,5.729)}
\gppoint{gp mark 1}{(2.248,6.103)}
\gppoint{gp mark 1}{(10.009,2.904)}
\gppoint{gp mark 1}{(2.869,6.087)}
\gppoint{gp mark 1}{(9.222,3.152)}
\gppoint{gp mark 1}{(10.951,2.551)}
\gppoint{gp mark 1}{(1.530,6.920)}
\gppoint{gp mark 1}{(1.733,6.544)}
\gppoint{gp mark 1}{(7.056,4.192)}
\gppoint{gp mark 1}{(8.438,3.674)}
\gppoint{gp mark 1}{(2.638,5.945)}
\gppoint{gp mark 1}{(11.529,2.299)}
\gppoint{gp mark 1}{(9.639,3.054)}
\gppoint{gp mark 1}{(2.014,6.169)}
\gppoint{gp mark 1}{(3.917,5.694)}
\gppoint{gp mark 1}{(9.973,2.834)}
\gppoint{gp mark 1}{(8.237,3.730)}
\gppoint{gp mark 1}{(4.590,5.405)}
\gppoint{gp mark 1}{(1.522,6.698)}
\gppoint{gp mark 1}{(10.772,2.663)}
\gppoint{gp mark 1}{(1.936,6.543)}
\gppoint{gp mark 1}{(6.402,4.221)}
\gppoint{gp mark 1}{(7.582,4.094)}
\gppoint{gp mark 1}{(11.634,2.240)}
\gppoint{gp mark 1}{(4.193,5.496)}
\gppoint{gp mark 1}{(11.327,2.384)}
\gppoint{gp mark 1}{(9.670,3.148)}
\gppoint{gp mark 1}{(1.519,6.622)}
\gppoint{gp mark 1}{(1.875,6.615)}
\gppoint{gp mark 1}{(4.801,5.065)}
\gppoint{gp mark 1}{(11.836,2.129)}
\gppoint{gp mark 1}{(2.616,5.978)}
\gppoint{gp mark 1}{(6.432,4.382)}
\gppoint{gp mark 1}{(9.948,2.973)}
\gppoint{gp mark 1}{(5.388,4.888)}
\gppoint{gp mark 1}{(3.081,5.889)}
\gppoint{gp mark 1}{(1.825,6.555)}
\gppoint{gp mark 1}{(9.413,3.122)}
\gppoint{gp mark 1}{(10.640,2.678)}
\gppoint{gp mark 1}{(7.947,3.748)}
\gppoint{gp mark 1}{(11.699,2.204)}
\gppoint{gp mark 1}{(11.467,2.276)}
\gppoint{gp mark 1}{(8.874,3.433)}
\gppoint{gp mark 1}{(7.420,4.118)}
\gppoint{gp mark 1}{(5.507,4.924)}
\gppoint{gp mark 1}{(10.789,2.555)}
\gppoint{gp mark 1}{(3.346,5.679)}
\gppoint{gp mark 1}{(11.723,2.201)}
\gppoint{gp mark 1}{(11.518,2.261)}
\gppoint{gp mark 1}{(3.821,5.633)}
\gppoint{gp mark 1}{(6.929,4.276)}
\gppoint{gp mark 1}{(2.735,5.988)}
\gppoint{gp mark 1}{(1.595,6.585)}
\gppoint{gp mark 1}{(1.752,6.433)}
\gppoint{gp mark 1}{(10.914,2.493)}
\gppoint{gp mark 1}{(3.146,5.894)}
\gppoint{gp mark 1}{(11.855,2.102)}
\gppoint{gp mark 1}{(9.469,3.249)}
\gppoint{gp mark 1}{(7.359,4.061)}
\gppoint{gp mark 1}{(5.606,4.830)}
\gppoint{gp mark 1}{(8.212,3.664)}
\gppoint{gp mark 1}{(11.316,2.325)}
\gppoint{gp mark 1}{(1.724,6.490)}
\gppoint{gp mark 1}{(2.976,6.036)}
\gppoint{gp mark 1}{(6.910,4.209)}
\gppoint{gp mark 1}{(7.748,3.894)}
\gppoint{gp mark 1}{(9.708,3.099)}
\gppoint{gp mark 1}{(11.859,2.119)}
\gppoint{gp mark 1}{(2.184,6.209)}
\gppoint{gp mark 1}{(3.376,5.836)}
\gppoint{gp mark 1}{(2.488,6.227)}
\gppoint{gp mark 1}{(11.377,2.336)}
\gppoint{gp mark 1}{(1.512,6.551)}
\gppoint{gp mark 1}{(6.495,4.359)}
\gppoint{gp mark 1}{(1.701,6.509)}
\gppoint{gp mark 1}{(8.840,3.435)}
\gppoint{gp mark 1}{(11.863,2.126)}
\gppoint{gp mark 1}{(3.197,5.876)}
\gppoint{gp mark 1}{(6.112,4.636)}
\gppoint{gp mark 1}{(5.350,4.951)}
\gppoint{gp mark 1}{(11.429,2.364)}
\gppoint{gp mark 1}{(1.568,6.632)}
\gppoint{gp mark 1}{(10.097,2.884)}
\gppoint{gp mark 1}{(9.102,3.319)}
\gppoint{gp mark 1}{(1.845,6.533)}
\gppoint{gp mark 1}{(9.720,3.035)}
\gppoint{gp mark 1}{(5.043,5.056)}
\gppoint{gp mark 1}{(4.373,5.108)}
\gppoint{gp mark 1}{(1.510,6.533)}
\gppoint{gp mark 1}{(1.562,6.399)}
\gppoint{gp mark 1}{(2.592,6.058)}
\gppoint{gp mark 1}{(11.252,2.362)}
\gppoint{gp mark 1}{(7.603,3.967)}
\gppoint{gp mark 1}{(3.767,5.622)}
\gppoint{gp mark 1}{(2.905,6.070)}
\gppoint{gp mark 1}{(9.905,2.908)}
\gppoint{gp mark 1}{(11.795,2.163)}
\gppoint{gp mark 1}{(8.951,3.402)}
\gppoint{gp mark 1}{(2.234,6.258)}
\gppoint{gp mark 1}{(3.237,5.814)}
\gppoint{gp mark 1}{(6.522,4.448)}
\gppoint{gp mark 1}{(5.824,4.814)}
\gppoint{gp mark 1}{(10.396,2.725)}
\gppoint{gp mark 1}{(1.649,6.558)}
\gppoint{gp mark 1}{(8.570,3.506)}
\gppoint{gp mark 1}{(6.869,4.260)}
\gppoint{gp mark 1}{(7.544,3.977)}
\gppoint{gp mark 1}{(10.808,2.598)}
\gppoint{gp mark 1}{(3.093,5.981)}
\gppoint{gp mark 1}{(8.199,3.710)}
\gppoint{gp mark 1}{(11.545,2.280)}
\gppoint{gp mark 1}{(2.412,6.125)}
\gppoint{gp mark 1}{(11.360,2.333)}
\gppoint{gp mark 1}{(4.289,5.292)}
\gppoint{gp mark 1}{(11.142,2.414)}
\gppoint{gp mark 1}{(5.879,4.665)}
\gpsetpointsize{2.80}
\gppoint{gp mark 12}{(3.582,6.348)}
\gppoint{gp mark 12}{(8.696,4.317)}
\gppoint{gp mark 12}{(2.115,3.397)}
\gppoint{gp mark 12}{(10.460,3.121)}
\gppoint{gp mark 12}{(1.785,3.042)}
\gppoint{gp mark 12}{(7.839,2.389)}
\gppoint{gp mark 12}{(11.109,2.350)}
\gppoint{gp mark 12}{(5.762,2.823)}
\gppoint{gp mark 12}{(2.785,2.525)}
\gppoint{gp mark 12}{(4.422,2.210)}
\gppoint{gp mark 12}{(10.179,2.484)}
\gppoint{gp mark 12}{(1.607,2.213)}
\gppoint{gp mark 12}{(2.362,2.382)}
\gppoint{gp mark 12}{(6.009,2.396)}
\gppoint{gp mark 12}{(10.680,2.179)}
\gppoint{gp mark 12}{(4.921,2.221)}
\gppoint{gp mark 12}{(9.511,2.176)}
\gppoint{gp mark 12}{(11.653,2.314)}
\gppoint{gp mark 12}{(6.153,2.263)}
\gppoint{gp mark 12}{(4.142,2.233)}
\gppoint{gp mark 12}{(1.557,2.213)}
\gppoint{gp mark 12}{(11.712,2.102)}
\gppoint{gp mark 12}{(9.039,2.221)}
\gppoint{gp mark 12}{(11.203,2.201)}
\gppoint{gp mark 12}{(1.544,2.232)}
\gppoint{gp mark 12}{(8.031,2.202)}
\gppoint{gp mark 12}{(4.542,2.195)}
\gppoint{gp mark 12}{(2.428,2.184)}
\gppoint{gp mark 12}{(11.751,2.221)}
\gppoint{gp mark 12}{(7.113,2.171)}
\gppoint{gp mark 12}{(5.470,2.190)}
\gppoint{gp mark 12}{(4.001,2.207)}
\gppoint{gp mark 12}{(2.248,2.177)}
\gppoint{gp mark 12}{(10.009,2.162)}
\gppoint{gp mark 12}{(2.869,2.138)}
\gppoint{gp mark 12}{(9.222,2.177)}
\gppoint{gp mark 12}{(10.951,2.168)}
\gppoint{gp mark 12}{(1.530,2.148)}
\gppoint{gp mark 12}{(1.733,2.162)}
\gppoint{gp mark 12}{(7.056,2.151)}
\gppoint{gp mark 12}{(8.438,2.157)}
\gppoint{gp mark 12}{(2.638,2.145)}
\gppoint{gp mark 12}{(11.529,2.123)}
\gppoint{gp mark 12}{(9.639,2.109)}
\gppoint{gp mark 12}{(2.014,2.131)}
\gppoint{gp mark 12}{(3.917,2.127)}
\gppoint{gp mark 12}{(9.973,2.183)}
\gppoint{gp mark 12}{(8.237,2.128)}
\gppoint{gp mark 12}{(4.590,2.156)}
\gppoint{gp mark 12}{(1.522,2.160)}
\gppoint{gp mark 12}{(10.772,2.141)}
\gppoint{gp mark 12}{(1.936,2.121)}
\gppoint{gp mark 12}{(6.402,2.144)}
\gppoint{gp mark 12}{(7.582,2.139)}
\gppoint{gp mark 12}{(11.634,2.125)}
\gppoint{gp mark 12}{(4.193,2.158)}
\gppoint{gp mark 12}{(11.327,2.162)}
\gppoint{gp mark 12}{(9.670,2.138)}
\gppoint{gp mark 12}{(1.519,2.131)}
\gppoint{gp mark 12}{(1.875,2.123)}
\gppoint{gp mark 12}{(4.801,2.133)}
\gppoint{gp mark 12}{(11.836,2.124)}
\gppoint{gp mark 12}{(2.616,2.139)}
\gppoint{gp mark 12}{(6.432,2.113)}
\gppoint{gp mark 12}{(9.948,2.129)}
\gppoint{gp mark 12}{(5.388,2.134)}
\gppoint{gp mark 12}{(3.081,2.121)}
\gppoint{gp mark 12}{(1.825,2.125)}
\gppoint{gp mark 12}{(9.413,2.146)}
\gppoint{gp mark 12}{(10.640,2.105)}
\gppoint{gp mark 12}{(7.947,2.119)}
\gppoint{gp mark 12}{(11.699,2.131)}
\gppoint{gp mark 12}{(11.467,2.158)}
\gppoint{gp mark 12}{(8.874,2.137)}
\gppoint{gp mark 12}{(7.420,2.117)}
\gppoint{gp mark 12}{(5.507,2.128)}
\gppoint{gp mark 12}{(10.789,2.145)}
\gppoint{gp mark 12}{(3.346,2.122)}
\gppoint{gp mark 12}{(11.723,2.122)}
\gppoint{gp mark 12}{(11.518,2.105)}
\gppoint{gp mark 12}{(3.821,2.126)}
\gppoint{gp mark 12}{(6.929,2.128)}
\gppoint{gp mark 12}{(2.735,2.130)}
\gppoint{gp mark 12}{(1.595,2.122)}
\gppoint{gp mark 12}{(1.752,2.117)}
\gppoint{gp mark 12}{(10.914,2.107)}
\gppoint{gp mark 12}{(3.146,2.130)}
\gppoint{gp mark 12}{(11.855,2.147)}
\gppoint{gp mark 12}{(9.469,2.115)}
\gppoint{gp mark 12}{(7.359,2.135)}
\gppoint{gp mark 12}{(5.606,2.126)}
\gppoint{gp mark 12}{(8.212,2.121)}
\gppoint{gp mark 12}{(11.316,2.124)}
\gppoint{gp mark 12}{(1.724,2.106)}
\gppoint{gp mark 12}{(2.976,2.114)}
\gppoint{gp mark 12}{(6.910,2.127)}
\gppoint{gp mark 12}{(7.748,2.100)}
\gppoint{gp mark 12}{(9.708,2.114)}
\gppoint{gp mark 12}{(11.859,2.106)}
\gppoint{gp mark 12}{(2.184,2.122)}
\gppoint{gp mark 12}{(3.376,2.121)}
\gppoint{gp mark 12}{(2.488,2.114)}
\gppoint{gp mark 12}{(11.377,2.108)}
\gppoint{gp mark 12}{(1.512,2.125)}
\gppoint{gp mark 12}{(6.495,2.120)}
\gppoint{gp mark 12}{(1.701,2.118)}
\gppoint{gp mark 12}{(8.840,2.109)}
\gppoint{gp mark 12}{(11.863,2.128)}
\gppoint{gp mark 12}{(3.197,2.116)}
\gppoint{gp mark 12}{(6.112,2.112)}
\gppoint{gp mark 12}{(5.350,2.111)}
\gppoint{gp mark 12}{(11.429,2.121)}
\gppoint{gp mark 12}{(1.568,2.111)}
\gppoint{gp mark 12}{(10.097,2.111)}
\gppoint{gp mark 12}{(9.102,2.127)}
\gppoint{gp mark 12}{(1.845,2.113)}
\gppoint{gp mark 12}{(9.720,2.109)}
\gppoint{gp mark 12}{(5.043,2.118)}
\gppoint{gp mark 12}{(4.373,2.132)}
\gppoint{gp mark 12}{(1.510,2.110)}
\gppoint{gp mark 12}{(1.562,2.116)}
\gppoint{gp mark 12}{(2.592,2.124)}
\gppoint{gp mark 12}{(11.252,2.108)}
\gppoint{gp mark 12}{(7.603,2.108)}
\gppoint{gp mark 12}{(3.767,2.121)}
\gppoint{gp mark 12}{(2.905,2.112)}
\gppoint{gp mark 12}{(9.905,2.127)}
\gppoint{gp mark 12}{(11.795,2.114)}
\gppoint{gp mark 12}{(8.951,2.107)}
\gppoint{gp mark 12}{(2.234,2.110)}
\gppoint{gp mark 12}{(3.237,2.124)}
\gppoint{gp mark 12}{(6.522,2.119)}
\gppoint{gp mark 12}{(5.824,2.117)}
\gppoint{gp mark 12}{(10.396,2.113)}
\gppoint{gp mark 12}{(1.649,2.121)}
\gppoint{gp mark 12}{(8.570,2.117)}
\gppoint{gp mark 12}{(6.869,2.105)}
\gppoint{gp mark 12}{(7.544,2.118)}
\gppoint{gp mark 12}{(10.808,2.124)}
\gppoint{gp mark 12}{(3.093,2.114)}
\gppoint{gp mark 12}{(8.199,2.109)}
\gppoint{gp mark 12}{(11.545,2.113)}
\gppoint{gp mark 12}{(2.412,2.116)}
\gppoint{gp mark 12}{(11.360,2.101)}
\gppoint{gp mark 12}{(4.289,2.105)}
\gppoint{gp mark 12}{(11.142,2.121)}
\gppoint{gp mark 12}{(5.879,2.123)}
\gpsetpointsize{2.00}
\gppoint{gp mark 7}{(10.460,8.135)}
\gppoint{gp mark 7}{(1.785,6.787)}
\gppoint{gp mark 7}{(7.839,8.215)}
\gppoint{gp mark 7}{(11.109,8.099)}
\gppoint{gp mark 7}{(5.762,6.744)}
\gppoint{gp mark 7}{(2.785,8.795)}
\gppoint{gp mark 7}{(4.422,8.271)}
\gppoint{gp mark 7}{(10.179,7.175)}
\gppoint{gp mark 7}{(1.607,8.430)}
\gppoint{gp mark 7}{(2.362,8.458)}
\gppoint{gp mark 7}{(6.009,8.276)}
\gppoint{gp mark 7}{(10.680,8.247)}
\gppoint{gp mark 7}{(4.921,8.284)}
\gppoint{gp mark 7}{(9.511,7.626)}
\gppoint{gp mark 7}{(11.653,7.764)}
\gppoint{gp mark 7}{(6.153,8.197)}
\gppoint{gp mark 7}{(4.142,8.364)}
\gppoint{gp mark 7}{(1.557,7.239)}
\gppoint{gp mark 7}{(11.712,7.753)}
\gppoint{gp mark 7}{(9.039,7.430)}
\gppoint{gp mark 7}{(11.203,7.916)}
\gppoint{gp mark 7}{(1.544,8.272)}
\gppoint{gp mark 7}{(8.031,8.030)}
\gppoint{gp mark 7}{(4.542,8.340)}
\gppoint{gp mark 7}{(2.428,7.106)}
\gppoint{gp mark 7}{(11.751,7.688)}
\gppoint{gp mark 7}{(7.113,7.550)}
\gppoint{gp mark 7}{(5.470,7.419)}
\gppoint{gp mark 7}{(4.001,8.268)}
\gppoint{gp mark 7}{(2.248,8.028)}
\gppoint{gp mark 7}{(10.009,7.950)}
\gppoint{gp mark 7}{(2.869,7.985)}
\gppoint{gp mark 7}{(9.222,7.919)}
\gppoint{gp mark 7}{(10.951,7.907)}
\gppoint{gp mark 7}{(1.530,7.313)}
\gppoint{gp mark 7}{(1.733,8.205)}
\gppoint{gp mark 7}{(7.056,7.423)}
\gppoint{gp mark 7}{(8.438,7.487)}
\gppoint{gp mark 7}{(2.638,8.149)}
\gppoint{gp mark 7}{(11.529,7.620)}
\gppoint{gp mark 7}{(9.639,7.717)}
\gppoint{gp mark 7}{(2.014,7.257)}
\gppoint{gp mark 7}{(3.917,7.382)}
\gppoint{gp mark 7}{(9.973,7.867)}
\gppoint{gp mark 7}{(8.237,7.578)}
\gppoint{gp mark 7}{(4.590,7.293)}
\gppoint{gp mark 7}{(1.522,8.054)}
\gppoint{gp mark 7}{(10.772,7.512)}
\gppoint{gp mark 7}{(1.936,8.001)}
\gppoint{gp mark 7}{(6.402,7.824)}
\gppoint{gp mark 7}{(7.582,7.363)}
\gppoint{gp mark 7}{(11.634,7.632)}
\gppoint{gp mark 7}{(4.193,7.430)}
\gppoint{gp mark 7}{(11.327,7.581)}
\gppoint{gp mark 7}{(9.670,7.432)}
\gppoint{gp mark 7}{(1.519,7.951)}
\gppoint{gp mark 7}{(1.875,8.123)}
\gppoint{gp mark 7}{(4.801,7.342)}
\gppoint{gp mark 7}{(11.836,7.739)}
\gppoint{gp mark 7}{(2.616,7.972)}
\gppoint{gp mark 7}{(6.432,7.505)}
\gppoint{gp mark 7}{(9.948,7.625)}
\gppoint{gp mark 7}{(5.388,7.364)}
\gppoint{gp mark 7}{(3.081,7.978)}
\gppoint{gp mark 7}{(1.825,7.465)}
\gppoint{gp mark 7}{(9.413,7.507)}
\gppoint{gp mark 7}{(10.640,7.810)}
\gppoint{gp mark 7}{(7.947,7.790)}
\gppoint{gp mark 7}{(11.699,7.633)}
\gppoint{gp mark 7}{(11.467,7.641)}
\gppoint{gp mark 7}{(8.874,7.932)}
\gppoint{gp mark 7}{(7.420,7.499)}
\gppoint{gp mark 7}{(5.507,7.890)}
\gppoint{gp mark 7}{(10.789,7.783)}
\gppoint{gp mark 7}{(3.346,7.862)}
\gppoint{gp mark 7}{(11.723,7.679)}
\gppoint{gp mark 7}{(11.518,7.734)}
\gppoint{gp mark 7}{(3.821,7.403)}
\gppoint{gp mark 7}{(6.929,7.958)}
\gppoint{gp mark 7}{(2.735,7.458)}
\gppoint{gp mark 7}{(1.595,7.427)}
\gppoint{gp mark 7}{(1.752,8.036)}
\gppoint{gp mark 7}{(10.914,7.706)}
\gppoint{gp mark 7}{(3.146,8.026)}
\gppoint{gp mark 7}{(11.855,7.672)}
\gppoint{gp mark 7}{(9.469,7.657)}
\gppoint{gp mark 7}{(7.359,7.477)}
\gppoint{gp mark 7}{(5.606,7.848)}
\gppoint{gp mark 7}{(8.212,7.891)}
\gppoint{gp mark 7}{(11.316,7.714)}
\gppoint{gp mark 7}{(1.724,7.423)}
\gppoint{gp mark 7}{(2.976,7.530)}
\gppoint{gp mark 7}{(6.910,7.843)}
\gppoint{gp mark 7}{(7.748,7.624)}
\gppoint{gp mark 7}{(9.708,7.504)}
\gppoint{gp mark 7}{(11.859,7.670)}
\gppoint{gp mark 7}{(2.184,7.904)}
\gppoint{gp mark 7}{(3.376,8.016)}
\gppoint{gp mark 7}{(2.488,7.947)}
\gppoint{gp mark 7}{(11.377,7.529)}
\gppoint{gp mark 7}{(1.512,7.917)}
\gppoint{gp mark 7}{(6.495,7.893)}
\gppoint{gp mark 7}{(1.701,7.894)}
\gppoint{gp mark 7}{(8.840,7.882)}
\gppoint{gp mark 7}{(11.863,7.717)}
\gppoint{gp mark 7}{(3.197,7.858)}
\gppoint{gp mark 7}{(6.112,7.773)}
\gppoint{gp mark 7}{(5.350,7.940)}
\gppoint{gp mark 7}{(11.429,7.729)}
\gppoint{gp mark 7}{(1.568,7.569)}
\gppoint{gp mark 7}{(10.097,7.791)}
\gppoint{gp mark 7}{(9.102,7.732)}
\gppoint{gp mark 7}{(1.845,7.523)}
\gppoint{gp mark 7}{(9.720,7.724)}
\gppoint{gp mark 7}{(5.043,7.426)}
\gppoint{gp mark 7}{(4.373,7.889)}
\gppoint{gp mark 7}{(1.510,7.909)}
\gppoint{gp mark 7}{(1.562,7.808)}
\gppoint{gp mark 7}{(2.592,7.479)}
\gppoint{gp mark 7}{(11.252,7.722)}
\gppoint{gp mark 7}{(7.603,7.842)}
\gppoint{gp mark 7}{(3.767,7.472)}
\gppoint{gp mark 7}{(2.905,7.564)}
\gppoint{gp mark 7}{(9.905,7.741)}
\gppoint{gp mark 7}{(11.795,7.631)}
\gppoint{gp mark 7}{(8.951,7.794)}
\gppoint{gp mark 7}{(2.234,7.946)}
\gppoint{gp mark 7}{(3.237,7.462)}
\gppoint{gp mark 7}{(6.522,7.529)}
\gppoint{gp mark 7}{(5.824,7.482)}
\gppoint{gp mark 7}{(10.396,7.509)}
\gppoint{gp mark 7}{(1.649,7.363)}
\gppoint{gp mark 7}{(8.570,7.808)}
\gppoint{gp mark 7}{(6.869,7.630)}
\gppoint{gp mark 7}{(7.544,7.736)}
\gppoint{gp mark 7}{(10.808,7.634)}
\gppoint{gp mark 7}{(3.093,7.539)}
\gppoint{gp mark 7}{(8.199,7.798)}
\gppoint{gp mark 7}{(11.545,7.732)}
\gppoint{gp mark 7}{(2.412,7.441)}
\gppoint{gp mark 7}{(11.360,7.708)}
\gppoint{gp mark 7}{(4.289,7.597)}
\gppoint{gp mark 7}{(11.142,7.605)}
\gppoint{gp mark 7}{(5.879,7.866)}
\gpcolor{color=gp lt color border}
\gpsetlinetype{gp lt border}
\draw[gp path] (1.504,9.029)--(1.504,0.985)--(11.893,0.985)--(11.893,9.029)--cycle;
\gpdefrectangularnode{gp plot 1}{\pgfpoint{1.504cm}{0.985cm}}{\pgfpoint{11.893cm}{9.029cm}}
\end{tikzpicture}
\end{center}
\end{figure}
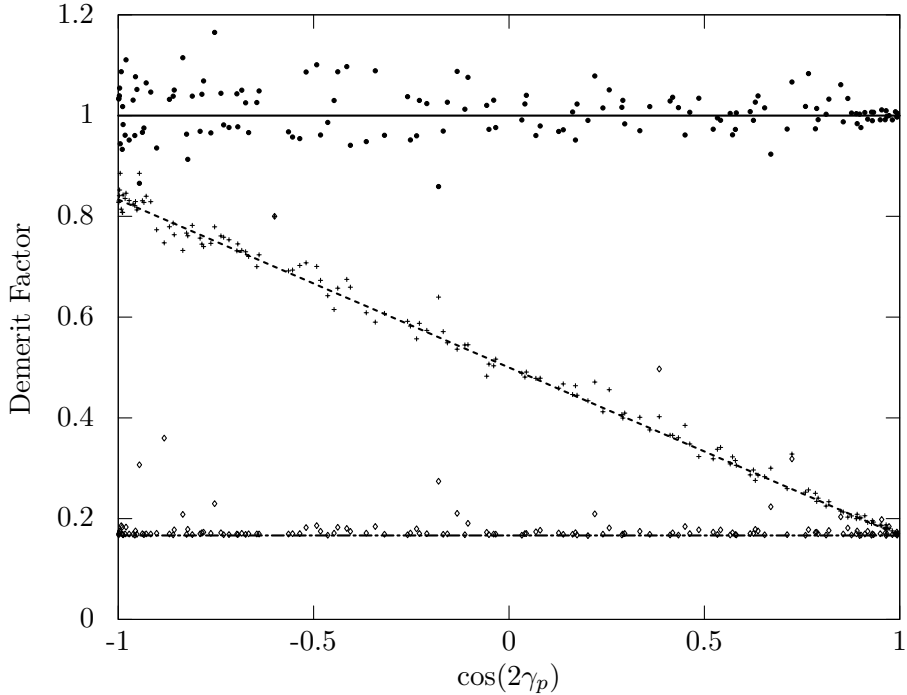
\end{center}

With Legendre sequences, we can improve the autocorrelation merit factors of our sequences to an asymptotic value as high as $6.342061\ldots$ \cite[Theorem 1.1]{Jedwab-Katz-Schmidt-2013-Littlewood}, the largest root of $29 x^3 - 249 x^2 +417 x-27$, by a process known as {\it appending}.
Instead of using the sequences $\fpnat$, $\gpnat$, and $\hpnat$ of shift $(p-1)/4$ and length $p$, we use sequences of the form $\fpapp=\fpslu$, $\gpapp=\gpslu$, and $\hpapp=\hpslu$ where $\ell$ is slightly larger than $p$: if $\Lambda_{\app}=1.057827\ldots$, the middle root of $4 x^3-30 x+27$, then we choose $\ell$ as close as possible to $p \cdot \Lambda_{\app}$ (we round to the nearest integer).  We choose $s$ to be as close as possible to $p \cdot (3-2 \Lambda_{\app})/4$.  To summarize, $\fpapp$ (or $\gpapp$ or $\hpapp$) is obtained by shifting $\fpzpu$ (or $\gpzpu$ or $\hpzpu$) to the left by about $22$ percent of its length, and then periodically extending the shifted sequence by about $6$ percent.

In \cite[Theorem 1.1]{Jedwab-Katz-Schmidt-2013-Littlewood} (or Theorem \ref{Rebecca} in this paper) it is shown that sequences $\hpapp$ achieve the record high asymptotic merit factor value of $6.342061\ldots$ (demerit factor $0.157677\ldots$) as $p\to \infty$.
With $\fpapp$ and $\gpapp$, Theorem \ref{Anne} shows that we cannot exceed this asymptotic value, but that we achieve the same asymptotic value if and only if we select sequences belonging to a sequence of primes $p$ where $\cos(2\gamma_p)\to 1$ as $p\to\infty$.
When this happens, Theorem \ref{Christopher} says that the asymptotic crosscorrelation merit factor for pairs $(\fpapp,\gpapp)$ becomes $0.994058\ldots$ (demerit factor $1.005976\ldots$).
When $\cos(2\gamma_p)$ tends to lower values, we obtain worse autocorrelation but better crosscorrelation, and in the extreme where $\cos(2\gamma_p) \to -1$ we obtain asymptotic autocorrelation merit factor of $1.158888\ldots$ (demerit factor $0.862896\ldots$) and asymptotic crosscorrelation merit factor $3.325929\ldots$ (demerit factor $0.300758\ldots$).  
Indeed, Theorems \ref{Anne} and \ref{Christopher} show that if $\gamma$ is a real number, and we have an infinite family of our sequence pairs $(\fpapp,\gpapp)$ using primes $p$ such that $\gamma_p \to \gamma$ as $p\to\infty$, then
\begin{align}
\DF(\fpapp) & \to (0.510286\ldots) - (0.352609\ldots) \cos(2\gamma) \nonumber \\
\DF(\gpapp) & \to (0.510286\ldots) - (0.352609\ldots) \cos(2\gamma) \label{Helen} \\
\CDF(\fpapp,\gpapp) & \to (0.653368\ldots) + (0.352609\ldots) \cos(2\gamma), \nonumber \\
\PSC(\fpapp,\gpapp) & \to 1.163654\ldots.\nonumber
\end{align}
Note that the asymptotic Pursley-Sarwate Criterion here is slightly lower than what we obtained above for unappended sequences.

In Figure \ref{Boris}, we show the dependence of autocorrelation and crosscorrelation demerit factors of $\fpapp$ and $\gpapp$ on $\cos(2\gamma_p)$ when we use sequence pairs $(\fpapp,\gpapp)$.  The lines indicate the asymptotic values calculated in Theorems \ref{Anne} and \ref{Christopher}, while the data points show the actual values for the sequence pair $(\fpapp,\gpapp)$ for all primes $p$ with $p\equiv 1 \pmod{4}$ and $p < 2000$.  To avoid clutter, we have only plotted the autocorrelation demerit factors of $\fpapp$: the values for $\gpapp$ are similar.  Comparison with Figure \ref{Aaron} shows that autocorrelation performance has become slightly more extreme as a result of appending (better in the best cases, worse in the worst).  Similarly, crosscorrelation performance has become more extreme.

\begin{center}
\begin{figure}
\begin{center}
\caption{Demerit factors of appended quartic residue sequences as a function of $\cos(2\gamma_p)$: autocorrelation (plusses, dashed line) and crosscorrelation (filled circles, solid line)}\label{Boris}
\begin{tikzpicture}[gnuplot]
\path (0.000,0.000) rectangle (12.446,9.398);
\gpcolor{color=gp lt color border}
\gpsetlinetype{gp lt border}
\gpsetlinewidth{1.00}
\draw[gp path] (1.504,0.985)--(1.684,0.985);
\draw[gp path] (11.893,0.985)--(11.713,0.985);
\node[gp node right] at (1.320,0.985) { 0};
\draw[gp path] (1.504,2.326)--(1.684,2.326);
\draw[gp path] (11.893,2.326)--(11.713,2.326);
\node[gp node right] at (1.320,2.326) { 0.2};
\draw[gp path] (1.504,3.666)--(1.684,3.666);
\draw[gp path] (11.893,3.666)--(11.713,3.666);
\node[gp node right] at (1.320,3.666) { 0.4};
\draw[gp path] (1.504,5.007)--(1.684,5.007);
\draw[gp path] (11.893,5.007)--(11.713,5.007);
\node[gp node right] at (1.320,5.007) { 0.6};
\draw[gp path] (1.504,6.348)--(1.684,6.348);
\draw[gp path] (11.893,6.348)--(11.713,6.348);
\node[gp node right] at (1.320,6.348) { 0.8};
\draw[gp path] (1.504,7.688)--(1.684,7.688);
\draw[gp path] (11.893,7.688)--(11.713,7.688);
\node[gp node right] at (1.320,7.688) { 1};
\draw[gp path] (1.504,9.029)--(1.684,9.029);
\draw[gp path] (11.893,9.029)--(11.713,9.029);
\node[gp node right] at (1.320,9.029) { 1.2};
\draw[gp path] (1.504,0.985)--(1.504,1.165);
\draw[gp path] (1.504,9.029)--(1.504,8.849);
\node[gp node center] at (1.504,0.677) {-1};
\draw[gp path] (4.101,0.985)--(4.101,1.165);
\draw[gp path] (4.101,9.029)--(4.101,8.849);
\node[gp node center] at (4.101,0.677) {-0.5};
\draw[gp path] (6.699,0.985)--(6.699,1.165);
\draw[gp path] (6.699,9.029)--(6.699,8.849);
\node[gp node center] at (6.699,0.677) { 0};
\draw[gp path] (9.296,0.985)--(9.296,1.165);
\draw[gp path] (9.296,9.029)--(9.296,8.849);
\node[gp node center] at (9.296,0.677) { 0.5};
\draw[gp path] (11.893,0.985)--(11.893,1.165);
\draw[gp path] (11.893,9.029)--(11.893,8.849);
\node[gp node center] at (11.893,0.677) { 1};
\draw[gp path] (1.504,9.029)--(1.504,0.985)--(11.893,0.985)--(11.893,9.029)--cycle;
\node[gp node center,rotate=-270] at (0.246,5.007) {Demerit Factor};
\node[gp node center] at (6.698,0.215) {$\cos(2\gamma_p)$};
\gpcolor{rgb color={0.000,0.000,0.000}}
\gpsetlinetype{gp lt plot 2}
\gpsetlinewidth{2.00}
\draw[gp path] (1.504,6.769)--(1.609,6.722)--(1.714,6.674)--(1.819,6.626)--(1.924,6.578)%
  --(2.029,6.531)--(2.134,6.483)--(2.239,6.435)--(2.344,6.387)--(2.448,6.340)--(2.553,6.292)%
  --(2.658,6.244)--(2.763,6.196)--(2.868,6.149)--(2.973,6.101)--(3.078,6.053)--(3.183,6.005)%
  --(3.288,5.958)--(3.393,5.910)--(3.498,5.862)--(3.603,5.814)--(3.708,5.767)--(3.813,5.719)%
  --(3.918,5.671)--(4.023,5.623)--(4.127,5.576)--(4.232,5.528)--(4.337,5.480)--(4.442,5.432)%
  --(4.547,5.385)--(4.652,5.337)--(4.757,5.289)--(4.862,5.241)--(4.967,5.194)--(5.072,5.146)%
  --(5.177,5.098)--(5.282,5.050)--(5.387,5.003)--(5.492,4.955)--(5.597,4.907)--(5.702,4.859)%
  --(5.807,4.812)--(5.911,4.764)--(6.016,4.716)--(6.121,4.668)--(6.226,4.620)--(6.331,4.573)%
  --(6.436,4.525)--(6.541,4.477)--(6.646,4.429)--(6.751,4.382)--(6.856,4.334)--(6.961,4.286)%
  --(7.066,4.238)--(7.171,4.191)--(7.276,4.143)--(7.381,4.095)--(7.486,4.047)--(7.590,4.000)%
  --(7.695,3.952)--(7.800,3.904)--(7.905,3.856)--(8.010,3.809)--(8.115,3.761)--(8.220,3.713)%
  --(8.325,3.665)--(8.430,3.618)--(8.535,3.570)--(8.640,3.522)--(8.745,3.474)--(8.850,3.427)%
  --(8.955,3.379)--(9.060,3.331)--(9.165,3.283)--(9.270,3.236)--(9.374,3.188)--(9.479,3.140)%
  --(9.584,3.092)--(9.689,3.045)--(9.794,2.997)--(9.899,2.949)--(10.004,2.901)--(10.109,2.854)%
  --(10.214,2.806)--(10.319,2.758)--(10.424,2.710)--(10.529,2.663)--(10.634,2.615)--(10.739,2.567)%
  --(10.844,2.519)--(10.949,2.472)--(11.053,2.424)--(11.158,2.376)--(11.263,2.328)--(11.368,2.281)%
  --(11.473,2.233)--(11.578,2.185)--(11.683,2.137)--(11.788,2.090)--(11.893,2.042);
\gpsetlinetype{gp lt plot 0}
\draw[gp path] (1.504,3.001)--(1.609,3.049)--(1.714,3.097)--(1.819,3.144)--(1.924,3.192)%
  --(2.029,3.240)--(2.134,3.288)--(2.239,3.335)--(2.344,3.383)--(2.448,3.431)--(2.553,3.479)%
  --(2.658,3.526)--(2.763,3.574)--(2.868,3.622)--(2.973,3.670)--(3.078,3.717)--(3.183,3.765)%
  --(3.288,3.813)--(3.393,3.861)--(3.498,3.908)--(3.603,3.956)--(3.708,4.004)--(3.813,4.052)%
  --(3.918,4.099)--(4.023,4.147)--(4.127,4.195)--(4.232,4.243)--(4.337,4.290)--(4.442,4.338)%
  --(4.547,4.386)--(4.652,4.434)--(4.757,4.481)--(4.862,4.529)--(4.967,4.577)--(5.072,4.625)%
  --(5.177,4.672)--(5.282,4.720)--(5.387,4.768)--(5.492,4.816)--(5.597,4.863)--(5.702,4.911)%
  --(5.807,4.959)--(5.911,5.007)--(6.016,5.054)--(6.121,5.102)--(6.226,5.150)--(6.331,5.198)%
  --(6.436,5.245)--(6.541,5.293)--(6.646,5.341)--(6.751,5.389)--(6.856,5.436)--(6.961,5.484)%
  --(7.066,5.532)--(7.171,5.580)--(7.276,5.627)--(7.381,5.675)--(7.486,5.723)--(7.590,5.771)%
  --(7.695,5.818)--(7.800,5.866)--(7.905,5.914)--(8.010,5.962)--(8.115,6.009)--(8.220,6.057)%
  --(8.325,6.105)--(8.430,6.153)--(8.535,6.200)--(8.640,6.248)--(8.745,6.296)--(8.850,6.344)%
  --(8.955,6.391)--(9.060,6.439)--(9.165,6.487)--(9.270,6.535)--(9.374,6.582)--(9.479,6.630)%
  --(9.584,6.678)--(9.689,6.726)--(9.794,6.773)--(9.899,6.821)--(10.004,6.869)--(10.109,6.917)%
  --(10.214,6.964)--(10.319,7.012)--(10.424,7.060)--(10.529,7.108)--(10.634,7.155)--(10.739,7.203)%
  --(10.844,7.251)--(10.949,7.299)--(11.053,7.346)--(11.158,7.394)--(11.263,7.442)--(11.368,7.490)%
  --(11.473,7.537)--(11.578,7.585)--(11.683,7.633)--(11.788,7.681)--(11.893,7.728);
\gpsetlinewidth{1.00}
\gpsetpointsize{2.00}
\gppoint{gp mark 1}{(3.582,6.348)}
\gppoint{gp mark 1}{(8.696,3.926)}
\gppoint{gp mark 1}{(2.115,6.654)}
\gppoint{gp mark 1}{(10.460,3.371)}
\gppoint{gp mark 1}{(1.785,7.323)}
\gppoint{gp mark 1}{(7.839,3.864)}
\gppoint{gp mark 1}{(11.109,2.541)}
\gppoint{gp mark 1}{(5.762,4.615)}
\gppoint{gp mark 1}{(2.785,6.561)}
\gppoint{gp mark 1}{(4.422,4.983)}
\gppoint{gp mark 1}{(10.179,3.021)}
\gppoint{gp mark 1}{(1.607,7.080)}
\gppoint{gp mark 1}{(2.362,6.347)}
\gppoint{gp mark 1}{(6.009,4.847)}
\gppoint{gp mark 1}{(10.680,2.781)}
\gppoint{gp mark 1}{(4.921,5.227)}
\gppoint{gp mark 1}{(9.511,3.343)}
\gppoint{gp mark 1}{(11.653,2.133)}
\gppoint{gp mark 1}{(6.153,4.973)}
\gppoint{gp mark 1}{(4.142,5.666)}
\gppoint{gp mark 1}{(1.557,7.028)}
\gppoint{gp mark 1}{(11.712,2.102)}
\gppoint{gp mark 1}{(9.039,3.638)}
\gppoint{gp mark 1}{(11.203,2.408)}
\gppoint{gp mark 1}{(1.544,6.958)}
\gppoint{gp mark 1}{(8.031,3.940)}
\gppoint{gp mark 1}{(4.542,5.752)}
\gppoint{gp mark 1}{(2.428,6.691)}
\gppoint{gp mark 1}{(11.751,2.140)}
\gppoint{gp mark 1}{(7.113,4.361)}
\gppoint{gp mark 1}{(5.470,5.055)}
\gppoint{gp mark 1}{(4.001,5.735)}
\gppoint{gp mark 1}{(2.248,6.402)}
\gppoint{gp mark 1}{(10.009,2.853)}
\gppoint{gp mark 1}{(2.869,5.980)}
\gppoint{gp mark 1}{(9.222,3.328)}
\gppoint{gp mark 1}{(10.951,2.471)}
\gppoint{gp mark 1}{(1.530,6.747)}
\gppoint{gp mark 1}{(1.733,6.820)}
\gppoint{gp mark 1}{(7.056,4.504)}
\gppoint{gp mark 1}{(8.438,3.683)}
\gppoint{gp mark 1}{(2.638,6.246)}
\gppoint{gp mark 1}{(11.529,2.264)}
\gppoint{gp mark 1}{(9.639,3.045)}
\gppoint{gp mark 1}{(2.014,6.383)}
\gppoint{gp mark 1}{(3.917,5.651)}
\gppoint{gp mark 1}{(9.973,2.961)}
\gppoint{gp mark 1}{(8.237,3.746)}
\gppoint{gp mark 1}{(4.590,5.431)}
\gppoint{gp mark 1}{(1.522,6.784)}
\gppoint{gp mark 1}{(10.772,2.559)}
\gppoint{gp mark 1}{(1.936,6.573)}
\gppoint{gp mark 1}{(6.402,4.585)}
\gppoint{gp mark 1}{(7.582,4.056)}
\gppoint{gp mark 1}{(11.634,2.150)}
\gppoint{gp mark 1}{(4.193,5.672)}
\gppoint{gp mark 1}{(11.327,2.309)}
\gppoint{gp mark 1}{(9.670,3.061)}
\gppoint{gp mark 1}{(1.519,6.641)}
\gppoint{gp mark 1}{(1.875,6.675)}
\gppoint{gp mark 1}{(4.801,5.493)}
\gppoint{gp mark 1}{(11.836,2.057)}
\gppoint{gp mark 1}{(2.616,6.273)}
\gppoint{gp mark 1}{(6.432,4.507)}
\gppoint{gp mark 1}{(9.948,2.969)}
\gppoint{gp mark 1}{(5.388,5.038)}
\gppoint{gp mark 1}{(3.081,6.063)}
\gppoint{gp mark 1}{(1.825,6.653)}
\gppoint{gp mark 1}{(9.413,3.141)}
\gppoint{gp mark 1}{(10.640,2.610)}
\gppoint{gp mark 1}{(7.947,3.729)}
\gppoint{gp mark 1}{(11.699,2.166)}
\gppoint{gp mark 1}{(11.467,2.216)}
\gppoint{gp mark 1}{(8.874,3.416)}
\gppoint{gp mark 1}{(7.420,4.103)}
\gppoint{gp mark 1}{(5.507,5.044)}
\gppoint{gp mark 1}{(10.789,2.529)}
\gppoint{gp mark 1}{(3.346,5.971)}
\gppoint{gp mark 1}{(11.723,2.133)}
\gppoint{gp mark 1}{(11.518,2.220)}
\gppoint{gp mark 1}{(3.821,5.734)}
\gppoint{gp mark 1}{(6.929,4.320)}
\gppoint{gp mark 1}{(2.735,6.191)}
\gppoint{gp mark 1}{(1.595,6.688)}
\gppoint{gp mark 1}{(1.752,6.650)}
\gppoint{gp mark 1}{(10.914,2.472)}
\gppoint{gp mark 1}{(3.146,6.084)}
\gppoint{gp mark 1}{(11.855,2.048)}
\gppoint{gp mark 1}{(9.469,3.229)}
\gppoint{gp mark 1}{(7.359,4.106)}
\gppoint{gp mark 1}{(5.606,4.944)}
\gppoint{gp mark 1}{(8.212,3.696)}
\gppoint{gp mark 1}{(11.316,2.320)}
\gppoint{gp mark 1}{(1.724,6.671)}
\gppoint{gp mark 1}{(2.976,6.104)}
\gppoint{gp mark 1}{(6.910,4.316)}
\gppoint{gp mark 1}{(7.748,4.007)}
\gppoint{gp mark 1}{(9.708,3.014)}
\gppoint{gp mark 1}{(11.859,2.054)}
\gppoint{gp mark 1}{(2.184,6.405)}
\gppoint{gp mark 1}{(3.376,5.995)}
\gppoint{gp mark 1}{(2.488,6.410)}
\gppoint{gp mark 1}{(11.377,2.308)}
\gppoint{gp mark 1}{(1.512,6.828)}
\gppoint{gp mark 1}{(6.495,4.465)}
\gppoint{gp mark 1}{(1.701,6.692)}
\gppoint{gp mark 1}{(8.840,3.458)}
\gppoint{gp mark 1}{(11.863,2.063)}
\gppoint{gp mark 1}{(3.197,5.902)}
\gppoint{gp mark 1}{(6.112,4.658)}
\gppoint{gp mark 1}{(5.350,4.986)}
\gppoint{gp mark 1}{(11.429,2.292)}
\gppoint{gp mark 1}{(1.568,6.762)}
\gppoint{gp mark 1}{(10.097,2.866)}
\gppoint{gp mark 1}{(9.102,3.297)}
\gppoint{gp mark 1}{(1.845,6.650)}
\gppoint{gp mark 1}{(9.720,2.921)}
\gppoint{gp mark 1}{(5.043,5.151)}
\gppoint{gp mark 1}{(4.373,5.418)}
\gppoint{gp mark 1}{(1.510,6.757)}
\gppoint{gp mark 1}{(1.562,6.727)}
\gppoint{gp mark 1}{(2.592,6.239)}
\gppoint{gp mark 1}{(11.252,2.294)}
\gppoint{gp mark 1}{(7.603,4.071)}
\gppoint{gp mark 1}{(3.767,5.789)}
\gppoint{gp mark 1}{(2.905,6.259)}
\gppoint{gp mark 1}{(9.905,2.956)}
\gppoint{gp mark 1}{(11.795,2.107)}
\gppoint{gp mark 1}{(8.951,3.413)}
\gppoint{gp mark 1}{(2.234,6.485)}
\gppoint{gp mark 1}{(3.237,5.990)}
\gppoint{gp mark 1}{(6.522,4.500)}
\gppoint{gp mark 1}{(5.824,4.850)}
\gppoint{gp mark 1}{(10.396,2.693)}
\gppoint{gp mark 1}{(1.649,6.615)}
\gppoint{gp mark 1}{(8.570,3.548)}
\gppoint{gp mark 1}{(6.869,4.381)}
\gppoint{gp mark 1}{(7.544,4.075)}
\gppoint{gp mark 1}{(10.808,2.591)}
\gppoint{gp mark 1}{(3.093,6.111)}
\gppoint{gp mark 1}{(8.199,3.706)}
\gppoint{gp mark 1}{(11.545,2.213)}
\gppoint{gp mark 1}{(2.412,6.248)}
\gppoint{gp mark 1}{(11.360,2.273)}
\gppoint{gp mark 1}{(4.289,5.502)}
\gppoint{gp mark 1}{(11.142,2.397)}
\gppoint{gp mark 1}{(5.879,4.754)}
\gppoint{gp mark 7}{(3.582,3.398)}
\gppoint{gp mark 7}{(8.696,6.252)}
\gppoint{gp mark 7}{(2.115,4.419)}
\gppoint{gp mark 7}{(10.460,7.088)}
\gppoint{gp mark 7}{(1.785,3.272)}
\gppoint{gp mark 7}{(7.839,5.912)}
\gppoint{gp mark 7}{(11.109,7.483)}
\gppoint{gp mark 7}{(5.762,5.055)}
\gppoint{gp mark 7}{(2.785,3.645)}
\gppoint{gp mark 7}{(4.422,4.689)}
\gppoint{gp mark 7}{(10.179,7.005)}
\gppoint{gp mark 7}{(1.607,2.956)}
\gppoint{gp mark 7}{(2.362,3.537)}
\gppoint{gp mark 7}{(6.009,5.160)}
\gppoint{gp mark 7}{(10.680,7.248)}
\gppoint{gp mark 7}{(4.921,4.723)}
\gppoint{gp mark 7}{(9.511,6.673)}
\gppoint{gp mark 7}{(11.653,7.596)}
\gppoint{gp mark 7}{(6.153,5.141)}
\gppoint{gp mark 7}{(4.142,4.228)}
\gppoint{gp mark 7}{(1.557,3.141)}
\gppoint{gp mark 7}{(11.712,7.690)}
\gppoint{gp mark 7}{(9.039,6.498)}
\gppoint{gp mark 7}{(11.203,7.453)}
\gppoint{gp mark 7}{(1.544,3.154)}
\gppoint{gp mark 7}{(8.031,5.913)}
\gppoint{gp mark 7}{(4.542,4.272)}
\gppoint{gp mark 7}{(2.428,3.354)}
\gppoint{gp mark 7}{(11.751,7.653)}
\gppoint{gp mark 7}{(7.113,5.526)}
\gppoint{gp mark 7}{(5.470,4.860)}
\gppoint{gp mark 7}{(4.001,4.185)}
\gppoint{gp mark 7}{(2.248,3.384)}
\gppoint{gp mark 7}{(10.009,6.950)}
\gppoint{gp mark 7}{(2.869,3.680)}
\gppoint{gp mark 7}{(9.222,6.526)}
\gppoint{gp mark 7}{(10.951,7.308)}
\gppoint{gp mark 7}{(1.530,3.025)}
\gppoint{gp mark 7}{(1.733,3.094)}
\gppoint{gp mark 7}{(7.056,5.577)}
\gppoint{gp mark 7}{(8.438,6.181)}
\gppoint{gp mark 7}{(2.638,3.537)}
\gppoint{gp mark 7}{(11.529,7.580)}
\gppoint{gp mark 7}{(9.639,6.827)}
\gppoint{gp mark 7}{(2.014,3.338)}
\gppoint{gp mark 7}{(3.917,4.124)}
\gppoint{gp mark 7}{(9.973,6.836)}
\gppoint{gp mark 7}{(8.237,6.085)}
\gppoint{gp mark 7}{(4.590,4.413)}
\gppoint{gp mark 7}{(1.522,3.011)}
\gppoint{gp mark 7}{(10.772,7.172)}
\gppoint{gp mark 7}{(1.936,3.225)}
\gppoint{gp mark 7}{(6.402,5.307)}
\gppoint{gp mark 7}{(7.582,5.696)}
\gppoint{gp mark 7}{(11.634,7.605)}
\gppoint{gp mark 7}{(4.193,4.164)}
\gppoint{gp mark 7}{(11.327,7.476)}
\gppoint{gp mark 7}{(9.670,6.737)}
\gppoint{gp mark 7}{(1.519,3.065)}
\gppoint{gp mark 7}{(1.875,3.149)}
\gppoint{gp mark 7}{(4.801,4.481)}
\gppoint{gp mark 7}{(11.836,7.737)}
\gppoint{gp mark 7}{(2.616,3.544)}
\gppoint{gp mark 7}{(6.432,5.252)}
\gppoint{gp mark 7}{(9.948,6.830)}
\gppoint{gp mark 7}{(5.388,4.798)}
\gppoint{gp mark 7}{(3.081,3.784)}
\gppoint{gp mark 7}{(1.825,3.120)}
\gppoint{gp mark 7}{(9.413,6.589)}
\gppoint{gp mark 7}{(10.640,7.167)}
\gppoint{gp mark 7}{(7.947,5.954)}
\gppoint{gp mark 7}{(11.699,7.602)}
\gppoint{gp mark 7}{(11.467,7.537)}
\gppoint{gp mark 7}{(8.874,6.416)}
\gppoint{gp mark 7}{(7.420,5.675)}
\gppoint{gp mark 7}{(5.507,4.787)}
\gppoint{gp mark 7}{(10.789,7.204)}
\gppoint{gp mark 7}{(3.346,3.851)}
\gppoint{gp mark 7}{(11.723,7.660)}
\gppoint{gp mark 7}{(11.518,7.590)}
\gppoint{gp mark 7}{(3.821,4.050)}
\gppoint{gp mark 7}{(6.929,5.482)}
\gppoint{gp mark 7}{(2.735,3.555)}
\gppoint{gp mark 7}{(1.595,3.055)}
\gppoint{gp mark 7}{(1.752,3.182)}
\gppoint{gp mark 7}{(10.914,7.326)}
\gppoint{gp mark 7}{(3.146,3.746)}
\gppoint{gp mark 7}{(11.855,7.710)}
\gppoint{gp mark 7}{(9.469,6.631)}
\gppoint{gp mark 7}{(7.359,5.642)}
\gppoint{gp mark 7}{(5.606,4.854)}
\gppoint{gp mark 7}{(8.212,6.037)}
\gppoint{gp mark 7}{(11.316,7.543)}
\gppoint{gp mark 7}{(1.724,3.107)}
\gppoint{gp mark 7}{(2.976,3.665)}
\gppoint{gp mark 7}{(6.910,5.493)}
\gppoint{gp mark 7}{(7.748,5.883)}
\gppoint{gp mark 7}{(9.708,6.702)}
\gppoint{gp mark 7}{(11.859,7.738)}
\gppoint{gp mark 7}{(2.184,3.317)}
\gppoint{gp mark 7}{(3.376,3.855)}
\gppoint{gp mark 7}{(2.488,3.426)}
\gppoint{gp mark 7}{(11.377,7.458)}
\gppoint{gp mark 7}{(1.512,3.007)}
\gppoint{gp mark 7}{(6.495,5.294)}
\gppoint{gp mark 7}{(1.701,3.114)}
\gppoint{gp mark 7}{(8.840,6.373)}
\gppoint{gp mark 7}{(11.863,7.698)}
\gppoint{gp mark 7}{(3.197,3.793)}
\gppoint{gp mark 7}{(6.112,5.093)}
\gppoint{gp mark 7}{(5.350,4.751)}
\gppoint{gp mark 7}{(11.429,7.554)}
\gppoint{gp mark 7}{(1.568,3.041)}
\gppoint{gp mark 7}{(10.097,6.915)}
\gppoint{gp mark 7}{(9.102,6.467)}
\gppoint{gp mark 7}{(1.845,3.149)}
\gppoint{gp mark 7}{(9.720,6.770)}
\gppoint{gp mark 7}{(5.043,4.571)}
\gppoint{gp mark 7}{(4.373,4.318)}
\gppoint{gp mark 7}{(1.510,3.042)}
\gppoint{gp mark 7}{(1.562,3.043)}
\gppoint{gp mark 7}{(2.592,3.518)}
\gppoint{gp mark 7}{(11.252,7.433)}
\gppoint{gp mark 7}{(7.603,5.758)}
\gppoint{gp mark 7}{(3.767,4.049)}
\gppoint{gp mark 7}{(2.905,3.652)}
\gppoint{gp mark 7}{(9.905,6.820)}
\gppoint{gp mark 7}{(11.795,7.640)}
\gppoint{gp mark 7}{(8.951,6.388)}
\gppoint{gp mark 7}{(2.234,3.350)}
\gppoint{gp mark 7}{(3.237,3.761)}
\gppoint{gp mark 7}{(6.522,5.270)}
\gppoint{gp mark 7}{(5.824,4.923)}
\gppoint{gp mark 7}{(10.396,7.022)}
\gppoint{gp mark 7}{(1.649,3.093)}
\gppoint{gp mark 7}{(8.570,6.197)}
\gppoint{gp mark 7}{(6.869,5.448)}
\gppoint{gp mark 7}{(7.544,5.755)}
\gppoint{gp mark 7}{(10.808,7.209)}
\gppoint{gp mark 7}{(3.093,3.707)}
\gppoint{gp mark 7}{(8.199,6.075)}
\gppoint{gp mark 7}{(11.545,7.549)}
\gppoint{gp mark 7}{(2.412,3.442)}
\gppoint{gp mark 7}{(11.360,7.513)}
\gppoint{gp mark 7}{(4.289,4.302)}
\gppoint{gp mark 7}{(11.142,7.387)}
\gppoint{gp mark 7}{(5.879,4.996)}
\gpcolor{color=gp lt color border}
\gpsetlinetype{gp lt border}
\draw[gp path] (1.504,9.029)--(1.504,0.985)--(11.893,0.985)--(11.893,9.029)--cycle;
\gpdefrectangularnode{gp plot 1}{\pgfpoint{1.504cm}{0.985cm}}{\pgfpoint{11.893cm}{9.029cm}}
\end{tikzpicture}
\end{center}
\end{figure}
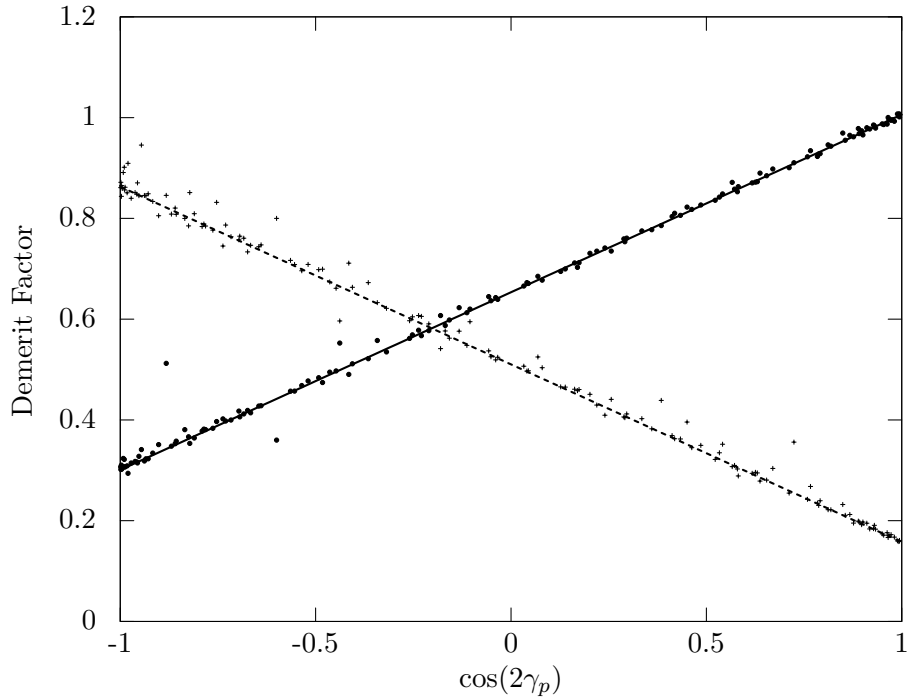
\end{center}

If we crosscorrelate our appended sequences $\fpapp$ and $\gpapp$ derived from quartic characters with appended shifted Legendre sequences $\hpapp$, the performance does not depend appreciably on the prime $p$.  This is also in accord with our proof in Theorem \ref{Clarence} that the asymptotic crosscorrelation demerit factor for $(\fpapp,\hpapp)$ or $(\gpapp,\hpapp)$ should always tend to $1.005976\ldots$ (merit factor $0.994058\ldots$) as $p\to\infty$.
In Figure \ref{Edith}, we show the dependence of autocorrelation and crosscorrelation demerit factors for $\fpapp$ and $\hpapp$ on $\cos(2\gamma_p)$.
The lines indicate the asymptotic values calculated in Theorems \ref{Anne}, \ref{Rebecca}, and \ref{Clarence} of this paper.
The data points show the actual values for the sequence pair $(\fpapp,\hpapp)$ for all primes $p$ with $p\equiv 1 \pmod{4}$ and $p < 2000$.  Note that the autocorrelation demerit factor for $\hpapp$ shows little dependence on $\gamma_p$ (it has asymptotic value of $0.157677\ldots$, or equivalently merit factor $6.342061\ldots$, regardless of $\gamma_p$), while the autocorrelation demerit factor of $\fpapp$ has significant dependence, as was already seen in Figure \ref{Boris}.  The crosscorrelation demerit factor for $(\fpapp,\hpapp)$ is always around $1.006976\ldots$ (merit factor $0.994058\ldots$), which is the asymptotic value given by Theorem \ref{Clarence}.
One would get a similar plot if $\fpapp$ were replaced with $\gpapp$.

\begin{center}
\begin{figure}
\begin{center}
\caption{Demerit factors as a function of $\cos(2\gamma_p)$: autocorrelation of appended quartic residue sequence (plusses, dashed line); autocorrelation of appended Legendre sequence (diamonds, dot-dash line); crosscorrelation of these paired (filled circles, solid line)}\label{Edith}
\begin{tikzpicture}[gnuplot]
\path (0.000,0.000) rectangle (12.446,9.398);
\gpcolor{color=gp lt color border}
\gpsetlinetype{gp lt border}
\gpsetlinewidth{1.00}
\draw[gp path] (1.504,0.985)--(1.684,0.985);
\draw[gp path] (11.893,0.985)--(11.713,0.985);
\node[gp node right] at (1.320,0.985) { 0};
\draw[gp path] (1.504,2.326)--(1.684,2.326);
\draw[gp path] (11.893,2.326)--(11.713,2.326);
\node[gp node right] at (1.320,2.326) { 0.2};
\draw[gp path] (1.504,3.666)--(1.684,3.666);
\draw[gp path] (11.893,3.666)--(11.713,3.666);
\node[gp node right] at (1.320,3.666) { 0.4};
\draw[gp path] (1.504,5.007)--(1.684,5.007);
\draw[gp path] (11.893,5.007)--(11.713,5.007);
\node[gp node right] at (1.320,5.007) { 0.6};
\draw[gp path] (1.504,6.348)--(1.684,6.348);
\draw[gp path] (11.893,6.348)--(11.713,6.348);
\node[gp node right] at (1.320,6.348) { 0.8};
\draw[gp path] (1.504,7.688)--(1.684,7.688);
\draw[gp path] (11.893,7.688)--(11.713,7.688);
\node[gp node right] at (1.320,7.688) { 1};
\draw[gp path] (1.504,9.029)--(1.684,9.029);
\draw[gp path] (11.893,9.029)--(11.713,9.029);
\node[gp node right] at (1.320,9.029) { 1.2};
\draw[gp path] (1.504,0.985)--(1.504,1.165);
\draw[gp path] (1.504,9.029)--(1.504,8.849);
\node[gp node center] at (1.504,0.677) {-1};
\draw[gp path] (4.101,0.985)--(4.101,1.165);
\draw[gp path] (4.101,9.029)--(4.101,8.849);
\node[gp node center] at (4.101,0.677) {-0.5};
\draw[gp path] (6.699,0.985)--(6.699,1.165);
\draw[gp path] (6.699,9.029)--(6.699,8.849);
\node[gp node center] at (6.699,0.677) { 0};
\draw[gp path] (9.296,0.985)--(9.296,1.165);
\draw[gp path] (9.296,9.029)--(9.296,8.849);
\node[gp node center] at (9.296,0.677) { 0.5};
\draw[gp path] (11.893,0.985)--(11.893,1.165);
\draw[gp path] (11.893,9.029)--(11.893,8.849);
\node[gp node center] at (11.893,0.677) { 1};
\draw[gp path] (1.504,9.029)--(1.504,0.985)--(11.893,0.985)--(11.893,9.029)--cycle;
\node[gp node center,rotate=-270] at (0.246,5.007) {Demerit Factor};
\node[gp node center] at (6.698,0.215) {$\cos(2\gamma_p)$};
\gpcolor{rgb color={0.000,0.000,0.000}}
\gpsetlinetype{gp lt plot 2}
\gpsetlinewidth{2.00}
\draw[gp path] (1.504,6.769)--(1.609,6.722)--(1.714,6.674)--(1.819,6.626)--(1.924,6.578)%
  --(2.029,6.531)--(2.134,6.483)--(2.239,6.435)--(2.344,6.387)--(2.448,6.340)--(2.553,6.292)%
  --(2.658,6.244)--(2.763,6.196)--(2.868,6.149)--(2.973,6.101)--(3.078,6.053)--(3.183,6.005)%
  --(3.288,5.958)--(3.393,5.910)--(3.498,5.862)--(3.603,5.814)--(3.708,5.767)--(3.813,5.719)%
  --(3.918,5.671)--(4.023,5.623)--(4.127,5.576)--(4.232,5.528)--(4.337,5.480)--(4.442,5.432)%
  --(4.547,5.385)--(4.652,5.337)--(4.757,5.289)--(4.862,5.241)--(4.967,5.194)--(5.072,5.146)%
  --(5.177,5.098)--(5.282,5.050)--(5.387,5.003)--(5.492,4.955)--(5.597,4.907)--(5.702,4.859)%
  --(5.807,4.812)--(5.911,4.764)--(6.016,4.716)--(6.121,4.668)--(6.226,4.620)--(6.331,4.573)%
  --(6.436,4.525)--(6.541,4.477)--(6.646,4.429)--(6.751,4.382)--(6.856,4.334)--(6.961,4.286)%
  --(7.066,4.238)--(7.171,4.191)--(7.276,4.143)--(7.381,4.095)--(7.486,4.047)--(7.590,4.000)%
  --(7.695,3.952)--(7.800,3.904)--(7.905,3.856)--(8.010,3.809)--(8.115,3.761)--(8.220,3.713)%
  --(8.325,3.665)--(8.430,3.618)--(8.535,3.570)--(8.640,3.522)--(8.745,3.474)--(8.850,3.427)%
  --(8.955,3.379)--(9.060,3.331)--(9.165,3.283)--(9.270,3.236)--(9.374,3.188)--(9.479,3.140)%
  --(9.584,3.092)--(9.689,3.045)--(9.794,2.997)--(9.899,2.949)--(10.004,2.901)--(10.109,2.854)%
  --(10.214,2.806)--(10.319,2.758)--(10.424,2.710)--(10.529,2.663)--(10.634,2.615)--(10.739,2.567)%
  --(10.844,2.519)--(10.949,2.472)--(11.053,2.424)--(11.158,2.376)--(11.263,2.328)--(11.368,2.281)%
  --(11.473,2.233)--(11.578,2.185)--(11.683,2.137)--(11.788,2.090)--(11.893,2.042);
\gpsetlinetype{gp lt plot 4}
\draw[gp path] (1.504,2.042)--(1.609,2.042)--(1.714,2.042)--(1.819,2.042)--(1.924,2.042)%
  --(2.029,2.042)--(2.134,2.042)--(2.239,2.042)--(2.344,2.042)--(2.448,2.042)--(2.553,2.042)%
  --(2.658,2.042)--(2.763,2.042)--(2.868,2.042)--(2.973,2.042)--(3.078,2.042)--(3.183,2.042)%
  --(3.288,2.042)--(3.393,2.042)--(3.498,2.042)--(3.603,2.042)--(3.708,2.042)--(3.813,2.042)%
  --(3.918,2.042)--(4.023,2.042)--(4.127,2.042)--(4.232,2.042)--(4.337,2.042)--(4.442,2.042)%
  --(4.547,2.042)--(4.652,2.042)--(4.757,2.042)--(4.862,2.042)--(4.967,2.042)--(5.072,2.042)%
  --(5.177,2.042)--(5.282,2.042)--(5.387,2.042)--(5.492,2.042)--(5.597,2.042)--(5.702,2.042)%
  --(5.807,2.042)--(5.911,2.042)--(6.016,2.042)--(6.121,2.042)--(6.226,2.042)--(6.331,2.042)%
  --(6.436,2.042)--(6.541,2.042)--(6.646,2.042)--(6.751,2.042)--(6.856,2.042)--(6.961,2.042)%
  --(7.066,2.042)--(7.171,2.042)--(7.276,2.042)--(7.381,2.042)--(7.486,2.042)--(7.590,2.042)%
  --(7.695,2.042)--(7.800,2.042)--(7.905,2.042)--(8.010,2.042)--(8.115,2.042)--(8.220,2.042)%
  --(8.325,2.042)--(8.430,2.042)--(8.535,2.042)--(8.640,2.042)--(8.745,2.042)--(8.850,2.042)%
  --(8.955,2.042)--(9.060,2.042)--(9.165,2.042)--(9.270,2.042)--(9.374,2.042)--(9.479,2.042)%
  --(9.584,2.042)--(9.689,2.042)--(9.794,2.042)--(9.899,2.042)--(10.004,2.042)--(10.109,2.042)%
  --(10.214,2.042)--(10.319,2.042)--(10.424,2.042)--(10.529,2.042)--(10.634,2.042)--(10.739,2.042)%
  --(10.844,2.042)--(10.949,2.042)--(11.053,2.042)--(11.158,2.042)--(11.263,2.042)--(11.368,2.042)%
  --(11.473,2.042)--(11.578,2.042)--(11.683,2.042)--(11.788,2.042)--(11.893,2.042);
\gpsetlinetype{gp lt plot 0}
\draw[gp path] (1.504,7.728)--(1.609,7.728)--(1.714,7.728)--(1.819,7.728)--(1.924,7.728)%
  --(2.029,7.728)--(2.134,7.728)--(2.239,7.728)--(2.344,7.728)--(2.448,7.728)--(2.553,7.728)%
  --(2.658,7.728)--(2.763,7.728)--(2.868,7.728)--(2.973,7.728)--(3.078,7.728)--(3.183,7.728)%
  --(3.288,7.728)--(3.393,7.728)--(3.498,7.728)--(3.603,7.728)--(3.708,7.728)--(3.813,7.728)%
  --(3.918,7.728)--(4.023,7.728)--(4.127,7.728)--(4.232,7.728)--(4.337,7.728)--(4.442,7.728)%
  --(4.547,7.728)--(4.652,7.728)--(4.757,7.728)--(4.862,7.728)--(4.967,7.728)--(5.072,7.728)%
  --(5.177,7.728)--(5.282,7.728)--(5.387,7.728)--(5.492,7.728)--(5.597,7.728)--(5.702,7.728)%
  --(5.807,7.728)--(5.911,7.728)--(6.016,7.728)--(6.121,7.728)--(6.226,7.728)--(6.331,7.728)%
  --(6.436,7.728)--(6.541,7.728)--(6.646,7.728)--(6.751,7.728)--(6.856,7.728)--(6.961,7.728)%
  --(7.066,7.728)--(7.171,7.728)--(7.276,7.728)--(7.381,7.728)--(7.486,7.728)--(7.590,7.728)%
  --(7.695,7.728)--(7.800,7.728)--(7.905,7.728)--(8.010,7.728)--(8.115,7.728)--(8.220,7.728)%
  --(8.325,7.728)--(8.430,7.728)--(8.535,7.728)--(8.640,7.728)--(8.745,7.728)--(8.850,7.728)%
  --(8.955,7.728)--(9.060,7.728)--(9.165,7.728)--(9.270,7.728)--(9.374,7.728)--(9.479,7.728)%
  --(9.584,7.728)--(9.689,7.728)--(9.794,7.728)--(9.899,7.728)--(10.004,7.728)--(10.109,7.728)%
  --(10.214,7.728)--(10.319,7.728)--(10.424,7.728)--(10.529,7.728)--(10.634,7.728)--(10.739,7.728)%
  --(10.844,7.728)--(10.949,7.728)--(11.053,7.728)--(11.158,7.728)--(11.263,7.728)--(11.368,7.728)%
  --(11.473,7.728)--(11.578,7.728)--(11.683,7.728)--(11.788,7.728)--(11.893,7.728);
\gpsetlinewidth{1.00}
\gpsetpointsize{2.00}
\gppoint{gp mark 1}{(3.582,6.348)}
\gppoint{gp mark 1}{(8.696,3.926)}
\gppoint{gp mark 1}{(2.115,6.654)}
\gppoint{gp mark 1}{(10.460,3.371)}
\gppoint{gp mark 1}{(1.785,7.323)}
\gppoint{gp mark 1}{(7.839,3.864)}
\gppoint{gp mark 1}{(11.109,2.541)}
\gppoint{gp mark 1}{(5.762,4.615)}
\gppoint{gp mark 1}{(2.785,6.561)}
\gppoint{gp mark 1}{(4.422,4.983)}
\gppoint{gp mark 1}{(10.179,3.021)}
\gppoint{gp mark 1}{(1.607,7.080)}
\gppoint{gp mark 1}{(2.362,6.347)}
\gppoint{gp mark 1}{(6.009,4.847)}
\gppoint{gp mark 1}{(10.680,2.781)}
\gppoint{gp mark 1}{(4.921,5.227)}
\gppoint{gp mark 1}{(9.511,3.343)}
\gppoint{gp mark 1}{(11.653,2.133)}
\gppoint{gp mark 1}{(6.153,4.973)}
\gppoint{gp mark 1}{(4.142,5.666)}
\gppoint{gp mark 1}{(1.557,7.028)}
\gppoint{gp mark 1}{(11.712,2.102)}
\gppoint{gp mark 1}{(9.039,3.638)}
\gppoint{gp mark 1}{(11.203,2.408)}
\gppoint{gp mark 1}{(1.544,6.958)}
\gppoint{gp mark 1}{(8.031,3.940)}
\gppoint{gp mark 1}{(4.542,5.752)}
\gppoint{gp mark 1}{(2.428,6.691)}
\gppoint{gp mark 1}{(11.751,2.140)}
\gppoint{gp mark 1}{(7.113,4.361)}
\gppoint{gp mark 1}{(5.470,5.055)}
\gppoint{gp mark 1}{(4.001,5.735)}
\gppoint{gp mark 1}{(2.248,6.402)}
\gppoint{gp mark 1}{(10.009,2.853)}
\gppoint{gp mark 1}{(2.869,5.980)}
\gppoint{gp mark 1}{(9.222,3.328)}
\gppoint{gp mark 1}{(10.951,2.471)}
\gppoint{gp mark 1}{(1.530,6.747)}
\gppoint{gp mark 1}{(1.733,6.820)}
\gppoint{gp mark 1}{(7.056,4.504)}
\gppoint{gp mark 1}{(8.438,3.683)}
\gppoint{gp mark 1}{(2.638,6.246)}
\gppoint{gp mark 1}{(11.529,2.264)}
\gppoint{gp mark 1}{(9.639,3.045)}
\gppoint{gp mark 1}{(2.014,6.383)}
\gppoint{gp mark 1}{(3.917,5.651)}
\gppoint{gp mark 1}{(9.973,2.961)}
\gppoint{gp mark 1}{(8.237,3.746)}
\gppoint{gp mark 1}{(4.590,5.431)}
\gppoint{gp mark 1}{(1.522,6.784)}
\gppoint{gp mark 1}{(10.772,2.559)}
\gppoint{gp mark 1}{(1.936,6.573)}
\gppoint{gp mark 1}{(6.402,4.585)}
\gppoint{gp mark 1}{(7.582,4.056)}
\gppoint{gp mark 1}{(11.634,2.150)}
\gppoint{gp mark 1}{(4.193,5.672)}
\gppoint{gp mark 1}{(11.327,2.309)}
\gppoint{gp mark 1}{(9.670,3.061)}
\gppoint{gp mark 1}{(1.519,6.641)}
\gppoint{gp mark 1}{(1.875,6.675)}
\gppoint{gp mark 1}{(4.801,5.493)}
\gppoint{gp mark 1}{(11.836,2.057)}
\gppoint{gp mark 1}{(2.616,6.273)}
\gppoint{gp mark 1}{(6.432,4.507)}
\gppoint{gp mark 1}{(9.948,2.969)}
\gppoint{gp mark 1}{(5.388,5.038)}
\gppoint{gp mark 1}{(3.081,6.063)}
\gppoint{gp mark 1}{(1.825,6.653)}
\gppoint{gp mark 1}{(9.413,3.141)}
\gppoint{gp mark 1}{(10.640,2.610)}
\gppoint{gp mark 1}{(7.947,3.729)}
\gppoint{gp mark 1}{(11.699,2.166)}
\gppoint{gp mark 1}{(11.467,2.216)}
\gppoint{gp mark 1}{(8.874,3.416)}
\gppoint{gp mark 1}{(7.420,4.103)}
\gppoint{gp mark 1}{(5.507,5.044)}
\gppoint{gp mark 1}{(10.789,2.529)}
\gppoint{gp mark 1}{(3.346,5.971)}
\gppoint{gp mark 1}{(11.723,2.133)}
\gppoint{gp mark 1}{(11.518,2.220)}
\gppoint{gp mark 1}{(3.821,5.734)}
\gppoint{gp mark 1}{(6.929,4.320)}
\gppoint{gp mark 1}{(2.735,6.191)}
\gppoint{gp mark 1}{(1.595,6.688)}
\gppoint{gp mark 1}{(1.752,6.650)}
\gppoint{gp mark 1}{(10.914,2.472)}
\gppoint{gp mark 1}{(3.146,6.084)}
\gppoint{gp mark 1}{(11.855,2.048)}
\gppoint{gp mark 1}{(9.469,3.229)}
\gppoint{gp mark 1}{(7.359,4.106)}
\gppoint{gp mark 1}{(5.606,4.944)}
\gppoint{gp mark 1}{(8.212,3.696)}
\gppoint{gp mark 1}{(11.316,2.320)}
\gppoint{gp mark 1}{(1.724,6.671)}
\gppoint{gp mark 1}{(2.976,6.104)}
\gppoint{gp mark 1}{(6.910,4.316)}
\gppoint{gp mark 1}{(7.748,4.007)}
\gppoint{gp mark 1}{(9.708,3.014)}
\gppoint{gp mark 1}{(11.859,2.054)}
\gppoint{gp mark 1}{(2.184,6.405)}
\gppoint{gp mark 1}{(3.376,5.995)}
\gppoint{gp mark 1}{(2.488,6.410)}
\gppoint{gp mark 1}{(11.377,2.308)}
\gppoint{gp mark 1}{(1.512,6.828)}
\gppoint{gp mark 1}{(6.495,4.465)}
\gppoint{gp mark 1}{(1.701,6.692)}
\gppoint{gp mark 1}{(8.840,3.458)}
\gppoint{gp mark 1}{(11.863,2.063)}
\gppoint{gp mark 1}{(3.197,5.902)}
\gppoint{gp mark 1}{(6.112,4.658)}
\gppoint{gp mark 1}{(5.350,4.986)}
\gppoint{gp mark 1}{(11.429,2.292)}
\gppoint{gp mark 1}{(1.568,6.762)}
\gppoint{gp mark 1}{(10.097,2.866)}
\gppoint{gp mark 1}{(9.102,3.297)}
\gppoint{gp mark 1}{(1.845,6.650)}
\gppoint{gp mark 1}{(9.720,2.921)}
\gppoint{gp mark 1}{(5.043,5.151)}
\gppoint{gp mark 1}{(4.373,5.418)}
\gppoint{gp mark 1}{(1.510,6.757)}
\gppoint{gp mark 1}{(1.562,6.727)}
\gppoint{gp mark 1}{(2.592,6.239)}
\gppoint{gp mark 1}{(11.252,2.294)}
\gppoint{gp mark 1}{(7.603,4.071)}
\gppoint{gp mark 1}{(3.767,5.789)}
\gppoint{gp mark 1}{(2.905,6.259)}
\gppoint{gp mark 1}{(9.905,2.956)}
\gppoint{gp mark 1}{(11.795,2.107)}
\gppoint{gp mark 1}{(8.951,3.413)}
\gppoint{gp mark 1}{(2.234,6.485)}
\gppoint{gp mark 1}{(3.237,5.990)}
\gppoint{gp mark 1}{(6.522,4.500)}
\gppoint{gp mark 1}{(5.824,4.850)}
\gppoint{gp mark 1}{(10.396,2.693)}
\gppoint{gp mark 1}{(1.649,6.615)}
\gppoint{gp mark 1}{(8.570,3.548)}
\gppoint{gp mark 1}{(6.869,4.381)}
\gppoint{gp mark 1}{(7.544,4.075)}
\gppoint{gp mark 1}{(10.808,2.591)}
\gppoint{gp mark 1}{(3.093,6.111)}
\gppoint{gp mark 1}{(8.199,3.706)}
\gppoint{gp mark 1}{(11.545,2.213)}
\gppoint{gp mark 1}{(2.412,6.248)}
\gppoint{gp mark 1}{(11.360,2.273)}
\gppoint{gp mark 1}{(4.289,5.502)}
\gppoint{gp mark 1}{(11.142,2.397)}
\gppoint{gp mark 1}{(5.879,4.754)}
\gpsetpointsize{2.80}
\gppoint{gp mark 12}{(3.582,6.348)}
\gppoint{gp mark 12}{(8.696,2.832)}
\gppoint{gp mark 12}{(2.115,3.013)}
\gppoint{gp mark 12}{(10.460,2.813)}
\gppoint{gp mark 12}{(1.785,2.633)}
\gppoint{gp mark 12}{(7.839,2.297)}
\gppoint{gp mark 12}{(11.109,2.336)}
\gppoint{gp mark 12}{(5.762,2.495)}
\gppoint{gp mark 12}{(2.785,2.410)}
\gppoint{gp mark 12}{(4.422,2.288)}
\gppoint{gp mark 12}{(10.179,2.263)}
\gppoint{gp mark 12}{(1.607,2.227)}
\gppoint{gp mark 12}{(2.362,2.215)}
\gppoint{gp mark 12}{(6.009,2.210)}
\gppoint{gp mark 12}{(10.680,2.166)}
\gppoint{gp mark 12}{(4.921,2.168)}
\gppoint{gp mark 12}{(9.511,2.102)}
\gppoint{gp mark 12}{(11.653,2.211)}
\gppoint{gp mark 12}{(6.153,2.200)}
\gppoint{gp mark 12}{(4.142,2.161)}
\gppoint{gp mark 12}{(1.557,2.128)}
\gppoint{gp mark 12}{(11.712,2.088)}
\gppoint{gp mark 12}{(9.039,2.146)}
\gppoint{gp mark 12}{(11.203,2.140)}
\gppoint{gp mark 12}{(1.544,2.124)}
\gppoint{gp mark 12}{(8.031,2.149)}
\gppoint{gp mark 12}{(4.542,2.170)}
\gppoint{gp mark 12}{(2.428,2.139)}
\gppoint{gp mark 12}{(11.751,2.184)}
\gppoint{gp mark 12}{(7.113,2.105)}
\gppoint{gp mark 12}{(5.470,2.110)}
\gppoint{gp mark 12}{(4.001,2.126)}
\gppoint{gp mark 12}{(2.248,2.085)}
\gppoint{gp mark 12}{(10.009,2.112)}
\gppoint{gp mark 12}{(2.869,2.080)}
\gppoint{gp mark 12}{(9.222,2.097)}
\gppoint{gp mark 12}{(10.951,2.060)}
\gppoint{gp mark 12}{(1.530,2.088)}
\gppoint{gp mark 12}{(1.733,2.096)}
\gppoint{gp mark 12}{(7.056,2.080)}
\gppoint{gp mark 12}{(8.438,2.108)}
\gppoint{gp mark 12}{(2.638,2.080)}
\gppoint{gp mark 12}{(11.529,2.109)}
\gppoint{gp mark 12}{(9.639,2.080)}
\gppoint{gp mark 12}{(2.014,2.096)}
\gppoint{gp mark 12}{(3.917,2.088)}
\gppoint{gp mark 12}{(9.973,2.098)}
\gppoint{gp mark 12}{(8.237,2.079)}
\gppoint{gp mark 12}{(4.590,2.087)}
\gppoint{gp mark 12}{(1.522,2.095)}
\gppoint{gp mark 12}{(10.772,2.066)}
\gppoint{gp mark 12}{(1.936,2.074)}
\gppoint{gp mark 12}{(6.402,2.095)}
\gppoint{gp mark 12}{(7.582,2.076)}
\gppoint{gp mark 12}{(11.634,2.071)}
\gppoint{gp mark 12}{(4.193,2.078)}
\gppoint{gp mark 12}{(11.327,2.092)}
\gppoint{gp mark 12}{(9.670,2.087)}
\gppoint{gp mark 12}{(1.519,2.073)}
\gppoint{gp mark 12}{(1.875,2.075)}
\gppoint{gp mark 12}{(4.801,2.063)}
\gppoint{gp mark 12}{(11.836,2.049)}
\gppoint{gp mark 12}{(2.616,2.077)}
\gppoint{gp mark 12}{(6.432,2.069)}
\gppoint{gp mark 12}{(9.948,2.069)}
\gppoint{gp mark 12}{(5.388,2.087)}
\gppoint{gp mark 12}{(3.081,2.057)}
\gppoint{gp mark 12}{(1.825,2.074)}
\gppoint{gp mark 12}{(9.413,2.076)}
\gppoint{gp mark 12}{(10.640,2.058)}
\gppoint{gp mark 12}{(7.947,2.056)}
\gppoint{gp mark 12}{(11.699,2.070)}
\gppoint{gp mark 12}{(11.467,2.074)}
\gppoint{gp mark 12}{(8.874,2.061)}
\gppoint{gp mark 12}{(7.420,2.067)}
\gppoint{gp mark 12}{(5.507,2.070)}
\gppoint{gp mark 12}{(10.789,2.080)}
\gppoint{gp mark 12}{(3.346,2.065)}
\gppoint{gp mark 12}{(11.723,2.064)}
\gppoint{gp mark 12}{(11.518,2.058)}
\gppoint{gp mark 12}{(3.821,2.063)}
\gppoint{gp mark 12}{(6.929,2.060)}
\gppoint{gp mark 12}{(2.735,2.056)}
\gppoint{gp mark 12}{(1.595,2.070)}
\gppoint{gp mark 12}{(1.752,2.068)}
\gppoint{gp mark 12}{(10.914,2.059)}
\gppoint{gp mark 12}{(3.146,2.060)}
\gppoint{gp mark 12}{(11.855,2.077)}
\gppoint{gp mark 12}{(9.469,2.058)}
\gppoint{gp mark 12}{(7.359,2.068)}
\gppoint{gp mark 12}{(5.606,2.070)}
\gppoint{gp mark 12}{(8.212,2.058)}
\gppoint{gp mark 12}{(11.316,2.063)}
\gppoint{gp mark 12}{(1.724,2.057)}
\gppoint{gp mark 12}{(2.976,2.062)}
\gppoint{gp mark 12}{(6.910,2.064)}
\gppoint{gp mark 12}{(7.748,2.050)}
\gppoint{gp mark 12}{(9.708,2.061)}
\gppoint{gp mark 12}{(11.859,2.046)}
\gppoint{gp mark 12}{(2.184,2.056)}
\gppoint{gp mark 12}{(3.376,2.060)}
\gppoint{gp mark 12}{(2.488,2.063)}
\gppoint{gp mark 12}{(11.377,2.062)}
\gppoint{gp mark 12}{(1.512,2.059)}
\gppoint{gp mark 12}{(6.495,2.060)}
\gppoint{gp mark 12}{(1.701,2.065)}
\gppoint{gp mark 12}{(8.840,2.064)}
\gppoint{gp mark 12}{(11.863,2.064)}
\gppoint{gp mark 12}{(3.197,2.054)}
\gppoint{gp mark 12}{(6.112,2.058)}
\gppoint{gp mark 12}{(5.350,2.053)}
\gppoint{gp mark 12}{(11.429,2.052)}
\gppoint{gp mark 12}{(1.568,2.058)}
\gppoint{gp mark 12}{(10.097,2.055)}
\gppoint{gp mark 12}{(9.102,2.059)}
\gppoint{gp mark 12}{(1.845,2.060)}
\gppoint{gp mark 12}{(9.720,2.047)}
\gppoint{gp mark 12}{(5.043,2.052)}
\gppoint{gp mark 12}{(4.373,2.057)}
\gppoint{gp mark 12}{(1.510,2.053)}
\gppoint{gp mark 12}{(1.562,2.060)}
\gppoint{gp mark 12}{(2.592,2.058)}
\gppoint{gp mark 12}{(11.252,2.047)}
\gppoint{gp mark 12}{(7.603,2.050)}
\gppoint{gp mark 12}{(3.767,2.053)}
\gppoint{gp mark 12}{(2.905,2.045)}
\gppoint{gp mark 12}{(9.905,2.067)}
\gppoint{gp mark 12}{(11.795,2.063)}
\gppoint{gp mark 12}{(8.951,2.043)}
\gppoint{gp mark 12}{(2.234,2.055)}
\gppoint{gp mark 12}{(3.237,2.061)}
\gppoint{gp mark 12}{(6.522,2.057)}
\gppoint{gp mark 12}{(5.824,2.057)}
\gppoint{gp mark 12}{(10.396,2.054)}
\gppoint{gp mark 12}{(1.649,2.054)}
\gppoint{gp mark 12}{(8.570,2.057)}
\gppoint{gp mark 12}{(6.869,2.043)}
\gppoint{gp mark 12}{(7.544,2.060)}
\gppoint{gp mark 12}{(10.808,2.056)}
\gppoint{gp mark 12}{(3.093,2.050)}
\gppoint{gp mark 12}{(8.199,2.054)}
\gppoint{gp mark 12}{(11.545,2.058)}
\gppoint{gp mark 12}{(2.412,2.064)}
\gppoint{gp mark 12}{(11.360,2.046)}
\gppoint{gp mark 12}{(4.289,2.049)}
\gppoint{gp mark 12}{(11.142,2.054)}
\gppoint{gp mark 12}{(5.879,2.061)}
\gpsetpointsize{2.00}
\gppoint{gp mark 7}{(10.460,8.456)}
\gppoint{gp mark 7}{(1.785,6.569)}
\gppoint{gp mark 7}{(7.839,8.160)}
\gppoint{gp mark 7}{(11.109,8.236)}
\gppoint{gp mark 7}{(5.762,6.895)}
\gppoint{gp mark 7}{(4.422,8.409)}
\gppoint{gp mark 7}{(10.179,7.303)}
\gppoint{gp mark 7}{(1.607,8.612)}
\gppoint{gp mark 7}{(2.362,8.417)}
\gppoint{gp mark 7}{(6.009,8.426)}
\gppoint{gp mark 7}{(10.680,7.914)}
\gppoint{gp mark 7}{(4.921,8.221)}
\gppoint{gp mark 7}{(9.511,7.521)}
\gppoint{gp mark 7}{(11.653,7.804)}
\gppoint{gp mark 7}{(6.153,8.376)}
\gppoint{gp mark 7}{(4.142,8.400)}
\gppoint{gp mark 7}{(1.557,7.350)}
\gppoint{gp mark 7}{(11.712,7.814)}
\gppoint{gp mark 7}{(9.039,7.475)}
\gppoint{gp mark 7}{(11.203,8.038)}
\gppoint{gp mark 7}{(1.544,8.479)}
\gppoint{gp mark 7}{(8.031,8.013)}
\gppoint{gp mark 7}{(4.542,8.316)}
\gppoint{gp mark 7}{(2.428,7.155)}
\gppoint{gp mark 7}{(11.751,7.775)}
\gppoint{gp mark 7}{(7.113,7.469)}
\gppoint{gp mark 7}{(5.470,7.423)}
\gppoint{gp mark 7}{(4.001,8.190)}
\gppoint{gp mark 7}{(2.248,8.122)}
\gppoint{gp mark 7}{(10.009,7.889)}
\gppoint{gp mark 7}{(2.869,8.092)}
\gppoint{gp mark 7}{(9.222,7.942)}
\gppoint{gp mark 7}{(10.951,7.942)}
\gppoint{gp mark 7}{(1.530,7.338)}
\gppoint{gp mark 7}{(1.733,8.213)}
\gppoint{gp mark 7}{(7.056,7.412)}
\gppoint{gp mark 7}{(8.438,7.587)}
\gppoint{gp mark 7}{(2.638,8.183)}
\gppoint{gp mark 7}{(11.529,7.779)}
\gppoint{gp mark 7}{(9.639,7.904)}
\gppoint{gp mark 7}{(2.014,7.280)}
\gppoint{gp mark 7}{(3.917,7.375)}
\gppoint{gp mark 7}{(9.973,7.918)}
\gppoint{gp mark 7}{(8.237,7.532)}
\gppoint{gp mark 7}{(4.590,7.450)}
\gppoint{gp mark 7}{(1.522,8.077)}
\gppoint{gp mark 7}{(10.772,7.548)}
\gppoint{gp mark 7}{(1.936,8.034)}
\gppoint{gp mark 7}{(6.402,8.020)}
\gppoint{gp mark 7}{(7.582,7.399)}
\gppoint{gp mark 7}{(11.634,7.758)}
\gppoint{gp mark 7}{(4.193,7.421)}
\gppoint{gp mark 7}{(11.327,7.663)}
\gppoint{gp mark 7}{(9.670,7.540)}
\gppoint{gp mark 7}{(1.519,8.008)}
\gppoint{gp mark 7}{(1.875,8.129)}
\gppoint{gp mark 7}{(4.801,7.485)}
\gppoint{gp mark 7}{(11.836,7.744)}
\gppoint{gp mark 7}{(2.616,8.044)}
\gppoint{gp mark 7}{(6.432,7.448)}
\gppoint{gp mark 7}{(9.948,7.635)}
\gppoint{gp mark 7}{(5.388,7.412)}
\gppoint{gp mark 7}{(3.081,8.005)}
\gppoint{gp mark 7}{(1.825,7.433)}
\gppoint{gp mark 7}{(9.413,7.533)}
\gppoint{gp mark 7}{(10.640,7.813)}
\gppoint{gp mark 7}{(7.947,7.907)}
\gppoint{gp mark 7}{(11.699,7.669)}
\gppoint{gp mark 7}{(11.467,7.687)}
\gppoint{gp mark 7}{(8.874,7.929)}
\gppoint{gp mark 7}{(7.420,7.536)}
\gppoint{gp mark 7}{(5.507,8.056)}
\gppoint{gp mark 7}{(10.789,7.812)}
\gppoint{gp mark 7}{(3.346,7.963)}
\gppoint{gp mark 7}{(11.723,7.715)}
\gppoint{gp mark 7}{(11.518,7.723)}
\gppoint{gp mark 7}{(3.821,7.457)}
\gppoint{gp mark 7}{(6.929,7.899)}
\gppoint{gp mark 7}{(2.735,7.461)}
\gppoint{gp mark 7}{(1.595,7.445)}
\gppoint{gp mark 7}{(1.752,8.048)}
\gppoint{gp mark 7}{(10.914,7.801)}
\gppoint{gp mark 7}{(3.146,8.009)}
\gppoint{gp mark 7}{(11.855,7.700)}
\gppoint{gp mark 7}{(9.469,7.661)}
\gppoint{gp mark 7}{(7.359,7.504)}
\gppoint{gp mark 7}{(5.606,7.974)}
\gppoint{gp mark 7}{(8.212,7.955)}
\gppoint{gp mark 7}{(11.316,7.743)}
\gppoint{gp mark 7}{(1.724,7.452)}
\gppoint{gp mark 7}{(2.976,7.489)}
\gppoint{gp mark 7}{(6.910,7.889)}
\gppoint{gp mark 7}{(7.748,7.629)}
\gppoint{gp mark 7}{(9.708,7.566)}
\gppoint{gp mark 7}{(11.859,7.733)}
\gppoint{gp mark 7}{(2.184,7.975)}
\gppoint{gp mark 7}{(3.376,8.035)}
\gppoint{gp mark 7}{(2.488,7.962)}
\gppoint{gp mark 7}{(11.377,7.656)}
\gppoint{gp mark 7}{(1.512,7.954)}
\gppoint{gp mark 7}{(6.495,7.896)}
\gppoint{gp mark 7}{(1.701,7.997)}
\gppoint{gp mark 7}{(8.840,7.880)}
\gppoint{gp mark 7}{(11.863,7.750)}
\gppoint{gp mark 7}{(3.197,7.920)}
\gppoint{gp mark 7}{(6.112,7.871)}
\gppoint{gp mark 7}{(5.350,7.971)}
\gppoint{gp mark 7}{(11.429,7.689)}
\gppoint{gp mark 7}{(1.568,7.534)}
\gppoint{gp mark 7}{(10.097,7.880)}
\gppoint{gp mark 7}{(9.102,7.838)}
\gppoint{gp mark 7}{(1.845,7.533)}
\gppoint{gp mark 7}{(9.720,7.822)}
\gppoint{gp mark 7}{(5.043,7.505)}
\gppoint{gp mark 7}{(4.373,7.918)}
\gppoint{gp mark 7}{(1.510,7.994)}
\gppoint{gp mark 7}{(1.562,7.969)}
\gppoint{gp mark 7}{(2.592,7.529)}
\gppoint{gp mark 7}{(11.252,7.802)}
\gppoint{gp mark 7}{(7.603,7.856)}
\gppoint{gp mark 7}{(3.767,7.488)}
\gppoint{gp mark 7}{(2.905,7.548)}
\gppoint{gp mark 7}{(9.905,7.848)}
\gppoint{gp mark 7}{(11.795,7.680)}
\gppoint{gp mark 7}{(8.951,7.846)}
\gppoint{gp mark 7}{(2.234,7.960)}
\gppoint{gp mark 7}{(3.237,7.496)}
\gppoint{gp mark 7}{(6.522,7.588)}
\gppoint{gp mark 7}{(5.824,7.511)}
\gppoint{gp mark 7}{(10.396,7.599)}
\gppoint{gp mark 7}{(1.649,7.468)}
\gppoint{gp mark 7}{(8.570,7.876)}
\gppoint{gp mark 7}{(6.869,7.613)}
\gppoint{gp mark 7}{(7.544,7.841)}
\gppoint{gp mark 7}{(10.808,7.654)}
\gppoint{gp mark 7}{(3.093,7.575)}
\gppoint{gp mark 7}{(8.199,7.852)}
\gppoint{gp mark 7}{(11.545,7.756)}
\gppoint{gp mark 7}{(2.412,7.520)}
\gppoint{gp mark 7}{(11.360,7.749)}
\gppoint{gp mark 7}{(4.289,7.595)}
\gppoint{gp mark 7}{(11.142,7.650)}
\gppoint{gp mark 7}{(5.879,7.897)}
\gpcolor{color=gp lt color border}
\gpsetlinetype{gp lt border}
\draw[gp path] (1.504,9.029)--(1.504,0.985)--(11.893,0.985)--(11.893,9.029)--cycle;
\gpdefrectangularnode{gp plot 1}{\pgfpoint{1.504cm}{0.985cm}}{\pgfpoint{11.893cm}{9.029cm}}
\end{tikzpicture}
\end{center}
\end{figure}
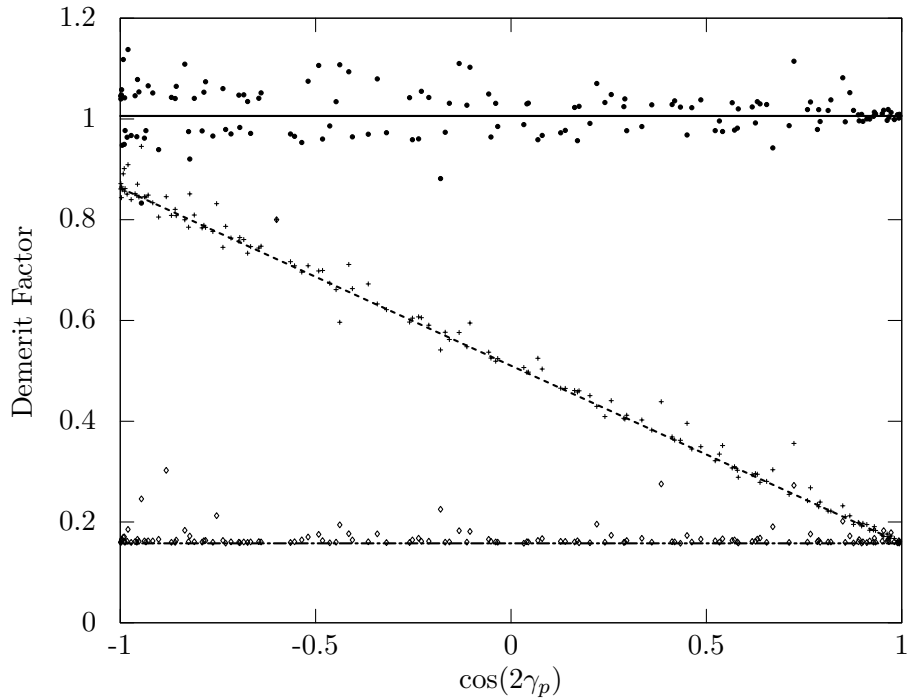
\end{center}

In summary, the sequences $\fpslu$ and $\gpslu$ derived from quartic characters furnish binary sequences with autocorrelation and crosscorrelation properties far superior to randomly selected sequences.  By a careful selection of primes, one can obtain the same asymptotic autocorrelation performance as one can obtain with the modified Legendre sequences, $\hpslu$, which furnish the highest known asymptotic autocorrelation merit factor for binary sequences ($6.342061\ldots$).  However, there is little reason to use $\fpslu$ and $\gpslu$ for applications that solely depend upon autocorrelation, as for most primes they will fall short of the performance of the modified Legendre sequences.
When crosscorrelation performance is important, then pairs $(\fpslu,\gpslu)$ become very interesting, since for most primes they have significantly lower mutual crosscorrelation than random sequences.  We see that there is a tradeoff between autocorrelation and crosscorrelation performance, but both can simultaneously be made considerably better than that of random sequences.
One can also crosscorrelate modified Legendre sequences $\hpslu$ with $\fpslu$ and $\gpslu$ and obtain crosscorrelation performance on par with random sequences while maintaining high autocorrelation performance.

We should remark that the autocorrelation behavior of $\fpslu$ and $\gpslu$ was independently discovered by G\"unther and Schmidt and reported in \cite{Gunther-Schmidt-2015-Merit}.  They obtain the same formulae for asymptotic autocorrelation merit factor as presented here in Theorem \ref{Anne}.  They focus on showing that $\fpslu$ and $\gpslu$ can (by careful choice of primes) obtain the same record asymptotic merit factor of $6.342061\ldots$ that has already been obtained in \cite{Jedwab-Katz-Schmidt-2013-Littlewood, Katz-2013-Asymptotic, Jedwab-Katz-Schmidt-2013-Advances} with the modified Legendre sequences $\hpslu$.  In their paper, they do not discuss crosscorrelation, which is the main concern here, in view of the discussion in the previous paragraph.

The sequences discussed in this Introduction are some of the the simplest examples that can be constructed using linear combinations of multiplicative characters.  The rest of this paper provides proofs of the asymptotic autocorrelation and crosscorrelation merit factors like \eqref{Janice} and \eqref{Helen} discussed in this Introduction (these occur in Section \ref{Gilbert}).  But our study goes much further: it provides formulae for a very general class of sequence constructions using linear combinations of multiplicative characters of finite fields.  Thus the proofs in Section \ref{Gilbert} of the facts adduced in this Introduction are simply an application of much more general results.  We hope these general formulae will provide tools that will enable researchers to design many interesting sequences and sequence families with superior autocorrelation and crosscorrelation performance.

We now summarize the rest of the paper.  Section \ref{Henry} introduces finite field characters and Section \ref{Karl} uses these to define {\it character combination sequences}, which are formed using linear combinations of multiplicative finite field characters.  Section \ref{Alan} has our main theorem that furnishes the asymptotic crosscorrelation merit factors for families of pairs of character combination sequences.  As a corollary, we also provide the asymptotic autocorrelation merit factors for families of character combination sequences.  Section \ref{Percy} shows that some terms in our asymptotic merit factor formulae are closely related to the {\it periodic} correlation properties of our sequences.  Section \ref{Ophelia} shows that the maximum asymptotic autocorrelation merit factor achievable with character combination sequences is $6.342061\ldots$, the highest value currently known for binary sequences.  Section \ref{Ophelia} also gives necessary and sufficient conditions for reaching the maximum value.  The connection that was made between periodic correlation and aperiodic correlation in Section \ref{Percy} now manifests itself as a principle of importance in Section \ref{Ophelia}: lower periodic correlation leads to lower aperiodic correlation for character combination sequences.  In her 1967 paper \cite[p.~157]{Boehmer-1967-Binary}, Boehmer states
\begin{quotation}
In fact, a necessary, but not sufficient, condition for low pulse compression code sidelobes is that the periodic sidelobes be low.  It was then hoped that at least some of the cyclic permutations of those good periodic codes would also be good pulse compression codes.  Many were.
\end{quotation}
Our results in Section \ref{Ophelia} provide further vindication of Boehmer's approach.

Section \ref{Emily} discusses a special case of character combination sequences: binary sequences derived from $2 m$th order residues in finite fields.  When $m=1$, we obtain Legendre sequences and their modifications $\hpslu$.  When $m=2$, we obtain the sequences $\fpslu$ and $\gpslu$ derived from quartic characters.  Higher values of $m$ can also be used to produce further examples not discussed here.  Section \ref{Gilbert} applies the results of Sections \ref{Alan}--\ref{Emily} to the sequences $\hpslu$, $\fpslu$, and $\gpslu$ to provide the results discussed in this Introduction.  Finally, Section \ref{James} demonstrates a construction (based on the sequences defined in this Introduction) which shows that there is no upper bound on the crosscorrelation merit factor of pairs of sequences derived from linear combinations of characters, that is, the crosscorrelation demerit factor can asymptotically approach $0$.

\section{Characters and Gauss Sums}\label{Henry}

In this section, we let $p$ be a prime, let $\Fp$ be the finite field of order $p$, and let $\Fpu$ denote the multiplicative group of nonzero elements of $\Fp$.
A {\it multiplicative character} of $\Fp$ is a group homomorphism from $\Fpu$ into the multiplicative group $\C^*$ of nonzero complex numbers.
We let $\mchars$ denote the set of multiplicative characters of $\Fp$, which is a group whose group operation is multiplication of functions: if $\chi,\psi\in\mchars$, then $(\chi\psi)(a)=\chi(a)\psi(a)$.
We write powers of characters $\chi^k$, including the negative powers, and $\chi^{-1}$ denotes the character with $\chi^{-1}(a)=1/\chi(a)$, not the inverse function.  We also write $\conj{\chi}$ for $\chi^{-1}$, and call it the {\it conjugate of $\chi$}.

The group $\mchars$ is isomorphic to $\Fpu$, and so is a cyclic group of order $p-1$.
Its identity element is the {\it trivial multiplicative character}, which maps every element of $\Fpu$ to $1$.
Let $\phi$ be the Euler phi function.
If we let $\omega$ be a generator of $\mchars$, we see that for each $m$ such that $m \mid p-1$, the group $\mchars$ has $\phi(m)$ characters of order $m$, namely the characters $\omega^{(p-1) k/m}$ with $0 \leq k < m$ and $\gcd(k,m)=1$.
For instance, if $p$ is odd, then there is a unique character $\omega^{(p-1)/2}$ of order $2$ called the {\it quadratic character} or {\it Legendre symbol}.
Or if $p \equiv 1 \pmod{4}$, then there are two {\it quartic characters}, $\omega^{(p-1)/4}$ and its conjugate $\omega^{3(p-1)/4}=\omega^{-(p-1)/4}$.

We extend every multiplicative character $\chi$ to be defined on all of $\Fp$ by setting $\chi(0)=0$.
We even extend the trivial multiplicative character so that it maps $0$ to $0$.

An {\it additive character} of $\Fp$ is a group homomorphism from the group $\Fp$ (with $+$ as the group operation) into the multiplicative group $\C^*$ of nonzero complex numbers.
If we let $\zeta_p=\exp(2\pi i/p)$, then the {\it canonical additive character}, which we write as $\epsilon\colon \Fp \to \C$, is defined by $\epsilon(x)=\zeta_p^x$.
For $a \in \Fp$, we let $\epsilon_a\colon \Fp \to \C$ be defined by $\epsilon_a(x)=\epsilon(a x)$: this is also an additive character, and in fact the set $\achars$ of all additive characters is just $\{\epsilon_a \colon a \in \Fp\}$.
Furthermore $\epsilon_a(x) \epsilon_b(x)=\epsilon_{a+b}(x)$, so that multiplication of characters makes $\achars$ a group that is isomorphic to $\Fp$ itself.
Note that $\epsilon_0$ is the identity element of this group: it maps every element of $\Fp$ to $1$, and so is called the {\it trivial additive character}.
Two relations that will be useful are the {\it orthogonality relations} for additive and multiplicative characters (see \cite[eq.~(5.3)]{Lidl-Niederreiter-1997-Finite}), which state that for $a \in \Fp$, we have
\begin{equation}\label{Oliver}
\sum_{b \in \Fp} \epsilon_a(b) = \begin{cases}
p & \text{if $a=0$}, \\
0 & \text{if $a\not=0$}
\end{cases}
\end{equation}
and for $\chi \in \mchars$, we have
\begin{equation}\label{Otto}
\sum_{b \in \Fpu} \chi(b) = \begin{cases}
p-1 & \text{if $\chi$ is the trivial character}, \\
0 & \text{if $\chi$ is a nontrivial character}.
\end{cases}
\end{equation}

If $\epsilon_a \in \achars$ and $\chi \in \mchars$, then the {\it Gauss sum associated to $\epsilon_a$ and $\chi$} is
\begin{equation}\label{Theodore}
\tau_a(\chi) = \sum_{x \in \Fpu} \epsilon_a(x) \chi(x).
\end{equation}
For $\chi \in \mchars$, we let
\[
\tau(\chi)=\tau_1(\chi).
\]
We can reduce general Gauss sums to these using \cite[Lemma 6(v)]{Katz-2016-Aperiodic}, which states that
\begin{equation}\label{George}
\tau_a(\chi) = \begin{cases}
p-1 & \text{if $a=0$ and $\chi$ is trivial}, \\
\conj{\chi}(a) \tau(\chi) & \text{otherwise}. \\
\end{cases}
\end{equation}
We also have by \cite[Lemma 6(vi)]{Katz-2016-Aperiodic} that
\begin{equation}\label{Katherine}
\conj{\tau(\chi)} = \chi(-1) \tau(\conj{\chi})
\end{equation}
for any $\chi \in \mchars$, and \cite[Lemma 6(iii),(iv)]{Katz-2016-Aperiodic} shows that
\begin{equation}\label{Manuel}
|\tau(\chi)| = 
\begin{cases}
\sqrt{p} & \text{if $\chi$ is nontrivial} \\
1 & \text{if $\chi$ is trivial}.
\end{cases}
\end{equation}
The Gauss sums serve as Fourier coefficients for the expressing multiplicative characters as linear combinations of additive characters, namely, if $\chi \in \mchars$, then \cite[Lemma 8]{Katz-2016-Aperiodic} tells us that
\begin{equation}\label{Genevieve}
\chi(a) = \frac{1}{p} \sum_{b \in \Fp} \tau_b(\chi) \conj{\epsilon_b}(a).
\end{equation}

\section{Character Combination Sequences}\label{Karl}

In this section we define the types of sequences, called {\it character combination sequences}, whose autocorrelation and crosscorrelation properties we investigate.  These give rise to the sequences $\fpslu$, $\gpslu$, and $\hpslu$ discussed in the Introduction.

We shall construct a sequence $f$ using a linear combination of the characters $\chi \in \mchars$.
The coefficients of the linear combination will be a family $\{f_\chi\}_{\chi \in \mchars}$ of complex numbers (typically only a few $f_\chi$ will be nonzero).
We shall always demand that $\sum_{\chi\in\mchars} |f_\chi|^2=1$ (a normalization condition), and if $\chi_0$ is the trivial multiplicative character, then $f_{\chi_0}=0$ (which makes the sequence ``balanced'' in the sense that the sum of the terms of the sequence will be zero).
This linear combination defines a function $F\colon \Fp \to \C$ with
\begin{equation}\label{Francis}
F(a)=\sum_{\chi\in\mchars} f_\chi \chi(a)
\end{equation}
that will be used to compute the terms of $f$.
We also associate to $f$ a {\it shift} $s$ and {\it length} $\ell$ and let
\[
f=(F(s),\ldots,F(s+\ell-1)),
\]
where $F(j)$ with $j \in \Z$ is interpreted by reducing $j$ modulo $p$ to get an element of $\Fp$, that is, $F(j)$ is interpreted as $F(j \!\pmod{p})$, whose value is given in \eqref{Francis} above.  Thus $f$ is a sequence of $\ell$ complex numbers, called the {\it character combination sequence with prime $p$, field $\Fp$, character combination $\{f_\chi\}_{\chi\in\mchars}$, shift $s$, and length $\ell$}.

As a result of the way we extend our characters, we note that $F(j)=0$ whenever $p\mid j$, so that $f$ may have some zero terms.
For typical applications we choose our linear combinations of characters so that $F(j)$ is a unimodular complex number (that is, of magnitude $1$) for every $j$ with $p\nmid j$.
Such character combinations are called {\it unimodularizable}, and the sequences they produce are called {\it unimodularizable character combination sequences}.
Of course, $F(j)=0$ when $p\mid j$, so some terms of a unimodularizable character combination sequences may not be unimodular.
In applications, it is often important that {\it every} term of our sequence $f$ be of complex magnitude $1$.
In this case, we replace the any term of $f$ of the form $F(j)=0$ (which occurs when $j\mid p$) with a term of complex magnitude $1$ (typically the value $1$ is used).
Such a modification is called a {\it unimodularization} of $f$ and produces a truly {\it unimodular sequence}, that is, a sequence whose terms are all of magnitude $1$.
For mathematical convenience, we shall analyze the sequences and arrays that retain the zero entries, and then demonstrate that replacing these zeroes with unimodular complex numbers does not change any of our asymptotic results.

We now see how the sequences $\fpslu$, $\gpslu$, and $\hpslu$ defined in \eqref{Elizabeth} in the Introduction fit into this formalism.
We see that $\hpslu$ is a unimodularization (setting all zero entries to $1$) of the character combination sequence with odd prime $p$, field $\Fp$, shift $s$, length $\ell$, and character combination $\{h_\chi\}_{\chi\in\mchars}$ where $h_\chi=0$ for every character $\chi$ except the quadratic character $\eta$, for which $h_\eta=1$.
We see that $\fpslu$ (resp., $\gpslu$) is a unimodularization (setting all zero entries to $1$) of the character combination sequence with prime $p \equiv 1 \pmod{4}$, field $\Fp$, shift $s$, length $\ell$, and character combination $\{f_\chi\}_{\chi\in\mchars}$ (resp., $\{g_\chi\}_{\chi \in \mchars}$) where $f_\chi=0$ (resp., $g_\chi=0$) for every character $\chi$ except the two quartic characters $\theta_p$ and $\conj{\theta_p}$, for which $f_{\theta_p}=(1-i)/2$ and $f_{\conj{\theta_p}}=(1+i)/2$ (resp., $g_{\theta_p}=(1+i)/2$ and $g_{\conj{\theta_p}}=(1-i)/2$).  In the Introduction, we had defined $\theta_p$ relative to the choice of a primitive element $\alpha_p$ of $\Fpu$, namely, by setting $\theta_p(\alpha_p)=i$.  Obviously changing our choice of $\alpha_p$ can exchange the labels $\theta_p$ and $\conj{\theta_p}$ for our characters, and thus exchange the labels $\fpslu$ and $\gpslu$ of our sequences.

In this paper we study the crosscorrelation of pairs $(f,g)$ of character combination sequences with $f$ and $g$ having the same prime $p$, field $\Fp$, and length $\ell$.  The character combination and shift can be different (when they are identical, we are studying autocorrelation).
Associated with every such sequence pair $(f,g)$ is a set of parameters:
\begin{align}
S_{f,g} & = \sums{\phi,\chi, \psi,\omega \in \mchars \\ \phi\chi=\psi\omega \\ \phi\not\in\{\conj{\chi},\psi,\omega\}} f_\phi g_\chi \conj{f_\psi g_\omega} \frac{\tau(\phi)\tau(\chi)\conj{\tau(\psi)\tau(\omega)}}{p^2} \nonumber \\
& \qquad - \sums{\phi \in \mchars} |f_{\phi} g_{\phi}|^2 - \sum_{\phi\in\mchars} |f_\phi g_{\conj{\phi}}|^2 - \sum_{\phi\in\mchars} f_{\phi} \conj{f_{\conj{\phi}}} g_{\conj{\phi}} \conj{g_{\phi}} + |f_\eta g_\eta|^2, \nonumber \\
U_{f,g} & = \left|\sum_{\phi \in \mchars} f_\phi \conj{g_\phi} \right|^2 \label{Patrick} \\ 
V_{f,g} & = \left|\sum_{\phi \in \mchars} f_\phi g_{\conj{\phi}} \,\, \phi(-1)\right|^2 \nonumber \\
W_f & = \sum_{\phi \in \mchars} |f_\phi| \nonumber \\
W_g & = \sum_{\phi \in \mchars} |g_\phi| \nonumber,
\end{align}
where $\tau$ indicates a Gauss sum (see Section \ref{Henry}) and $\eta$ denotes the quadratic character of $\Fp$.
Recall that our sequences are normalized so that $\sum_{\chi \in \mchars} |f_\chi|^2=\sum_{\chi\in\mchars} |g_\chi|^2=1$.
By the Cauchy-Schwarz inequality this means that
\begin{align}
0 & \leq U_{f,g},  V_{f,g} \leq 1 \label{Barbara} \\
0 & < W_f, W_g. \nonumber
\end{align}

Also for each single sequence $f$, we have the following parameters, which are just the parameters of the pair $(f,f)$, and which are useful for studying autocorrelation:
\begin{align}
S_{f,f} & = \sums{\phi,\chi, \psi,\omega \in \mchars \\ \phi\chi=\psi\omega \\ \phi\not\in\{\conj{\chi},\psi,\omega\}} f_\phi f_\chi \conj{f_\psi f_\omega} \frac{\tau(\phi)\tau(\chi)\conj{\tau(\psi)\tau(\omega)}}{p^2} \nonumber \\
& \qquad\qquad - 2 \sum_{\phi\in\mchars} |f_{\phi} f_{\conj{\phi}}|^2 - \sums{\phi \in \mchars} |f_{\phi}|^4 + |f_\eta|^4, \nonumber \\
U_{f,f} & = 1 \label{Albert} \\
V_{f,f} & = \left|\sum_{\phi \in \mchars} f_\phi f_{\conj{\phi}} \,\, \phi(-1)\right|^2 \nonumber \\
W_f & = \sum_{\phi \in \mchars} |f_\phi| \nonumber.
\end{align}
These parameters become important in our asymptotic calculations in the next section.

\section{Asymptotic Calculations}\label{Alan}

In this section we prove our main result (Theorem \ref{Mary}) on the asymptotic crosscorrelation merit factor for pairs of character combination sequences.  Since the autocorrelation of a sequence $f$ is just the crosscorrelation of $f$ with itself, we then obtain the asymptotic autocorrelation merit factor as Corollary \ref{Ellen}.
Our calculations are expressed in terms of the function
\[
\Omega(x,y) = \sum_{n \in \Z} \max(0,1-|n x - y|)^2,
\]
which is defined and continuous on $\{(x,y) \in \R^2: x\not=0\}$.
Although the sum in $\Omega$ appears infinite, it is locally finite (i.e., for a given $(x,y)$ in the domain, only finitely many summands are nonzero).

\begin{theorem}\label{Mary}
Let $\{(f_\iota,g_\iota)\}_{\iota \in I}$ be a family of pairs of unimodularizable character combination sequences, where for each $\iota \in I$, there is a prime $p_\iota$, field $F_\iota$ of order $p_\iota$, and length $\ell_\iota$ so that both $f_\iota$ and $g_\iota$ have these as their prime, field, and length, and suppose that $r_\iota$ and $s_\iota$ are the respective shifts of $f_\iota$ and $g_\iota$.
Suppose that $\{p_\iota\}_{\iota \in I}$ is infinite and that $\ell_\iota/p_\iota$ tends to a positive real number $\Lambda$ as $p_\iota \to \infty$.
Of the parameters defined in \eqref{Patrick}, suppose that $S_{f_\iota,g_\iota}$, $U_{f_\iota,g_\iota}$, and $V_{f_\iota,g_\iota}$ tend to real limits $S$, $U$, and $V$, respectively, as $p_\iota \to \infty$ and that $W_{f_\iota}^2 W_{g_\iota}^2 (\log p_\iota)^3/\sqrt{p_\iota} \to 0$ as $p_\iota \to \infty$.
If $U \not=0$, then suppose that $(r_\iota-s_\iota)/p_\iota$ tends to a real limit $\Delta$ as $p_\iota \to \infty$.
If $V \not=0$, then suppose that $(r_\iota+s_\iota)/p_\iota$ tends to a real limit $\Sigma$ as $p_\iota \to \infty$.
Then
\[
\CDF(f_\iota,g_\iota) \to S \cdot \frac{2}{3}\Lambda + \Omega\left(\frac{1}{\Lambda},0\right) + U \cdot \Omega\left(\frac{1}{\Lambda},\frac{\Delta}{\Lambda}\right) + V \cdot \Omega\left(\frac{1}{\Lambda},1+\frac{\Sigma}{\Lambda}\right)
\]
as $p_\iota \to \infty$.
If $f^u_\iota$ and $g^u_\iota$ are respective unimodularizations of $f_\iota$ and $g_\iota$ for each $\iota \in I$, then $\CDF(f^u_\iota,g^u_\iota)$ has the same limit as $\CDF(f_\iota,g_\iota)$.
\end{theorem}
\begin{proof}
Let $(f,g)$ be a pair of sequences in our family both having prime $p$, field $\Fp$, length $\ell$, and where $r$ and $s$ are the respective shifts of $f$ and $g$.  Let $\{f_\chi\}_{\chi\in \mchars}$ and $\{g_\chi\}_{\chi\in\mchars}$ be the character combinations of $f$ and $g$, and for $a \in \Fp$, we define
\begin{align*}
F(a) & = \sum_{\chi \in \mchars} f_\chi\chi(a)\\
G(a) & = \sum_{\chi \in \mchars} g_\chi\chi(a),
\end{align*}
so that
\begin{align*}
f & = (F(r),\ldots,F(r+\ell-1)) \\
g & = (G(s),\ldots,G(s+\ell-1)).
\end{align*}
Let $L=\{0,1,\ldots,\ell-1\}$.

We first compute the numerator $\sum_{j \in \Z} |C_{f,g}(j)|^2$ of our expression \eqref{Celeste} for $\CDF(f,g)$.  By \cite[eq.~(14)]{Katz-2016-Aperiodic}, we have
\begin{align}
\sum_{j \in \Z} |C_{f,g}(j)|^2
& = \sums{a,b \in L \\ c,d \in L \\ a+b=c+d} F(a+r) G(b+s) \conj{F(c+r) G(d+s)} \nonumber \\
& = \sum_{\phi,\chi,\psi,\omega \in \mchars} f_\phi g_\chi \conj{f_\psi g_\omega} \sums{a,b \in L \\ c,d \in L \\ a+b=c+d} A(\phi,\chi,\psi,\omega;a,b,c,d) \label{Lester}
\end{align}
where
\[
A(\phi,\chi,\psi,\omega;a,b,c,d) =  \phi(a+r) \chi(b+s) \conj{\psi(c+r) \omega(d+s)}.
\]
Now we use \eqref{Genevieve} to expand the multiplicative characters in $A$ in terms of additive characters, so that $A$ is
\[
\frac{1}{p^4} \!\! \sums{t,u \in \Fp \\v,w \in \Fp} \!\! \tau_t(\phi) \tau_u(\chi) \conj{\tau_v(\psi) \tau_w(\omega)} \epsilon(-t(a+r)-u(b+s)+v(c+r)+w(d+s)).
\]
For any $x \in F$, we note that $(t+x,u+x,v+x,w+x)$ runs through $\Fp^4$ as $(t,u,v,w)$ does, so we may replace $(t,u,v,w)$ in our last expression with $(t+x,u+x,v+x,w+x)$ and then average over $x \in \Fp$.
When we do this we note that $-(t+x)(a+r)-(u+x)(b+s)+(v+x)(c+r)+(w+x)(d+s)$ is equal to $-t(a+r)-u(b+s)+v(c+r)+w(d+s)$ because we always have $a+b=c+d$, and so we see that
\[
A = \frac{1}{p^2} \sums{t,u \in \Fp \\v,w \in \Fp} B_{\phi,\chi,\psi,\omega}(t,u,v,w) \epsilon(-t(a+r)-u(b+s)+v(c+r)+w(d+s)),
\]
where
\[
B_{\phi,\chi,\psi,\omega}(t,u,v,w) = \frac{1}{p^3} \sum_{x \in \Fp} \tau_{t+x}(\phi) \tau_{u+x}(\chi) \conj{\tau_{v+x}(\psi) \tau_{w+x}(\omega)}.
\]
Thus, returning to \eqref{Lester}, we see that $\sum_{j\in\Z} |C_{f,g}(j)|^2$ is equal to 
\begin{equation}\label{Roger}
\frac{1}{p^2} \! \sums{t,u \in \Fp \\v,w \in \Fp} \!\! \Gamma(t,u,v,w) \!\!\!\!\!\!\! \sums{a,b \in L \\ c,d \in L \\ a+b=c+d} \!\!\!\!\!\!\! \epsilon(-t(a+r)-u(b+s)+v(c+r)+w(d+s)),
\end{equation}
where
\begin{equation}\label{Gordon} 
\Gamma(t,u,v,w) = \sums{\phi,\chi \in \mchars \\ \psi,\omega \in \mchars} f_\phi g_\chi \conj{f_\psi g_\omega} B_{\phi,\chi,\psi,\omega}(t,u,v,w).
\end{equation}
Now $B_{\phi,\chi,\psi,\omega}(t,u,v,w)$ is computed in \cite[Lemma 14]{Katz-2016-Aperiodic} to be $B=Q+R$ for some $R \in \C$ with $|R| \leq 3/\sqrt{p}$ and 
\[
Q=\begin{cases}
\frac{\tau(\phi)\tau(\chi)\conj{\tau(\psi)\tau(\omega)}}{p^2} & \text{if $t=u=v=w$ and $\phi\chi=\psi\omega$}, \\
1 & \text{if $t=v\not=u=w$, $\phi=\psi$, and $\chi=\omega$}, \\
1 & \text{if $t=w\not=u=v$, $\phi=\omega$, and $\chi=\psi$}, \\
\phi\psi(-1) & \text{if $t=u\not=v=w$, $\phi=\conj{\chi}$, and $\phi=\conj{\omega}$}, \\
0 & \text{otherwise}.
\end{cases}
\]
We substitute $Q+R$ for $B$ in \eqref{Gordon}, and obtain a decomposition $\Gamma(t,u,v,w)=M(t,u,v,w)+E(t,u,v,w)$, with
\[
M(t,u,v,w) = \begin{cases}
\sums{\phi,\chi,\psi,\omega \in \mchars  \\ \phi\chi=\psi\omega} f_\phi g_\chi \conj{f_\psi g_\omega} \, \frac{\tau(\phi)\tau(\chi)\conj{\tau(\psi)\tau(\omega)}}{p^2} & \text{if $t=u=v=w$,} \\
\sums{\phi,\chi,\psi,\omega \in \mchars  \\ \phi=\psi,\chi=\omega} f_\phi g_\chi \conj{f_\psi g_\omega} & \text{if $t=v\not=u=w$,} \\
\sums{\phi,\chi,\psi,\omega \in \mchars  \\ \phi=\omega,\chi=\psi} f_\phi g_\chi \conj{f_\psi g_\omega} & \text{if $t=w\not=u=v$,} \\
\sums{\phi,\chi,\psi,\omega \in \mchars  \\ \phi=\conj{\chi},\psi=\conj{\omega}} f_\phi g_\chi \conj{f_\psi g_\omega} \, \phi\psi(-1) & \text{if $t=u\not=v=w$,} \\
0 & \text{otherwise,}
\end{cases}
\]
and
\[
|E(t,u,v,w)| \leq \frac{3 W_f^2 W_g^2}{\sqrt{p}}.
\]
Now recall the parameters $S_{f,g}$, $U_{f,g}$, $V_{f,g}$, $W_f$, and $W_g$ defined in \eqref{Patrick}, and remember that we insist on normalizing our sequences so that $\sum_{\chi} |f_\chi|^2 = \sum_{\chi} |g_\chi|^2=1$, and apply the portion of Lemma \ref{Harold} below about $S_{f,g}+1+U_{f,g}+V_{f,g}$, to see that
\[
M(t,u,v,w) = \begin{cases}
S_{f,g}+1+U_{f,g}+V_{f,g} & \text{if $t=u=v=w$,} \\
1 & \text{if $t=v\not=u=w$,} \\
U_{f,g} & \text{if $t=w\not=u=v$,} \\
V_{f,g} & \text{if $t=u\not=v=w$,} \\
0 & \text{otherwise.}
\end{cases}
\]

We regard $M(t,u,v,w)$ as the main term and $E(t,u,v,w)$ as the error term in our decomposition $\Gamma(t,u,v,w)=M(t,u,v,w)+E(t,u,v,w)$.  When we substitute this for $\Gamma(t,u,v,w)$ in \eqref{Roger}, we see that $\sum_{j\in\Z} |C_{f,g}(j)|^2=M_0+E_0$, where
\[
M_0 = \frac{1}{p^2} \! \sums{t,u \in \Fp \\v,w \in \Fp} \!\! M(t,u,v,w) \!\!\!\!\!\!\! \sums{a,b \in L \\ c,d \in L \\ a+b=c+d} \!\!\!\!\!\!\! \epsilon(-t(a+r)-u(b+s)+v(c+r)+w(d+s)),
\]
and
\[
E_0 = \frac{1}{p^2} \! \sums{t,u \in \Fp \\v,w \in \Fp} \!\! E(t,u,v,w) \!\!\!\!\!\!\! \sums{a,b \in L \\ c,d \in L \\ a+b=c+d} \!\!\!\!\!\!\! \epsilon(-t(a+r)-u(b+s)+v(c+r)+w(d+s)),
\]
and our bound on $|E(t,u,v,w)|$ makes
\[
|E_0| \leq \frac{3 W_f^2 W_g^2}{p^{5/2}} \!\!\! \sums{t,u \in \Fp \\v,w \in \Fp} \! \left|\sums{a,b \in L \\ c,d \in L \\ a+b=c+d} \!\!\!\!\!\!\! \epsilon(-t(a+r)-u(b+s)+v(c+r)+w(d+s))\right|.
\]
This sum is bounded in \cite[Lemma 15]{Katz-2016-Aperiodic}, which tells us that
\begin{equation}\label{Eustace}
|E_0| \leq 192 W_f^2 W_g^2 p^{3/2} (1+\log p)^3 \max(1,\ell/p)^3.
\end{equation}
Now we return to $M_0$, which we break into four terms, $M_1$, $M_2$, $M_3$, and $M_4$, by partitioning the summation over $t,u,v,w$ into the four regimes where $M(t,u,v,w)$ is nonzero: (i) $t=u=v=w$, (ii) $t=v\not=u=w$, (iii) $t=w\not=u=v$, and (iv) $t=u\not=v=w$.
So
\[
\sum_{j \in \Z} |C_{f,g}(j)|^2 = M_0+E_0 = M_1+M_2+M_3+M_4+E_0,
\]
where we have bounded $E_0$ in \eqref{Eustace}, and
\begin{align*}
M_1 & = (S_{f,g}+1+U_{f,g}+V_{f,g}) \cdot \frac{1}{p} \sums{a,b \in L \\ c,d \in L \\ a+b=c+d} 1, \\
M_2 & = 1 \cdot \frac{1}{p^2} \sums{t,u \in F \\ t\not=u} \sums{a,b \in L \\ c,d \in L \\ a+b=c+d} \epsilon((t-u)(c-a)), \\
M_3 & = U_{f,g} \cdot \frac{1}{p^2} \sums{t,u \in F \\ t\not=u} \sums{a,b \in L \\ c,d \in S \\ a+b=c+d} \epsilon((t-u)(d-a+s-r)), \\
M_4 & = V_{f,g} \cdot \frac{1}{p^2} \sums{t,v \in F \\ t\not=v} \sums{a,b \in L \\ c,d \in S \\ a+b=c+d} \epsilon((v-t)(a+b+r+s)).
\end{align*}
Let $I_1=\{(a,b,c,d) \in L^4: a+b=c+d\}$, so that
\[
M_1 = (S_{f,g}+1+U_{f,g}+V_{f,g}) \frac{\card{I_1}}{p}.
\]
Consider the sum over $t$ and $u$ in $M_2$.  If we restore the terms with $t=u$, we may apply the orthogonality relation \eqref{Oliver} to obtain 
\[
\frac{1}{p^2} \sum_{t,u \in \Fp} \sums{a,b \in L \\ c,d \in L \\ a+b=c+d} \epsilon((t-u)(c-a)) = \card{I_2},
\]
where
\[
I_2 = \{(a,b,c,d) \in I_1: c \equiv a \!\! \pmod{p}\}
\]
Thus if we remove the $t=u$ terms, see that $M_2$ actually comes out to
\[
M_2=\frac{1}{p^2} \sums{t,u \in \Fp \\ t\not=u} \sums{a,b \in L \\ c,d \in L \\ a+b=c+d} \epsilon((t-u)(c-a)) = \card{I_2}-\frac{\card{I_1}}{p}.
\]

Similarly, one shows that
\begin{align*}
M_3 & = U_{f,g} \left(\card{I_3}-\frac{\card{I_1}}{p}\right) \\
M_4 & = V_{f,g} \left(\card{I_4}-\frac{\card{I_1}}{p}\right), \\
\end{align*}
where
\begin{align*}
I_3 & = \{(a,b,c,d) \in I_1: d-a \equiv r-s \!\! \pmod{p}\}, \\
I_4 & = \{(a,b,c,d) \in I_1: a+b \equiv -(r+s) \!\! \pmod{p}\}. \\
\end{align*}
Thus we obtain $\sum_{j\in\Z} |C_{f,g}(j)|^2=M_S+M_T+M_U+M_V+E_0$, where $E_0$ is bounded in \eqref{Eustace} and 
\begin{align*}
M_S & = S_{f,g} \card{I_1}/p \\
M_T & = \card{I_2} \\
M_U & = U_{f,g} \card{I_3} \\
M_V & = V_{f,g} \card{I_4}.
\end{align*}
We use the computation of the cardinalities of the sets $I_1$, $I_2$, $I_3$, and $I_4$ in \cite[Lemmata 11--13]{Katz-2016-Aperiodic} to obtain
\begin{align*}
M_S & = S_{f,g} \cdot \frac{1}{p} \left(\frac{2 \ell^3+\ell}{3}\right) \\
M_T & = \ell^2 \cdot \Omega\left(\frac{p}{\ell},0\right) \\
M_U & = U_{f,g} \cdot \ell^2 \cdot \Omega\left(\frac{p}{\ell},\frac{r-s}{\ell}\right) \\
M_V & = V_{f,g} \cdot \ell^2 \cdot \Omega\left(\frac{p}{\ell},1+\frac{r+s-1}{\ell}\right),
\end{align*}
so then
\[
\frac{\sum_{j\in\Z} |C_{f,g}(j)|^2}{\ell^2} = N_S+N_T+N_U+N_V+E_1,
\]
where
\begin{align*}
N_S & = S_{f,g} \cdot \frac{2 \ell^2+1}{3 p \ell} \\
N_T & = \Omega\left(\frac{p}{\ell},0\right) \\
N_U & = U_{f,g} \cdot \Omega\left(\frac{p}{\ell},\frac{r-s}{\ell}\right) \\
N_V & = V_{f,g} \cdot \Omega\left(\frac{p}{\ell},1+\frac{r+s-1}{\ell}\right), \\
|E_1| & \leq  192 \frac{W_f^2 W_g^2 (1+\log p)^3}{\sqrt{p}} \left(\frac{p}{\ell}\right)^2 \max\left(1,\frac{\ell}{p}\right)^3.
\end{align*}

Now we look at the asymptotic behavior of $\CDF(f,g)$ for pairs $(f,g)$ in our family.
First of all, note that as $p \to \infty$ we also have $\ell\to\infty$ because $p/\ell$ tends to a positive real limit $\Lambda$ as $p\to \infty$.
Thus we may apply Lemma \ref{Lawrence} below to see that $|C_{f,f}(0)|/\ell \to 1$ and $|C_{g,g}(0)|/\ell \to 1$ as $p\to \infty$.
And so $\CDF(f,g)$ has the same limiting behavior as $\sum_{j \in\Z} |C_{f,g}(j)|^2/\ell^2$ as $p\to \infty$.

We now compute limits as $p\to \infty$ using the fact that $\ell/p \to \Lambda$ as $p\to\infty$ and that $\Omega$ is continuous on its domain.
We see that $N_S \to 2 S_{f,g} \Lambda/3$ and $N_T \to \Omega(1/\Lambda,0)$.

When $U \not=0$, we have an additional assumption that $(r-s)/p \to \Delta$ as $p\to \infty$, so we have $N_U \to U \cdot \Omega(1/\Lambda,\Delta/\Lambda)$.  When $U=0$, then $\Omega\left(\frac{p}{\ell},\frac{r-s}{\ell}\right)$ may not tend to a stable limit, but since $\ell/p \to \Lambda$, we know that for sufficiently large $p$ we can guarantee that $\ell/p \leq \floor{\Lambda}+1$, which makes $\Omega\left(\frac{p}{\ell},\frac{r-s}{\ell}\right) \leq 2(\floor{\Lambda}+1)$ by Lemma \ref{Jessica} below, and so $N_U \to 0 = U \cdot \Omega(1/\Lambda,\Delta/\Lambda)$ also in this case.

When $V_{f,g} \not=0$, we have an additional assumption that $(r+s)/p \to \Sigma$ as $p\to \infty$, so we have $N_V \to V \cdot \Omega(1/\Lambda,1+\Sigma/\Lambda)$.  This limit is also true when $V=0$ by the same argument we used when $U=0$ in the previous paragraph.

Finally, $|E_1| \to 0$ by our assumption about the asymptotic behavior of $W_f^2 W_g^2$.

This completes our proof for the limit of $\CDF(f,g)$ when our sequence pairs are true character combination sequences (not unimodularized).
If our sequence pairs $(f_\iota,g_\iota)$ our unimodularizable, and we replace them with their unimodularizations, then Lemma \ref{Nancy} below shows that this does not change the limiting behavior of the crosscorrelation demerit factor.
To see that Lemma \ref{Nancy} applies, we note that $\ell\to\infty$ and $\ell/p^2 \to 0$ as $p\to\infty$ because $\ell/p$ tends to the positive real number $\Lambda$ as $p\to \infty$.
\end{proof}
If we specialize Theorem \ref{Mary} to autocorrelation we obtain the following result for asymptotic autocorrelation demerit factor.
\begin{corollary}\label{Ellen}
Let $\{f_\iota\}_{\iota \in I}$ be a family of unimodularizable character combination sequences, where for each $\iota \in I$, the sequence $f_\iota$ has prime $p_\iota$, field $F_\iota$ of order $p_\iota$, length $\ell_\iota$, and shift $r_\iota$.
Suppose that $\{p_\iota\}_{\iota \in I}$ is infinite and that $\ell_\iota/p_\iota$ tends to a positive real number $\Lambda$ as $p_\iota \to \infty$.
Of the parameters defined in \eqref{Albert}, suppose that $S_{f_\iota,f_\iota}$ and $V_{f_\iota,f_\iota}$ tend to real limits $S$ and $V$, respectively, as $p_\iota \to \infty$ and that $W_{f_\iota}^4 (\log p_\iota)^3/\sqrt{p_\iota} \to 0$ as $p_\iota \to \infty$.
If $V \not=0$, then suppose that $r_\iota/p_\iota$ tends to a real limit $R$ as $p_\iota \to \infty$.
Then
\[
\DF(f_\iota) \to -1 + S \cdot \frac{2}{3}\Lambda + 2 \Omega\left(\frac{1}{\Lambda},0\right) + V \cdot \Omega\left(\frac{1}{\Lambda},1+\frac{2 R}{\Lambda}\right)
\]
as $p_\iota \to \infty$.
If $f^u_\iota$ is a unimodularization of $f_\iota$ for each $\iota \in I$, then $\DF(f^u_\iota)$ has the same limit as $\DF(f_\iota)$.
\end{corollary}
\begin{proof}
This is just the special case of Theorem \ref{Mary} where we let $f_\iota=g_\iota$ (so $r_\iota=s_\iota$) for all $\iota \in I$.
Then we consider the versions of our parameters for autocorrelation in \eqref{Albert}, and recall that $\DF(f)=\CDF(f,f)-1$.
\end{proof}
We conclude this section with the technical lemmata we needed for our proof.
\begin{lemma}\label{Harold}
Let $f, g$ be character combination sequences with prime $p$ and respective character combinations $\{f_\chi\}_{\chi \in \mchars}$ and $\{g_\chi\}_{\chi \in \mchars}$.  If $S_{f,g}$, $U_{f,g}$, and $V_{f,g}$ are the parameters defined in \eqref{Patrick}, then
\begin{align*}
S_{f,g}+1+U_{f,g}+V_{f,g} & =  \sums{\phi,\chi, \psi,\omega \in \mchars \\ \phi\chi=\psi\omega} f_\phi g_\chi \conj{f_\psi g_\omega} \frac{\tau(\phi)\tau(\chi)\conj{\tau(\psi)\tau(\omega)}}{p^2} \\
1 & = \sums{\phi,\chi, \psi,\omega \in \mchars \\ \phi=\psi \\ \chi=\omega} f_\phi g_\chi \conj{f_\psi g_\omega} \frac{\tau(\phi)\tau(\chi)\conj{\tau(\psi)\tau(\omega)}}{p^2} \\
U_{f,g} & = \sums{\phi,\chi, \psi,\omega \in \mchars \\ \phi=\omega \\ \chi=\psi} f_\phi g_\chi \conj{f_\psi g_\omega} \frac{\tau(\phi)\tau(\chi)\conj{\tau(\psi)\tau(\omega)}}{p^2} \\
V_{f,g} & = \sums{\phi,\chi, \psi,\omega \in \mchars \\ \phi=\conj{\chi} \\ \psi=\conj{\omega}} f_\phi g_\chi \conj{f_\psi g_\omega} \frac{\tau(\phi)\tau(\chi)\conj{\tau(\psi)\tau(\omega)}}{p^2}.
\end{align*}
\end{lemma}
\begin{proof}
Note that \eqref{Manuel} shows that $\tau(\phi)\tau(\chi) \conj{\tau(\phi)\tau(\chi)}/p^2=1$ for every nontrivial $\phi,\chi \in \mchars$.  Since we insist (see Section \ref{Karl}) that $f_{\chi_0}=g_{\chi_0}=0$ for the trivial character $\chi_0$, and that $\sum_{\chi \in \mchars} |f_\chi|^2=\sum_{\chi\in\mchars} |g_\chi|^2=1$, we see that
\[
1 = \sums{\phi,\chi\in\mchars} f_\phi g_\chi \conj{f_\phi g_\chi} \frac{\tau(\phi)\tau(\chi)\conj{\tau(\phi)\tau(\chi)}}{p^2},
\]
which establishes the second identity that we were to show.

Likewise \eqref{Manuel} shows that $\tau(\phi)\tau(\chi) \conj{\tau(\chi)\tau(\phi)}/p^2=1$ for every nontrivial $\phi,\chi \in \mchars$, and since $f_{\chi_0}=g_{\chi_0}=0$ for the trivial character $\chi_0$, we see that
\[
U_{f,g} = \sums{\phi,\chi \in\mchars} f_\phi g_\chi \conj{f_\chi g_\phi} \frac{\tau(\phi)\tau(\chi)\conj{\tau(\chi)\tau(\phi)}}{p^2},
\]
which establishes the third identity that we were to show.

Similarly \eqref{Katherine} and \eqref{Manuel} show that $\tau(\phi)\tau(\conj{\phi}) \conj{\tau(\psi)\tau(\conj{\psi})}/p^2=\phi\conj{\psi}(-1)$ for every nontrivial $\phi,\chi \in \mchars$, and since $f_{\chi_0}=g_{\chi_0}=0$ for the trivial character $\chi_0$, we see that
\[
V_{f,g} = \sum_{\phi,\psi \in \mchars} f_\phi g_{\conj{\phi}} \conj{f_\psi g_{\conj{\psi}}} \frac{\tau(\phi)\tau(\conj{\phi}) \conj{\tau(\psi)\tau(\conj{\psi})}}{p^2},
\]
which establishes the fourth identity we were to show.

If we add the three identities we proved here for $1$, $U_{f,g}$, and $V_{f,g}$, we find that we are summing terms of the form $f_\phi g_\chi \conj{f_\psi g_\omega} \tau(\phi)\tau(\chi)\conj{\tau(\psi)\tau(\omega)}/p^2$ for a collection of quadruples (possibly with repetitions) drawn from $Q=\{(\phi,\chi,\psi,\omega) \in \mchars^4: \phi\chi=\psi\omega\}$.
\begin{itemize}
\item The sum for $1$ uses those quadruples in $Q$ with $\phi=\psi$. (The other condition that $\chi=\omega$ is automatically fulfilled because of the condition $\phi\chi=\psi\omega$ for belonging to $Q$.)
\item The sum for $U_{f,g}$ uses those quadruples in $Q$ with $\phi=\omega$.
\item The sum for $V_{f,g}$ uses those quadruples in $Q$ with $\phi=\conj{\chi}$.
\end{itemize}
Thus by the inclusion-exclusion principle, we have
\begin{align*}
1+U_{f,g}+V_{f,g}
& = \sums{(\phi,\chi,\psi,\omega) \in Q \\ \phi \in \{\conj{\chi},\psi,\omega\}} f_\phi g_\chi \conj{f_\psi g_\omega} \frac{\tau(\phi)\tau(\chi)\conj{\tau(\psi)\tau(\omega)}}{p^2} \\
& \qquad + \sums{(\phi,\chi,\psi,\omega) \in Q \\ \phi=\chi=\psi=\omega} f_\phi g_\chi \conj{f_\psi g_\omega} \frac{\tau(\phi)\tau(\chi)\conj{\tau(\psi)\tau(\omega)}}{p^2} \\
& \qquad + \sums{(\phi,\chi,\psi,\omega) \in Q \\ \phi=\conj{\chi}=\psi=\conj{\omega}} f_\phi g_\chi \conj{f_\psi g_\omega}  \frac{\tau(\phi)\tau(\chi)\conj{\tau(\psi)\tau(\omega)}}{p^2} \\
& \qquad + \sums{(\phi,\chi,\psi,\omega) \in Q \\ \phi=\conj{\chi}=\conj{\psi}=\omega} f_\phi g_\chi \conj{f_\psi g_\omega} \frac{\tau(\phi)\tau(\chi)\conj{\tau(\psi)\tau(\omega)}}{p^2} \\
& \qquad - \sums{(\phi,\chi,\psi,\omega) \in Q \\ \phi=\chi=\psi=\omega=\conj{\phi}=\conj{\chi}=\conj{\psi}=\conj{\omega}} \!\!\!\!\!\!\!\!\!\! f_\phi g_\chi \conj{f_\psi g_\omega}  \frac{\tau(\phi)\tau(\chi)\conj{\tau(\psi)\tau(\omega)}}{p^2},
\end{align*}
and then we use \eqref{Manuel} to obtain
\begin{align*}
1+U_{f,g}+V_{f,g}
& = \sums{\phi,\chi,\psi,\omega \in \mchars\\ \phi\chi=\psi\omega \\ \phi \in \{\conj{\chi},\psi,\omega\}} f_\phi g_\chi \conj{f_\psi g_\omega} \frac{\tau(\phi)\tau(\chi)\conj{\tau(\psi)\tau(\omega)}}{p^2} \\
& \qquad\quad + \sum_{\phi} |f_\phi g_\phi|^2 + \sum_{\phi} |f_\phi g_{\conj{\phi}}|^2 + \sum_{\phi} f_\phi \conj{f_{\conj{\phi}}} g_{\conj{\phi}} \conj{g_\phi}    - |f_\eta g_\eta|^2,
\end{align*}
where $\eta$ is the quadratic character.  Note that the only characters $\phi$ with $\phi=\conj{\phi}$ are the trivial character and the quadratic character, but we insist (see Section \ref{Karl}) that $f_{\chi_0}=g_{\chi_0}=0$ for the trivial character $\chi_0$.
When we add our last expression for $1+U_{f,g}+V_{f,g}$ to $S_{f,g}$ from \eqref{Patrick}, we get the desired relation.
\end{proof}
\begin{lemma}\label{Lawrence}
Let $\{f_\iota\}_{\iota \in I}$ be a family of unimodularizable character combination sequences where sequence $f_\iota$ has prime $p_\iota$ and length $\ell_\iota$, and let $f^u_\iota$ be a unimodularization of $f_\iota$.
Suppose that $\{p_\iota\}_{\iota \in I}$ is infinite and $\ell_\iota \to \infty$ as $p_\iota \to \infty$.
Then $C_{f^u_\iota,f^u_\iota}(0)=\ell_\iota$.
If $\iota \in I$, then $C_{f_\iota,f_\iota}(0)$ is either $\ell_\iota-\ceil{\ell_\iota/p_\iota}$ or $\ell_\iota-\floor{\ell_\iota/p_\iota}$.
And $C_{f_\iota,f_\iota}(0)/\ell_\iota \to 1$ as $p_\iota \to \infty$.
\end{lemma}
\begin{proof}
The unimodularization $f^u_\iota$ is a sequence of length $\ell_\iota$ whose terms are all of magnitude $1$, so $C_{f^u_\iota,f^u_\iota}(0)=\ell_\iota$.  On the other hand, $f_\iota$ can be viewed as a finite segment of length $\ell_\iota$ of a periodic sequence of period $p_\iota$, where each period has $p_\iota-1$ unimodular terms and one zero term.  So $f_\iota$ has all terms of magnitude $1$, except for either $\floor{\ell_\iota/p_\iota}$ or $\ceil{\ell_\iota/p_\iota}$ zero terms, thus proving our second claim.
Thus $C_{f_\iota,f_\iota}(0)$ differs from $\ell_\iota$ by a quantity of magnitude less than $1+\ell_\iota/p_\iota$, and so
\[
\left|\frac{C_{f_\iota,f_\iota}(0)}{\ell_\iota} -1\right| < \frac{1}{\ell_\iota} + \frac{1}{p_\iota},
\]
which tends to $0$ as $p_\iota \to \infty$ by our given assumptions.
\end{proof}
\begin{lemma}\label{Nancy}
Let $\{(f_\iota,g_\iota)\}_{\iota \in I}$ be a family of pairs of unimodularizable character combination sequences, where for each $\iota \in I$, both $f_\iota$ and $g_\iota$ have prime $p_\iota$ and length $\ell_\iota$.
For each $\iota \in I$, let $f^u_\iota$ and $g^u_\iota$ be unimodularizations of $f_\iota$ and $g_\iota$, respectively.  Suppose that $\{p_\iota\}_{\iota \in I}$ is infinite and $\ell_\iota \to \infty$ and $\ell_\iota/p_\iota^2 \to 0$ as $p_\iota \to \infty$.  Then $\CDF(f^u_\iota,g^u_\iota)$ tends to a real number as $p_\iota \to \infty$ if and only if $\CDF(f_\iota,g_\iota)$ does, and in this case, the limits are equal.
\end{lemma}
\begin{proof}
In this proof, we identify any sequence $h=(h_0,\ldots,h_{n-1})$ of complex numbers with the polynomial $h(z)=h_0+\cdots+h_{n-1} z^{n-1}$ and define the $L^2$ norm of $h(z)$ on the unit circle as
\[
\|h\|_2=\left(\frac{1}{2\pi} \int_0^{2\pi} |h(e^{i\theta})|^2 \, d\theta\right)^{1/2}.
\]
Then it is shown \cite[Section V]{Katz-2016-Aperiodic} that if $f$ and $g$ are sequences then
\begin{equation}\label{Arthur}
C_{f,f}(0)=\|f\|_2^2
\end{equation}
and
\begin{equation}\label{Herbert}
\CDF(f,g) = \frac{\|f g\|_2^2}{\|f\|_2^2 \|g\|_2^2}.
\end{equation}
The advantage of this point of view is that it enables us to use the triangle inequality for the $L^2$ norm.

We let $(f,g)$ be a sequence pair from our family, and let $f^u$ and $g^u$ be unimodularizations of $f$ and $g$ respectively.
Then the triangle inequality tells us that
\[
\left|\|f^u g^u\|_2 -\|f g\|_2\right| \leq \|(f^u-f)g^u\|_2 + \|f(g^u-g)\|_2.
\]
Now $f^u$ is the unimodularization of $f$, so Lemma \ref{Lawrence} shows that they differ in at most $\ceil{\ell/p}$ positions (where $f$ has a zero and $f^u$ has a unimodular complex number).  The same is true of $g^u$ as compared to $g$.
The triangle inequality and the fact that $\|z^j\|_2=1$ for all $j$ implies that if $a=\sum_{j=0}^d a_j z^j$ and $b(z)$ are polynomials, then
\[
\|a b\|_2 \leq \|b\|_2 \sum_{j=0}^d |a_j|.
\]
Thus we have
\[
\left|\|f^u g^u\|_2 -\|f g\|_2\right| \leq \ceil{\ell/p} \left(\|g^u\|_2+\|f\|_2\right),
\]
and since $g^u$ and $f$ are sequences of length $\ell$ with terms of magnitude at most $1$, we know that $C_{g^u,g^u}(0)$ and $C_{f,f}(0)$ are at most $\ell$, and so by \eqref{Arthur}, we know that $\|g^u\|_2, \|f\|_2 \leq \sqrt{\ell}$, and so
\[
\left|\|f^u g^u\|_2 -\|f g\|_2\right| \leq 2 \sqrt{\ell} \ceil{\ell/p}.
\]
Thus 
\begin{equation}\label{Timothy}
\frac{\left|\|f^u g^u\|_2 -\|f g\|_2\right|}{\|f\|_2 \|g\|_2} \leq \frac{2 \sqrt{\ell} (1+\ell/p)}{\|f\|_2 \|g\|_2},
\end{equation}
the right hand of which has the same asymptotic behavior as $2 (\ell^{-1/2}+\ell^{1/2} p^{-1})$ by Lemma \ref{Lawrence}.  Now our given assumption that $\ell\to\infty$ and $\ell/p^2 \to 0$ as $p\to \infty$ make $2(\ell^{-1/2}+\ell^{1/2} p^{-1}) \to 0$ as $p\to\infty$, and thus the right hand side of \eqref{Timothy} tends to $0$ in this limit.
So $\|f g\|_2/(\|f\|_2\|g\|_2)$ and $\|f^u g^u\|_2/(\|f\|_2 \|g\|_2)$ have the same limiting behavior.
Lemma \ref{Lawrence} then shows that $\|f^u g^u\|_2/(\|f\|_2 \|g\|_2)$ and $\|f^u g^u\|_2/(\|f^u\|_2 \|g^u\|_2)$ have the same limiting behavior, so by \eqref{Herbert}, $\CDF(f,g)$ and $\CDF(f^u,g^u)$ have the same limiting behavior.
\end{proof}
\begin{lemma}\label{Jessica}
If $x, y \in \R$ with $x > 0$, then $0 \leq \Omega(1/x,y) \leq 2 \ceil{x}$.
\end{lemma}
\begin{proof}
Since $\Omega(1/x,y)=\sum_{n \in \Z} \max(0,1-|n/x-y|)^2$, it is clearly nonnegative, and note that the $n$th term is at most $1$ and is nonzero if and only if $n/x \in (y-1,y+1)$.  The values $\{n/x: n \in \Z\}$ form a lattice in $\R$ with spacing $1/x \geq 1/\ceil{x}$, so the interval $(y-1,y+1)$ of length $2$ contains at most $2 \ceil{x}$ such lattice points.  So $\Omega(1/x,y)$ is a sum of $2 \ceil{x}$ terms, each at most $1$.
\end{proof}

\section{Connection to Periodic Correlation}\label{Percy}

We now explore the connection between aperiodic correlation and periodic correlation.
Suppose that $f=(f_0,\ldots,f_{n-1})$ is a sequence of complex numbers of length $n$.
We can regard it as a sequence of period $n$ by using the convention that $f_j=f_{j+n}$ for every $j \in \Z$.
In this case we write our sequence as $f=(f_j)_{j \in \Z/n\Z}$ to emphasize its periodic nature.
If $g=(g_j)_{j \in \Z/n\Z}$ is another such sequence, and if $s \in \Z$, then the {\it periodic crosscorrelation of $f$ with $g$ at shift $s$} is
\[
\PC_{f,g}(s)=\sum_{j \in \Z/n\Z} f_j \conj{g_{j+s}},
\]
where the fact that the summation is indexed over $\Z/n\Z$ tells us that we are treating the sequences periodically.
When we compare with our definition of aperiodic crosscorrelation in \eqref{Agnes}, we see that it is this indexing that is the only difference.
When $f=g$, then $\PC_{f,f}(s)$ is the {\it periodic autocorrelation of $f$ at shift $s$}.

Suppose that $f$ is a character combination sequence with prime $p$, field $\Fp$, shift $s$, length $\ell$, and character combination $\{f_\chi\}_{\chi\in \mchars}$.  Then we know that $f=(F(s),\ldots,F(s+\ell-1))$ where $F(j)=\sum_{\chi \in \mchars} f_\chi \chi(j)$.
The {\it periodic version of $f$}, written $\per(f)$, is the periodically-indexed sequence $(F(j))_{j \in \Fp}$, which is of length $p$ regardless of the length of the original sequence $f$.
We define $\per(f)$ this way because the function $F$ that generates the terms of $f$ has natural period $p$, and it is the sequence $\per(f)$ of natural period $p$ whose periodic correlation behavior is related to the aperiodic correlation behavior of the original sequence $f$.

For a periodic sequence $f=(f_j)_{j \in \Fp}$ of complex numbers, we define the {\it Fourier transform} of $f$ to be the map $\ft{f}\colon \Fp \to \C$ with $\ft{f}(a)=\sum_{x \in \Fp} f_x \epsilon_a(x)$.  We usually denote $\ft{f}(a)$ as $\ft{f}_a$ instead.  It is not hard to show that $f_x = \frac{1}{p} \sum_{a \in \Fp} \ft{f}_a \conj{\epsilon_a(x)}$, which is the inverse Fourier transform.

The principal result of this section is to connect the parameters from \eqref{Patrick} for a pair of character combination sequences to the periodic crosscorrelation.
This will become useful later in Section \ref{Ophelia}, where we show that low periodic autocorrelation leads to low aperiodic crosscorrelation.
\begin{proposition}\label{Hilda}
Let $f$ and $g$ be character combination sequences with prime $p$.  Let $S_{f,g}$, $U_{f,g}$, and $V_{f,g}$ be the parameters for these sequences defined in \eqref{Patrick}.
Then
\[
\frac{1}{p(p-1)} \sum_{a \in \Fp} |\PC_{\per(f),\per(g)}(a)|^2 = S_{f,g} + 1 + U_{f,g}+ V_{f,g},
\]
and if $f=g$, then $S_{f,f}+1+U_{f,f}+V_{f,f}=S_{f,f}+2+V_{f,f} \geq 1$.
\end{proposition}
\begin{proof}
Note that if $\{f_\chi\}_{\chi\in \mchars}$ is the character combination for $f$, then
\begin{align*}
\ft{\per(f)}_a
& = \sum_{\chi\in\mchars} f_\chi \sum_{j \in \Fp} \chi(j) \epsilon_a(j) \\
& = \sum_{\chi\in\mchars} f_\chi \tau_a(\chi) \\
& = \sum_{\chi\in\mchars} f_\chi \tau(\chi) \conj{\chi}(a),
\end{align*}
where we have used the definition of the Gauss sum from \eqref{Theodore} in the second step (keeping in mind that $\chi(0)=0$ by our convention for multiplicative characters), and we have used \eqref{George} in the last step.
And similarly, if $\{g_\chi\}_{\chi\in\mchars}$ is the character combination for $g$, then $\ft{\per(g)}_a=\sum_{\chi\in\mchars} g_\chi \tau(\chi) \conj{\chi}(a)$.
We use Lemma \ref{Norman} below and our values of $\ft{\per(f)}_a$ and $\ft{\per(g)}_a$ to see that
\begin{align*}
\sum_{a \in \Fp} \frac{|\PC_{\per(f),\per(g)}(a)|^2}{p(p-1)}
& = \frac{1}{p^2(p-1)} \sum_{a \in \Fp} |\ft{\per(f)}_a \ft{\per(g)}_a|^2 \\
& = \!\!\!\!\! \sum_{\phi,\chi,\psi,\omega \in \mchars} \!\!\!\!\! \frac{f_\phi g_\chi \conj{f_\psi g_\omega} \tau(\phi)\tau(\chi) \conj{\tau(\psi)\tau(\omega)}}{p^2 (p-1)} \sum_{a \in \Fp} \conj{\phi\chi}\psi\omega(a) \\
& = \sums{\phi,\chi,\psi,\omega \in \mchars \\ \phi\chi=\psi\omega} f_\phi g_\chi \conj{f_\psi g_\omega} \frac{\tau(\phi)\tau(\chi) \conj{\tau(\psi)\tau(\omega)}}{p^2} \\
& = S_{f,g}+1+U_{f,g}+V_{f,g},
\end{align*}
where we have used the orthogonality relation \eqref{Otto} in the penultimate step, and Lemma \ref{Harold} in the ultimate one.

When $f=g$, then $U_{f,f}=1$ per \eqref{Albert}.
By the Cauchy-Schwarz inequality applied to the vector $v$ of length $p-1$ whose entries are $|\ft{\per(f)}_a|^2$ for $a \in \{1,\ldots,p-1\}$ and the vector $w$ whose entries are all $1$, we see that
\[
\sum_{a \in \Fpu} |\ft{\per(f)}_a|^2 \leq \sqrt{p-1} \sqrt{\sum_{a \in \Fpu} |\ft{\per(f)}_a|^4},
\]
so that
\[
\sum_{a \in \Fp} |\ft{\per(f)}_a|^4 \geq \frac{1}{p-1} \left(\sum_{a \in \Fpu} |\ft{\per(f)}_a|^2\right)^2.
\]
From Lemma \ref{Julia} below, we know that $\sum_{a \in \Fpu} |\ft{\per(f)}_a|^2=(p-1)p$, so that
\[
\sum_{a \in \Fp} |\ft{\per(f)}_a|^4 \geq (p-1)p^2,
\]
so that Lemma \ref{Norman} shows that
\[
\frac{1}{p(p-1)} \sum_{a \in \Fp} |\PC_{\per(f),\per(f)}(a)|^2 \geq 1,
\]
and since we already know that the left hand side of this last inequality is $S_{f,f}+1+U_{f,f}+V_{f,f}=S_{f,f}+2+V_{f,f}$, this completes our proof.
\end{proof}
We close this section with the technical lemmata used in our proof.
\begin{lemma}\label{Julia}
Let $p$ be a prime.
Let $f$ be a character combination sequence with prime $p$.  Then
\[
\sum_{a \in \Fpu} |\ft{\per(f)}_a|^2 =  p \PC_{\per(f),\per(f)}(0) = p(p-1)
\]
and $\ft{\per(f)}_0=0$.
\end{lemma}
\begin{proof}
Let $\{f_\chi\}_{\chi \in \mchars}$ be the character combination and $s$ be the shift of $f$, so that the $j$th term of $f$ is $\sum_{\chi\in\mchars} f_\chi \chi(j+s)$.  Then
\begin{align*}
\ft{\per(f)}_0
& = \sum_{j \in \Fp} \sum_{\chi\in\mchars} f_\chi \chi(j+s) \epsilon_0(j) \\
& = \sum_{\chi \in \mchars} f_\chi \sum_{j \in \Fp} \chi(j+s),
\end{align*}
and the sum over $j$ is zero by orthogonality relation \eqref{Otto} unless $\chi$ is the trivial character.  But our character combination sequences always have $f_\chi=0$ when $\chi$ is the trivial character (see Section \ref{Karl}), so that $\ft{\per(f)}_0=0$.

Now that we know that $\ft{\per(f)}_0=0$, we use the Parseval theorem to see that $\sum_{a \in \Fpu} |\ft{\per(f)}_a|^2$ is $p$ times the sum of the squared magnitudes of the terms of $\per(f)$, that is,
\begin{align*}
\sum_{a \in \Fpu} |\ft{\per(f)}_a|^2
& = p \PC_{\per(f),\per(f)}(0) \\
& = p \sum_{j \in \Fp} |f_j|^2 \\
& = p \sum_{j \in \Fp} \left|\sum_{\phi\in\mchars} f_\phi \phi(j)\right|^2 \\
& = p \sum_{j \in \Fp} \sum_{\phi,\chi \in \mchars} f_\phi \phi(j) \conj{f_\chi \chi(j)} \\
& = p \sum_{\phi,\chi \in \mchars} f_\phi \conj{f_\chi} \sum_{j \in \Fp} \phi\conj{\chi}(j) \\
& = p(p-1) \sums{\phi,\chi \in \mchars \\ \phi=\chi} f_\phi \conj{f_\chi} \\
& = p(p-1),
\end{align*}
where we use the orthogonality relation \eqref{Otto} in the penultimate step, and the fact that we normalize our sequences so that $\sum_{\chi} |f_\chi|^2=1$ in the ultimate step.
\end{proof}
\begin{lemma}\label{Norman}
Let $n$ be a positive integer, and let $f=(f_j)_{j \in \Z/n\Z}$ and $g=(g_j)_{j \in \Z/n\Z}$ be periodic sequences of complex numbers.  Then
\[
\sum_{a \in \Z/n\Z} |\PC_{f,g}(a)|^2 = \frac{1}{n} \sum_{a \in \Z/n\Z} |\ft{f}_a \ft{g}_a|^2.
\]
\end{lemma}
\begin{proof}
Note that
\begin{align*}
\frac{1}{n} \sum_{a \in \Z/n\Z} |\ft{f}_a \ft{g}_a|^2
& = \frac{1}{n} \sum_{a,t,u,v,w \in \Z/n\Z} f_t \epsilon_a(t) g_u \epsilon_a(u) \conj{f_v \epsilon_a(v) g_w \epsilon_a(w)} \\
& = \sum_{t,u,v,w \in \Z/n\Z} f_t g_u \conj{f_v g_w} \cdot \frac{1}{n} \sum_{a \in \Z/n\Z} \epsilon(a(t+u-v-w)) \\
& = \sums{t,u,v,w \in \Z/n\Z \\ t+u=v+w} f_t g_u \conj{f_v g_w} \\
& = \sum_{s,t,v \in \Z/n\Z} f_t g_{v+s} \conj{f_{v} g_{t+s}},
\end{align*}
where in the last step, we reparameterize the sum with $w=t+s$ and $u=v+s$.
Thus 
\begin{align*}
\frac{1}{n} \sum_{a \in \Z/n\Z} |\ft{f}_a \ft{g}_a|^2
& = \sum_{s \in \Z/n\Z} \left(\sum_{t \in \Z/n\Z} f_t \conj{g_{t+s}}\right) \conj{\left(\sum_{v \in \Z/n\Z} f_v \conj{g_{v+s}}\right)} \\
& = \sum_{s \in \Z/n\Z} |\PC_{f,g}(s)|^2.\qedhere
\end{align*}
\end{proof}

\section{Optimum Performance}\label{Ophelia}

The following theorem shows that there is a maximum asymptotic autocorrelation merit factor that can be obtained by families of our character combination sequences in the limit described in Corollary \ref{Ellen}.  It also indicates precisely the parameter values that allow us to attain the maximum.
\begin{theorem}\label{Richard}
Let $\{f_\iota\}_{\iota \in I}$ be a family of character combination sequences meeting the hypotheses of Corollary \ref{Ellen}, with $p_\iota$ the prime for $f_\iota$, and with $f^u_\iota$ a unimodularization of $f_\iota$ for each $\iota \in I$.  Then the asymptotic demerit factor $\lim_{p_\iota \to \infty} \DF(f_\iota)=\lim_{p_\iota \to \infty} \DF(f^u_\iota)$ given by Corollary \ref{Ellen} is at least $0.157677\ldots$, the smallest root of $27 x^3+417 x^2+249 x+29$, or equivalently, the asymptotic merit factor is no greater than $6.342061\ldots$, the largest root of $29 x^3+249 x^2+417 x + 27$.  This lower bound on asymptotic demerit factor (or equivalent upper bound on merit factor) is achievable.  To achieve this optimum, it is necessary and sufficient that the limiting parameter values in Corollary \ref{Ellen} be $S=-2$, $V=1$, $\Lambda=1.057827\ldots$, the middle root of $4 x^3-30 x+27$, and $R \in \{\frac{1}{4}(1-2\Lambda)+\frac{n}{2}: n \in \Z\}$.
\end{theorem}
\begin{proof}
Let $f_\iota$ have $F_\iota$ as it field, $r_\iota$ as its shift, and $\ell_\iota$ as its length.
Let $S_\iota=S_{f_\iota,f_\iota}$, $U_\iota=U_{f_\iota,f_\iota}=1$ and $V_\iota=V_{f_\iota,f_\iota}$ be the parameters associated with $f_\iota$ as defined in \eqref{Albert}.
The assumptions of Corollary \ref{Ellen} tell us that $\{p_\iota: \iota \in I\}$ is infinite, and that there are real numbers $\Lambda>0$, $S$, and $V$ such that $\ell_\iota/p_\iota \to \Lambda$, $S_\iota \to S$, and $V_\iota \to V$ as $p_\iota \to \infty$.
And if $V \not=0$, then we are also given a real number $R$ such that $r_\iota/p_\iota \to R$ as $p_\iota \to \infty$.
Then Corollary \ref{Ellen} gives the value of the limiting demerit factor, which we call
\[
D=\lim_{p_\iota \to\infty} \DF(f_\iota) = \lim_{p_\iota \to\infty} \DF(f^u_\iota).
\]
Corollary \ref{Ellen} states that
\[
D+1=S\cdot \frac{2}{3}\Lambda + 2 \Omega\left(\frac{1}{\Lambda},0\right) + V \cdot \Omega\left(\frac{1}{\Lambda},1+\frac{2 R}{\Lambda}\right).
\]
Let $Q_\iota=S_\iota+1+U_\iota+V_\iota=S_\iota+2+V_\iota$ for each $\iota\in I$, and let $Q=S+2+V$, which is equal to $\lim_{p_\iota \to\infty} Q_\iota$.
Thus
\begin{equation}\label{Regina}
D+1 = (Q-V-2) \cdot \frac{2}{3}\Lambda + 2 \Omega\left(\frac{1}{\Lambda},0\right) + V \cdot \Omega\left(\frac{1}{\Lambda},1+\frac{2 R}{\Lambda}\right).
\end{equation}
By Proposition \ref{Hilda}, we see that $Q_\iota \geq 1$ for all $\iota \in I$, so then $Q \geq 1$, and so
\[
D+1 \geq (-V-1) \cdot \frac{2}{3}\Lambda  + 2 \Omega\left(\frac{1}{\Lambda},0\right) + V \cdot \Omega\left(\frac{1}{\Lambda},1+\frac{2 R}{\Lambda}\right),
\]
and since $\Lambda > 0$, equality is achievable if and only if $Q=1$.
Now \cite[Lemma 22]{Katz-2016-Aperiodic} tells us that for a fixed $\Lambda >0$, the function $\Omega(1/\Lambda,1+\frac{2 R}{\Lambda})$ achieves a global minimum value of $\Omega(1/\Lambda,1/(2\Lambda))$ when $R$ is chosen appropriately, so
\[
D+1 \geq (-V-1) \cdot \frac{2}{3}\Lambda  + 2 \Omega\left(\frac{1}{\Lambda},0\right) + V \cdot \Omega\left(\frac{1}{\Lambda},\frac{1}{2 \Lambda}\right).
\]
Then Lemma \ref{Justin} below tells us that $2 x/3-\Omega(1/x,1/(2 x)) > 0$ for all $x > 0$, so we see that the limiting demerit factor gets strictly smaller as $V$ increases.
Since $V_\iota \leq 1$ for all $\iota$ by \eqref{Barbara}, we see that $V \leq 1$, and so
\[
D+1 \geq -\frac{4}{3}\Lambda  + 2 \Omega\left(\frac{1}{\Lambda},0\right) + \Omega\left(\frac{1}{\Lambda},\frac{1}{2\Lambda}\right),
\]
and equality is achievable if and only if $Q=V=1$.

So we see that the global minimum of $D+1$ (if one exists) can only exist when $Q=V=1$ (or equivalently, $S=-2$ and $V=1$).
When $Q=V=1$, we return to \eqref{Regina} to see that
\[
D+1 = -\frac{4}{3}\Lambda + 2 \Omega\left(\frac{1}{\Lambda},0\right) + \Omega\left(\frac{1}{\Lambda},1+\frac{2 R}{\Lambda}\right).
\]
In \cite[Corollary 3.2]{Jedwab-Katz-Schmidt-2013-Littlewood}, it is shown that the function on the right hand side (for $\Lambda, R$ real numbers with $\Lambda > 0$) achieves a global minimum value with global minimizers as described in the statement of this theorem.  This minimum is obtainable using appended, shifted Legendre sequences and their unimodularizations (see \cite[Theorem 1.5]{Katz-2013-Asymptotic} and \cite[Corollary 3.2]{Jedwab-Katz-Schmidt-2013-Littlewood}).
\end{proof}
Theorem \ref{Richard} shows that having $S+2+V=1$ is a necessary condition for our sequence family $\{f_\iota\}_{\iota \in I}$ to achieve the maximum asymptotic autocorrelation merit factor.  Proposition \ref{Hilda} shows that this is tantamount to requiring that the quantity $\frac{1}{p_\iota(p_\iota-1)} \sum_{a \in \F_{p_\iota}} |\PC_{\per(f_\iota),\per(f_\iota)}(a)|^2$ (which is always at least $1$) tend to $1$ in the limit as $p_\iota \to \infty$.
Thus mean square periodic autocorrelation must be low in order to get low mean square aperiodic autocorrelation (high autocorrelation merit factor).  This provides a vindication of the ideas of Boehmer \cite[p.~157]{Boehmer-1967-Binary} quoted in the Introduction.

The maximum autocorrelation merit factor alluded to in Theorem \ref{Richard} is achieved with moderate appending (the sequences should have a limiting ratio of length to prime tending to $1.057827\ldots$, as noted in \cite{Jedwab-Katz-Schmidt-2013-Advances}).  Beyond this amount of appending, the autocorrelation merit factor drops because the periodic extension of the sequence causes large autocorrelation values at shifts that are multiples of the prime $p$, since the terms of the sequence repeat with period $p$.

For crosscorrelation, there is no upper bound on the merit factor of character combination sequences.
In Theorem \ref{Irene}, we shall see that we can make crosscorrelation merit factor increase without bound by using carefully selected sequences where the appending is such that each sequence repeats many periods.

We close this section with the following technical lemma was used in the proof of Theorem \ref{Richard}.
\begin{lemma}\label{Justin}
For all $x > 0$, we have $2 x/3 > \Omega(1/x,1/(2 x))$.
\end{lemma}
\begin{proof}
For $0 < x < 1/2$, we have $\Omega(1/x,1/(2 x))=0$ by \cite[Lemma 25]{Katz-2016-Aperiodic}, so we may assume $x \geq 1/2$ henceforth, and let $m$ be a positive integer chosen so that $x \in [m-1/2,m+1/2]$.  Then \cite[Lemma 25]{Katz-2016-Aperiodic} tells us that
\[
\Omega\left(\frac{1}{x},\frac{1}{2 x}\right) = 2 m- \frac{2 m^2}{x} + \frac{m(4 m^2-1)}{6 x^2},
\]
so it suffices to show that
\[
\frac{2 x}{3} - 2 m + \frac{2 m^2}{x} - \frac{m(4 m^2-1)}{6 x^2} > 0
\]
for $x \in [m-1/2,m+1/2]$.
Equivalently, it suffices to show that the polynomial $4 x^3-12 m x^2+12 m^2 x - 4 m^3+4 m=4(x-m)^3+m$ is strictly positive for $x \in [m-1/2,m+1/2]$.  This is so because $|x-m| \leq 1/2$ on our interval, making $|4(x-m)|^3 \leq 1/2$, and we are assuming that $m$ is a positive integer.
\end{proof}

\section{Sequences Derived from $2 m$th Order Residues}\label{Emily}

In this section, we construct a special class of character combination sequences based on power residues in finite fields.  These can be unimodularized to obtain (among many others) the sequences $\fpslu$, $\gpslu$, and $\hpslu$ in \eqref{Elizabeth} in the Introduction.  Throughout this section, we let $m$ be a fixed positive integer, and construct families of sequences indexed by the primes that are $1$ modulo $2 m$.

Suppose that $p$ is a prime with $p\equiv 1 \pmod{2 m}$.
Let $\Fputm$ be the subgroup of $\Fpu$ consisting of $2 m$th order residues, that is, $\Fputm=\{x^{2 m}: x \in \Fpu\}$, which is the unique subgroup of index $2 m$ in the cyclic group $\Fpu$.
Then the quotient group $\Fpu/\Fputm$ is a cyclic group of order $2 m$: if $\alpha$ is a primitive element of $\Fpu$, then $\Fpu/\Fputm$ is generated by the coset $\alpha\Fputm$, and consists of the cosets $\Fputm,\alpha\Fputm,\ldots,\alpha^{2 m-1}\Fputm$, which partition $\Fpu$.

Let ${\mathcal A}$ be a subset consisting of half the elements of $\Fpu/\Fputm$, that is, $\card{{\mathcal A}}=m$.  Consider the function $F_{p,{\mathcal A}}\colon \Fp \to \C$ with
\begin{equation}\label{Lisa}
F_{p,{\mathcal A}}(j)= \begin{cases}
0 & \text{if $j=0$,} \\
+1 & \text{if $j \in \bigcup_{A \in {\mathcal A}} A$}, \\
-1 & \text{if $j \in \Fpu\smallsetminus \bigcup_{A \in {\mathcal A}} A$}.
\end{cases}
\end{equation}
We regard $F_{p,{\mathcal A}}$ as a periodic sequence $(F_{p,{\mathcal A}}(j))_{j \in \Fp}$.
For $s,\ell \in \Z$ with $\ell >0$, we define $f_{p,{\mathcal A}}^{s,\ell}$ to be the (non-periodic) sequence
\begin{equation}\label{Fred}
f_{p,{\mathcal A}}^{s,\ell}= (F_{p,{\mathcal A}}(s),\ldots,F_{p,{\mathcal A}}(s+\ell-1)).
\end{equation}
We call $f_{p,{\mathcal A}}^{s,\ell}$ the {\it $2 m$th residue class sequence with prime $p$, classes ${\mathcal A}$, shift $s$, and length $\ell$}.

\begin{example}\label{Earl}
{\em Note that if we set $m=2$, let $p$ be a prime with $p\equiv 1 \pmod{4}$, let $\alpha$ be a primitive element of $\Fpu$, and let ${\mathcal A}=\{\Fpuf,\alpha\Fpuf\}$, then $F_{p,{\mathcal A}}(j)$ in \eqref{Lisa} is almost the same as $F_p$ defined in \eqref{Nelson} in the Introduction, with the only difference being that $F_{p,{\mathcal A}}(0)=0$ while $F_p(0)=1$.  We then let $\fpsl$ be $f_{p,{\mathcal A}}^{s,\ell}$ as in \eqref{Fred}.  Then $\fpslu$ from \eqref{Elizabeth} in the Introduction is just the unimodularization of $\fpsl=f_{p,{\mathcal A}}^{s,\ell}$ where we replace every $0$ term with a $1$.

Similarly, we let $\gpsl$ be $f_{p,{\mathcal B}}^{s,\ell}$ where ${\mathcal B}=\{\Fpuf,\alpha^3\Fpuf\}$, and then $\gpslu$ from \eqref{Elizabeth} in the Introduction is a unimodularization of $\gpsl$.

And likewise, $\hpslu$ from \eqref{Elizabeth} in the Introduction is a unimodularization of $f_{p,{\mathcal D}}^{s,\ell}$ where ${\mathcal D}=\{\Fpuf,\alpha^2\Fpuf\}$.

To obtain quadratic residue sequences based on quadratic characters modulo all odd primes $p$ (rather than just those primes that are $1$ modulo $4$), we instead set $m=1$, let $p$ be a prime with $p\equiv 1 \pmod{2}$, let ${\mathcal E}=\{\Fput\}$, and then let $\hpsl$ be $f_{p,{\mathcal E}}^{s,\ell}$ from \eqref{Fred}.  This $\hpsl$ is a generic quadratic residue sequence, which can be unimodularized to produce $\hpslu$, a Legendre sequence with shift $s$ and modified length $\ell$.  When $p$ is $1$ modulo $4$, this agrees with the definition of $\hpslu$ in the Introduction.}
\end{example}

Returning to $2 m$th residue class sequences in general, we claim that our sequence $f_{p,{\mathcal A}}^{s,\ell}$ as described in \eqref{Fred} is always a character combination sequence.
We prove this in the next lemma, but we need to set some notation first.
The group $\mchars$ of multiplicative characters of $\Fp$ is a cyclic group of order $p-1$, and so it has a unique cyclic group of order $2 m$ which we call $\Theta_p$.  If $\chi \in \Theta_p$, then we note that $\chi(x)=1$ for any $x \in \Fputm$ since $x=y^{2 m}$ for some $y \in \Fpu$, and thus $\chi(x)=\chi(y^{2 m})=\chi^{2 m}(y)=1$ since $\chi^{2 m}$ is the trivial character.  So $\chi$ has a constant value on any coset of $\Fputm$.  By abuse of notation, if $A$ is such a coset, that is, if $A \in \Fpu/\Fputm$, then $\chi(A)$ is just the value of $\chi(a)$ for any $a \in A$.  With this notation, and recalling that $\Fpu/\Fputm$ is a group, we see that $\chi(A B)=\chi(A)\chi(B)$ and $\chi(A^{-1})=\chi(A)^{-1}=\conj{\chi}(A)$ for any $A, B \in \Fpu/\Fputm$.

If $\alpha$ is a primitive element of $\Fpu$, then $\Fpu/\Fputm$ consists of the $2 m$ cosets $\alpha^0\Fputm,\ldots,\alpha^{2 m-1}\Fputm$.  And $\Theta_p$ consists of the $2 m$ characters $\theta_0,\ldots,\theta_{2 m-1}$, where $\theta_j(\alpha^k)=\exp(\pi i j k/m)$ for any $k \in \Z$.  Thus, using our notation, we have  $\theta_j(\alpha^k \Fputm)=\exp(\pi i j k/m)$, and from this it is not hard to show that for any $A \in \Fpu/\Fputm$, we have
\begin{equation}\label{William}
\frac{1}{2 m} \sum_{\chi \in \Theta_p} \chi(A) =\begin{cases}
1 & \text{if $A=\Fputm$}, \\
0 & \text{otherwise.}
\end{cases}
\end{equation}
\begin{lemma}\label{Veronica}
Let $m$ be a positive integer, let $p$ be a prime with $p\equiv 1 \pmod{2 m}$, let ${\mathcal A}$ be a subset of $\Fpu/\Fputm$ with $\card{{\mathcal A}}=m$, and let $f_{p,{\mathcal A}}^{s,\ell}$ be a $2 m$th residue class sequence with prime $p$, classes ${\mathcal A}$, shift $s$, and length $\ell$.
Let $\Theta_p$ be the unique subgroup of order $2 m$ in $\mchars$.
Then $f_{p,{\mathcal A}}^{s,\ell}$ is a unimodularizable character combination sequence with character combination $\{f_\chi\}_{\chi\in \mchars}$, where
\[
f_{\chi} = \begin{cases}
\frac{1}{m} \sum_{A \in {\mathcal A}} \conj{\chi}(A) & \text{if $\chi$ is nontrivial and $\chi \in \Theta_p$}, \\
0 & \text{otherwise},
\end{cases}
\]
and we have $f_{\conj{\chi}}=\conj{f_\chi}$ for all $\chi \in \mchars$.
\end{lemma}
\begin{proof}
Let $F_{p,{\mathcal A}} \colon \Fp \to \C$ be as defined in \eqref{Lisa}, so that
\[
f_{p,{\mathcal A}}^{s,\ell}=(F_{p,{\mathcal A}}(s),\ldots,F_{p,{\mathcal A}}(s+\ell-1)),
\]
and let $\chi_0$ denote the trivial character in $\mchars$.

Then we claim that
\begin{equation}\label{Ursula}
F_{p,{\mathcal A}}(j)=\sums{\chi\in\Theta_p \\ \chi\not=\chi_0} \left(\frac{1}{m} \sum_{A \in {\mathcal A}} \conj{\chi}(A)\right) \chi(j).
\end{equation}
The formula is clearly correct for $j=0$, so we assume $j \in \Fpu$ henceforth.  By \eqref{William}, we see that
\[
\sum_{A \in {\mathcal A}} \sum_{\chi \in \Theta_p} \conj{\chi}(A) \chi(j) = \sum_{A \in {\mathcal A}} \sum_{\chi\in\Theta_p} \conj{\chi}(j^{-1} A),
\]
is $2 m$ if and only if $j \in \bigcup_{A \in {\mathcal A}} A$; otherwise it is $0$.
Thus, if we scale by $1/m$ and subtract $1$, we obtain
\[
F_{p,{\mathcal A}}(j)=-1 + \sum_{\chi \in \Theta_p} \left(\frac{1}{m} \sum_{A \in {\mathcal A}} \conj{\chi}(A)\right) \chi(j).
\]
Now the trivial character $\chi_0$ contributes $1$ to the double sum regardless of the value of $j \in \Fpu$, so we obtain \eqref{Ursula}.
So the terms of $f_{p,{\mathcal A}}^{s,\ell}$ are indeed obtained from a linear combination of multiplicative characters with the coefficients as claimed in this lemma.

To see that $f_{p,{\mathcal A}}^{s,\ell}$ is a character combination sequence, we need to check that the sum of the squared magnitudes of these coefficients is $1$, so we compute
\begin{align*}
\sums{\chi\in\Theta_p \\ \chi\not=\chi_0} \left|\frac{1}{m} \sum_{A \in {\mathcal A}} \conj{\chi}(A)\right|^2
& = \frac{1}{m^2} \sums{\chi \in \Theta_p \\ \chi \not=\chi_0} \sum_{A,B \in {\mathcal A}} \chi(A^{-1} B) \\
& = \frac{1}{m^2} \sum_{A,B \in {\mathcal A}} (-1+ 2 m \delta_{A,B}) \\
& = 1,
\end{align*}
where $\delta$ is the Kronecker delta, and where we have used \eqref{William} in the second step, and the fact that $\card{{\mathcal A}}=m$ in the third step.
So $f_{p,{\mathcal A}}^{s,\ell}$ is indeed a character combination sequence.

It is clear from the formula for the coefficients that $f_{\conj{\chi}}=\conj{f_\chi}$ for all $\chi \in \mchars$, and it is clear that $f_{p,{\mathcal A}}^{s,\ell}$ is unimodularizable from the definition of $F_{p,{\mathcal A}}$.
\end{proof}
\begin{example}\label{Ethel}
{\em We continue Example \ref{Earl} (keeping the same notation), and now compute the character combinations for the sequences $\fpsl$, $\gpsl$, and $\hpsl$ defined there using the formula in Lemma \ref{Veronica}.  For instance, our $\fpsl$ is $f_{p,{\mathcal A}}^{s,\ell}$ where ${\mathcal A}=\{\Fpuf,\alpha\Fpuf\}$.  We need to use $\Theta_p$, the unique subgroup of order $4$ in $\mchars$.  If $\alpha$ is a primitive element of $\Fpu$, then one quartic (order $4$) character is $\theta\colon \Fpu \to \{\pm 1,\pm i\}$ where $\theta(\alpha^k)=i^k$.  Then $\Theta_p=\{\theta^0,\theta^1,\theta^2,\theta^3\}$, and $\theta^0$ is the trivial character, $\theta^2$ is the quadratic character, and $\theta^3$ is the other quartic character.  Note that $\theta^j(\alpha^k \Fpuf)=i^{j k}$, so then one uses the formula for the character combination coefficients in Lemma \ref{Veronica} to find that $\fpsl$ has character combination $\{f_\chi\}_{\chi\in\mchars}$ with
\[
f_\chi = \begin{cases}
\frac{1-i}{2} & \text{if $\chi=\theta$},\\
\frac{1+i}{2} & \text{if $\chi=\conj{\theta}$},\\
0 & \text{otherwise}.
\end{cases}
\]

Similarly, one can compute that the character combination for $\gpsl$ is $\{g_\chi\}_{\chi\in\mchars}$ with
\[
g_\chi = \begin{cases}
\frac{1+i}{2} & \text{if $\chi=\theta$},\\
\frac{1-i}{2} & \text{if $\chi=\conj{\theta}$},\\
0 & \text{otherwise}.
\end{cases}
\]

Each field $\Fp$ with $p$ an odd prime has a single quadratic character $\eta$, and $\eta$ maps the elements of $\Fput$ to $1$ and the other elements of $\Fpu$ to $-1$.  Given the way it is defined, $\hpsl$ transparently has character combination $\{h_\chi\}_{\chi\in\mchars}$ with
\[
h_\chi = \begin{cases}
1 & \text{if $\chi=\eta$},\\
0 & \text{otherwise},
\end{cases}
\]
or one can easily use the formula of Lemma \ref{Veronica} to compute these coefficients.}
\end{example}

Let us return to the general theory.  Suppose that we have a pair $(f,g)$ of $2 m$th residue class sequences, and we want to compute the parameters in \eqref{Patrick} for pairs of $2 m$th residue class sequences.  The parameter $S_{f,g}$ will turn out to be complicated, and often dependent on the prime.  For our simple cases, it will be easier to compute directly from \eqref{Patrick}.  But $U_{f,g}$ and $V_{f,g}$ have interesting combinatorial formulae.
If $m$ is a positive integer and $A \in \Fpu/\Fputm$, then we use the notation $-A$ to denote the coset $(-1)A$, and if ${\mathcal A} \subseteq \Fpu/\Fputm$, then by extension, we use $-{\mathcal A}$ to denote $\{-A: A \in {\mathcal A}\}$.
\begin{lemma}\label{Martha}
Let $m$ be a positive integer, let $p$ be a prime with $p\equiv 1 \pmod{2 m}$, let ${\mathcal A}$ and ${\mathcal B}$ be subsets of $\Fpu/\Fputm$ with $\card{{\mathcal A}}=\card{{\mathcal B}}=m$, and let $f$ and $g$ be $2 m$th residue class sequences with classes ${\mathcal A}$ and ${\mathcal B}$, respectively, and both with prime $p$.
Let $U_{f,g}$, $V_{f,g}$, $W_f$, and $W_g$ be the parameters for the pair $(f,g)$ as defined in \eqref{Patrick}.  Then
\begin{align*}
U_{f,g} & = \left(\frac{2}{m} \card{{\mathcal A}\cap {\mathcal B}}-1\right)^2 \\
V_{f,g} & = \begin{cases} \left(\frac{2}{m} \card{{\mathcal A}\cap {\mathcal B}}-1\right)^2 & \text{if $p\equiv 1 \pmod{4 m}$} \\ \left(\frac{2}{m}\card{{\mathcal A}\cap(-{\mathcal B})}-1\right)^2 & \text{if $p\not\equiv 1 \pmod{4 m}$}\end{cases} \\
1\leq W_f & \leq \sqrt{2 m-1} \\
1\leq W_g & \leq \sqrt{2 m-1}.
\end{align*}
If $p\not\equiv 1 \pmod{4 m}$ and $\alpha$ is a primitive element of $\Fp$, then $-B=\alpha^m B$ for every $B \in \Fputm/\Fpu$.

In the special case where $f=g$, we have
\begin{align*}
U_{f,f} & = 1 \\
V_{f,f} & = \begin{cases} 1 & \text{if $p\equiv 1 \pmod{4 m}$}, \\ \left(\frac{2}{m}\card{{\mathcal A}\cap(-{\mathcal A})}-1\right)^2 & \text{if $p\not\equiv 1 \pmod{4 m}$}.\end{cases}
\end{align*}
\end{lemma}
\begin{proof}
Let $\{f_\chi\}_{\chi \in \mchars}$ and $\{g_\chi\}_{\chi\in\mchars}$ be the respective character combinations of $f$ and $g$.  Let $\Theta_p$ be the unique subgroup of order $2 m$ in $\mchars$ and let $\chi_0$ be the trivial character.
First of all we know from Lemma \ref{Veronica} that $f_\chi=0$ when $\chi\not\in\Theta_p$ or when $\chi=\chi_0$, and our sequences are normalized so that $\sums{\chi\in\Theta_p \\ \chi\not=\chi_0} |f_\chi|^2=1$, which is a sum of $2 m-1$ squared magnitudes.  Given this constraint, the sum of magnitudes is minimized when only one magnitude is positive, and is maximized when all magnitudes are equal, so $W_f=\sums{\chi\in\Theta_p\\ \chi\not=\chi_0} |f_\chi|$ lies between $1$ and $\sqrt{2 m-1}$.  The same holds for $W_g$.

We use the values of $f_\phi$ and $g_\phi$ computed in Lemma \ref{Veronica} to obtain
\begin{align*}
V_{f,g}
& = \left|\frac{1}{m^2} \sums{\phi\in\Theta_p \\ \phi\not=\chi_0} \sums{A \in {\mathcal A}\\ B \in {\mathcal B}} \phi(-A^{-1} B)\right|^2 \\
& = \left|\frac{1}{m^2} \sums{A \in {\mathcal A}\\ B \in {\mathcal B}} (-1 +2 m \delta_{A,-B}) \right|^2 \\
& = \left|-1 + \frac{2}{m} \card{{\mathcal A} \cap (-{\mathcal B})}\right|^2.
\end{align*}
where $\delta$ is the Kronecker delta, and we have used \eqref{William} in the second equality, and the fact that $\card{{\mathcal A}}=\card{{\mathcal B}}=m$ in the third, at which point we realize that the quantity inside the absolute value is a real number, thus proving our identity for $V_{f,g}$.
Note that when $p\equiv 1 \pmod{4 m}$, the element $-1 \in \Fpu$ is a $2 m$th power, so then $-B=B$ for all $B \in {\mathcal B}$.  But if $p\not\equiv 1 \pmod{4 m}$, then $-1$ is an $m$th power but not a $2 m$th power, and so if $\alpha$ is a primitive element of $\Fp$, we see that $-\Fputm=\alpha^m\Fputm$, and thus $-B=\alpha^m B$ for any $B \in \Fpu/\Fputm$.

Lemma \ref{Veronica} tells us that $g_{\conj{\phi}}=\conj{g_\phi}$ for every $\phi\in\mchars$, so the calculation of $U_{f,g}$ is similar to that for $V_{f,g}$, but lacks a negative sign in the argument of $\phi$.  We thus obtain
\[
U_{f,g} = \sum_{\phi\in\mchars} |f_\phi g_{\conj{\phi}}|^2 = \left(-1+\frac{2}{m} \card{{\mathcal A} \cap {\mathcal B}}\right)^2.
\]
In the case where $f=g$, we have ${\mathcal A}={\mathcal B}$, and recall that $\card{{\mathcal A}}=m$, so that the results for $U_{f,f}$ and $V_{f,f}$ follow.
\end{proof}
\begin{example}\label{Eveline}
{\em We continue Examples \ref{Earl} and \ref{Ethel} (keeping the same notation), and now compute the $U$ and $V$ parameters from \eqref{Patrick} and \eqref{Albert} for pairs of sequences of the form $f=\fpsl$, $g=\gpsl$, and $h=\hpsl$ defined in Example \ref{Earl}.  Recall that when our prime is $p\equiv 1 \pmod{4}$ and $\alpha$ is a primitive element of $\Fp$, we can set $m=2$ and regard $f=\fpsl$, $g=\gpsl$, and $h=\hpsl$ as $4$th residue class sequences with classes ${\mathcal A}=\{\Fpuf,\alpha\Fpuf\}$, ${\mathcal B}=\{\Fpuf,\alpha^3\Fpuf\}$, and ${\mathcal D}=\{\Fpuf,\alpha^2\Fpuf\}$, respectively.

For any odd prime where we are only considering pairs of quadratic residue sequences, we can instead set $m=1$ and think of $h=\hpsl$ as the $2$nd residue sequence with class $\{\Fput\}$.

In each of these cases, we use the combinatorial formulae in Lemma \ref{Martha} to obtain
\begin{center}
\begin{equation*}
\begin{tabular}{cc}
$U_{f,f} = U_{g,g} = 1$ & $U_{h,h}=1$ \\
$V_{f,f} = V_{g,g} = 1$ & $V_{h,h}=1$ \\
& \\
& \\
$U_{f,g} = 0$ & $U_{f,h}=U_{g,h}=0$ \\
$V_{f,g} = 0$ & $V_{f,h}=V_{g,h}=0$,
\end{tabular}
\end{equation*}
\end{center}
regardless of whatever $p\equiv 1 \pmod{8}$ or $p\not\equiv 1 \pmod{8}$.}
\end{example}

\section{Sequences Derived from Quadratic and Quartic Residues}\label{Gilbert}

In this section, we study the sequences $\fpslu$, $\gpslu$, and $\hpslu$ defined in \eqref{Elizabeth} in the Introduction.
We constructed sequences $\fpsl$, $\gpsl$, and $\hpsl$ in Examples \ref{Earl}, \ref{Ethel}, and \ref{Eveline} such that $\fpslu$, $\gpslu$, and $\hpslu$ can be obtained from $\fpsl$, $\gpsl$, and $\hpsl$, respectively, by unimodularizing (in particular, replacing each zero term with a $1$).
The sequences $\fpsl$ and $\gpsl$ exist for primes $p$ with $p\equiv 1\pmod{4}$.  If $\alpha$ is a primitive element of $\Fp$, then we define the functions
\begin{align*}
\tilde{F}_p(x) & = \begin{cases}
+1 & \text{if $x \in \Fpuf \cup \alpha\Fpuf$} \\
-1 & \text{if $x \in \alpha^2\Fpuf \cup \alpha^3\Fpuf$} \\
 0 & \text{if $x=0$},
\end{cases} \\
\tilde{G}_p(x) & = \begin{cases}
+1 & \text{if $x \in \Fpuf \cup \alpha^3\Fpuf$} \\
-1 & \text{if $x \in \alpha\Fpuf \cup \alpha^2\Fpuf$} \\
0 & \text{if $x=0$},
\end{cases}
\end{align*}
and then
\begin{align}
\fpsl & =(\tilde{F}_p(s),\ldots,\tilde{F}_p(s+\ell-1)) \label{Florence} \\
\gpsl & =(\tilde{G}_p(s),\ldots,\tilde{G}_p(s+\ell-1)). \nonumber
\end{align}
And $\hpsl$ exists for all odd primes $p$, and if $\alpha$ is a primitive element of $\Fp$, then we define the function
\[
\tilde{H}_p(x) = \begin{cases}
+1 & \text{if $x \in \Fput$} \\
-1 & \text{if $x \in \alpha\Fput$} \\
0 & \text{if $x=0$},
\end{cases}
\]
and then
\begin{equation}\label{Ignatius}
\hpsl  =(\tilde{H}_p(s),\ldots,\tilde{H}_p(s+\ell-1)).
\end{equation}
Note that the functions $\tilde{F}_p$, $\tilde{G}_p$, and $\tilde{H}_p$ here differ from $F_p$, $G_p$, and $H_p$ in \eqref{Nelson} in the Introduction only in that $\tilde{F}_p(0)=\tilde{G}_p(0)=\tilde{H}_p(0)=0$ while $F_p(0)=G_p(0)=H_p(0)=1$.  So we see that $\fpslu$, $\gpslu$, and $\hpslu$ in \eqref{Elizabeth} in the Introduction are obtained from $\fpsl$, $\gpsl$, and $\hpsl$ here by a unimodularization which replaces any term $0$ with a $1$.

We have already computed the coefficients of the character combinations for $f=\fpsl$, $g=\gpsl$, and $h=\hpsl$ in Example \ref{Ethel} above.  In Example \ref{Eveline}, we found the $U$ and $V$ parameters from \eqref{Patrick} and \eqref{Albert} for all the mutual pairings of these sequences.  We now use the character combination coefficients from Example \ref{Ethel} to compute directly the $S$ and $W$ parameters from \eqref{Patrick} and \eqref{Albert}.  
In handling the $S$ parameter, we use the fact that $\tau(\conj{\chi})=\chi(-1) \conj{\tau(\chi)}$ from \eqref{Katherine}; other than this, the calculations are routine, and we obtain the following:
\begin{center}
\begin{equation}\label{Judith}
\begin{tabular}{cc}
& \\
$S_{f,f} = S_{g,g} = \frac{-3-\Re(\tau(\theta_p)^4/p^2)}{2}$ &  $S_{h,h}=-2$ \\
$U_{f,f} = U_{g,g} = 1$ & $U_{h,h}=1$ \\
$V_{f,f} = V_{g,g} = 1$ & $V_{h,h}=1$ \\
$W_f    = W_g = \sqrt{2}$ & $W_h=1$ \\
& \\
& \\
$S_{f,g} = \frac{-1+\Re(\tau(\theta_p)^4/p^2)}{2}$ & $S_{f,h}=S_{g,h}=0$ \\
$U_{f,g} = 0$ & $U_{f,h}=U_{g,h}=0$ \\
$V_{f,g} = 0$ & $V_{f,h}=V_{g,h}=0$ \\
& \\
\end{tabular}
\end{equation}
\end{center}
The quantity $\Re(\tau(\theta_p)^4/p^2)$ depends on the specific prime.
Each prime $p$ with $p\equiv 1 \pmod{4}$ can be written as
\[
p=a_p^2+b_p^2
\]
for unique positive integers $a_p, b_p$ with $a_p$ odd and $b_p$ even.
This is a consequence of how primes factor into irreducible elements in the ring of Gaussian integers $\Z[i]$.
For such a prime, we define $\gamma_p$ to be the unique value in $(0,\pi/2)$ with
\begin{equation}\label{Dorothy}
a_p = \sqrt{p} \cos(\gamma_p) \qquad\qquad b_p = \sqrt{p} \sin(\gamma_p).
\end{equation}
\begin{lemma}\label{Theresa}
Let $p$ be a prime with $p\equiv 1 \pmod{4}$ and suppose that $a_p, b_p, \gamma_p$ are as defined above.  If $\theta_p$ is a quartic character of $\Fpu$, then
\[
\Re\left(\frac{\tau(\theta_p)^4}{p^2}\right) = \Re\left(\frac{\tau(\conj{\theta_p})^4}{p^2}\right) =\frac{a_p^2-b_p^2}{p}=\cos(2\gamma_p).
\]
The set of $\gamma_p$ for all primes $p$ congruent to $1$ modulo $4$ is equidistributed in the interval $(0,\pi/2)$.  Thus for any $\gamma \in [0,\pi/2]$, there exists an increasing sequence $\{p_\iota\}_{\iota\in I}$ of primes congruent to $1$ modulo $4$ such that $\gamma_{p_\iota} \to \gamma$ as $p_\iota \to \infty$.
\end{lemma}
\begin{proof}
The relation $\Re(\tau(\theta_p)^4/p^2)=\Re(\tau(\conj{\theta_p})^4/p^2)=(a_p^2-b_p^2)/p$ is obtained by combining equations (4.4) and (4.1) of \cite{Berndt-Evans-1981-Determination}.  This, in turn, is $\cos(2 \gamma_p)$ by the definition of $\gamma_p$ and the double angle formula for cosine.

Note that the normalized Gauss sums $\tau(\theta_p)/\sqrt{p}$ and $\tau(\conj{\theta_p})/\sqrt{p}$ lie on the complex unit circle by \eqref{Manuel}, and it has been shown that the set of points one obtains from these two normalized sums for all primes $p$ congruent to $1$ modulo $4$ is equidistributed on the unit circle \cite[\S 1]{Heath-Brown-Patterson-1979-Distribution}, \cite{Patterson-1981-Distribution}, \cite{Patterson-1987-Distribution}.  Thus $\gamma_p$ is equidistributed in $(0,\pi/2)$.
\end{proof}
This last result shows that the value in \eqref{Judith} of $S_{f,f}=S_{g,g}$ for our sequences $f=\fpsl$ and $g=\gpsl$ varies between $-2$ and $-1$ depending on the prime, and by a careful selection of primes $p$ with $\gamma_p$ tending to $0$, we can make infinite families of sequences $\{f_\iota\}_{\iota \in I}$ where $S_{f_\iota,f_\iota}$ tends to any value in this range as the prime tends to infinity.  If we let these $S_{f_\iota,f_\iota}$ tend to $-2$, we obtain the same asymptotic autocorrelation behavior as we do for the quadratic residue sequences $h=\hpsl$.  (Note that $S_{h,h}$ is always $-2$, regardless of the prime $p$.)  Thus we can obtain the same maximum asymptotic autocorrelation merit factor of $6.342061\ldots$ with quartic residue sequences as we do with quadratic residue sequences. (We must, of course, choose the appropriate lengths and shifts, as specified in Theorem \ref{Richard}.)  However, for applications where the sole concern is low autocorrelation, quadratic residue sequences are preferable, because one does not need to select primes $p$ with $\gamma_p \to 0$ in order to approach this maximum.  In practical applications, this means that for quartic residue sequences, very low autocorrelation will only be manifest for certain primes, while quadratic residue sequences should reliably produce very low autocorrelation regardless of the prime (with the possible exception of very short sequences).

We now apply Corollary \ref{Ellen} to determine asymptotic autocorrelation merit factors for our quadratic and quartic residue sequences.  This involves substituting the parameters from \eqref{Judith} into Corollary \ref{Ellen}, and (for quartic residue sequences), adding a condition on the primes to make the $S$ parameter tend to a limit.
The first result for quadratic character sequences recapitulates \cite[eq.~(1.4)]{Hoholdt-Jensen-1988-Determination}, \cite[Theorem 1.5]{Katz-2013-Asymptotic}, \cite[Corollary 3.2]{Jedwab-Katz-Schmidt-2013-Littlewood}, and \cite[Theorem 2.1.(i)]{Jedwab-Katz-Schmidt-2013-Advances}.
\begin{theorem}\label{Rebecca}
Let $\{h_\iota\}_{\iota \in I}$ be a family of quadratic residue sequences, with each $h_\iota$ of the form $\tilde{h}_{p_\iota}^{r_\iota,\ell_\iota}$ as described in \eqref{Ignatius}.
Suppose that $\{p_\iota\}_{\iota \in I}$ is infinite and that there are real numbers $\Lambda >0$ and $R$ such that $\ell_\iota/p_\iota \to \Lambda$ and $r_\iota/p_\iota \to R$ as $p_\iota \to \infty$.  Then
\[
\DF(f_\iota) \to -1 - \frac{4}{3}\Lambda + 2 \Omega\left(\frac{1}{\Lambda},0\right) + \Omega\left(\frac{1}{\Lambda},1+\frac{2 R}{\Lambda}\right)
\]
as $p_\iota \to \infty$.
This limit is always at least $0.157677\ldots$, the smallest root of $27 x^3+417 x^2+249 x+29$, or equivalently, the asymptotic merit factor is no greater than $6.342061\ldots$, the largest root of $29 x^3+249 x+417 x + 27$.
This limiting merit factor is achieved if and only if $\Lambda=1.057827\ldots$, the middle root of $4 x^3-30 x+27$ and $R \in \{\frac{1}{4}(1-2\Lambda)+\frac{n}{2}: n \in \Z\}$.

If we fix $\Lambda=1$, then the asymptotic demerit factor is always at least $1/6$, or equivalently, the asymptotic merit factor is no greater than $6$, and this value is achieved if and only if $R \in \{(2 n+1)/4: n \in \Z\}$.

If $h^u_\iota$ is a unimodularization of $h_\iota$ for each $\iota \in I$ (for example, if $h^u_\iota=h_{p_\iota}^{r_\iota,\ell_\iota}$ as described in \eqref{Elizabeth}\em{)}, then $\DF(h^u_\iota)$ has the same limit as $\DF(h_\iota)$.
\end{theorem}
\begin{theorem}\label{Anne}
Let $\{f_\iota\}_{\iota \in I}$ be a family of quartic residue sequences, with each $f_\iota$ of the form $\tilde{f}_{p_\iota}^{r_\iota,\ell_\iota}$ or $\tilde{g}_{p_\iota}^{r_\iota,\ell_\iota}$ as described in \eqref{Florence}.
Suppose that $\{p_\iota\}_{\iota \in I}$ is infinite, and for each $p_\iota$, let $\gamma_{p_\iota}$ be as defined in \eqref{Dorothy}.
Suppose that there are real numbers $\Lambda > 0$, $R$, and $\gamma \in [0,\pi/2]$ such that $\ell_\iota/p_\iota \to \Lambda$, $r_\iota/p_\iota\to R$, and $\gamma_p \to \gamma$ as $p_\iota \to \infty$.  Then
\[
\DF(f_\iota) \to -1 - \frac{3+\cos(2\gamma)}{3}\Lambda + 2 \Omega\left(\frac{1}{\Lambda},0\right) + \Omega\left(\frac{1}{\Lambda},1+\frac{2 R}{\Lambda}\right)
\]
as $p_\iota \to \infty$.
This limit is always at least $0.157677\ldots$, the smallest root of $27 x^3+417 x^2+249 x+29$, or equivalently, the asymptotic merit factor is no greater than $6.342061\ldots$, the largest root of $29 x^3+249 x+417 x + 27$.
This limiting merit factor is achieved if and only if $\Lambda=1.057827\ldots$, the middle root of $4 x^3-30 x+27$ and $R \in \{\frac{1}{4}(1-2\Lambda)+\frac{n}{2}: n \in \Z\}$ and $\gamma=0$.

If we fix $\Lambda=1$, then the asymptotic demerit factor is always at least $1/6$, or equivalently, the asymptotic merit factor is no greater than $6$, and this value is achieved if and only if $R \in \{(2 n+1)/4: n \in \Z\}$ and $\gamma=0$.

If $f^u_\iota$ is a unimodularization of $f_\iota$ for each $\iota \in I$ (for example, if $f^u_\iota=f_{p_\iota}^{r_\iota,\ell_\iota}$ or $g_{p_\iota}^{r_\iota,\ell_\iota}$ as described in \eqref{Elizabeth} when $f_\iota$ is respectively $\tilde{f}_{p_\iota}^{r_\iota,\ell_\iota}$ or $\tilde{g}_{p_\iota}^{r_\iota,\ell_\iota}$ as described in \eqref{Florence}\em{)}, then $\DF(f^u_\iota)$ has the same limit as $\DF(f_\iota)$.
\end{theorem}
\begin{proof}
We apply Theorem \ref{Mary} with the parameters from \eqref{Judith}, and note that we meet the optimality conditions of Theorem \ref{Richard} only when $\gamma=0$.  When we fix $\Lambda=1$, then it is clear that using $\gamma=0$ always makes the limiting demerit factor smaller than it would otherwise be, and in that case, we have $\DF(f_\iota) \to -1/3 + \Omega(1,1+2 R)$.  Then we apply \cite[Lemma A.4.(ii)]{Katz-2013-Asymptotic} to see that this function achieves a minimum value of $1/6$ precisely when $R \in \{(2 n+1)/4: n \in \Z\}$.
\end{proof}
Now we proceed to crosscorrelation.  
We apply Theorem \ref{Mary} to calculate asymptotic crosscorrelation merit factors for our quadratic and quartic residue sequences, and present the results below in Theorems \ref{Clarence} and \ref{Christopher}.
This involves substituting the parameters from \eqref{Judith} into Theorem \ref{Mary}, and (in Theorem \ref{Christopher}) adding a condition on the primes to make the $S$ parameter tend to a limit.

Although quartic residue sequences are less useful than quadratic residue sequences from the point of view of autocorrelation, they enable us to obtain sequence pairs with very low crosscorrelation.
If we crosscorrelate our quartic residue sequences with quadratic residue sequences, the performance is average, as seen in Theorem \ref{Clarence}: we achieve asymptotic crosscorrelation demerit factor $1$, which is on par with the performance of random sequences (whose average crosscorrelation demerit factor at any given length is $1$ per Sarwate \cite[eq.~(38)]{Sarwate-1984-Mean}).
But if we crosscorrelate pairs of quartic residue sequences with each other, then we may obtain considerably lower crosscorrelation demerit factors, as seen in Theorem \ref{Christopher}.
\begin{theorem}\label{Clarence}
Let $\{(f_\iota,h_\iota)\}_{\iota \in I}$ be family of sequence pairs, with each $f_\iota$ a quartic residue sequence of the form $\tilde{f}_{p_\iota}^{r_\iota,\ell_\iota}$ or $\tilde{g}_{p_\iota}^{r_\iota,\ell_\iota}$ as described in \eqref{Florence}, and each $h_\iota$ a quadratic residue sequence of the form $\tilde{h}_{p_\iota}^{s_\iota,\ell_\iota}$ as described in \eqref{Ignatius}.
Suppose that $\{p_\iota\}_{\iota \in I}$ is infinite, and that there is a positive real number $\Lambda$ such that $\ell_\iota/p_\iota \to \Lambda$ as $p_\iota \to \infty$.  Then
\[
\CDF(f_\iota,h_\iota) \to \Omega\left(\frac{1}{\Lambda},0\right)
\]
as $p_\iota \to \infty$, which achieves a global minimum value of $1$, and this occurs if and only if $\Lambda \leq 1$.
If $f^u_\iota$ and $h^u_\iota$ are unimodularizations of $f_\iota$ and $h_\iota$, respectively, for each $\iota \in I$, then $\CDF(f^u_\iota,h^u_\iota)$ has the same limit as $\CDF(f_\iota,h_\iota)$.
\end{theorem}
\begin{proof}
We apply Theorem \ref{Mary} with the parameters from \eqref{Judith}.  The global minimization of $\Omega(1/\Lambda,0)$ is done in \cite[Lemma 24]{Katz-2016-Aperiodic}.
\end{proof}
\begin{theorem}\label{Christopher}
Let $\{(f_\iota,g_\iota)\}_{\iota \in I}$ be family of pairs of quartic residue sequences, with $f_\iota=\tilde{f}_{p_\iota}^{r_\iota,\ell_\iota}$ and $g_\iota=\tilde{g}_{p_\iota}^{s_\iota,\ell_\iota}$ as described in \eqref{Florence}.
Suppose that $\{p_\iota\}_{\iota \in I}$ is infinite, and for each $p_\iota$, let $\gamma_{p_\iota}$ be as defined in \eqref{Dorothy}.
Suppose that there are positive real numbers $\Lambda >0$ and $\gamma \in [0,\pi/2]$ such that $\ell_\iota/p_\iota \to \Lambda$ and $\gamma_p \to \gamma$ as $p_\iota \to \infty$.  Then
\[
\CDF(f_\iota,g_\iota) \to \frac{-1+\cos(2\gamma)}{3} \Lambda + \Omega\left(\frac{1}{\Lambda},0\right)
\]
as $p_\iota \to \infty$.  
If $\Lambda=1$, this equals $\frac{2+\cos(2\gamma)}{3}$ and so achieves a minimum value of $1/3$ when $\gamma=\pi/2$.  
If $f^u_\iota$ and $g^u_\iota$ are unimodularizations of $f_\iota$ and $g_\iota$, respectively, for each $\iota \in I$ (for example, if $f^u_\iota=f_{p_\iota}^{r_\iota,\ell_\iota}$ and $g^u_\iota=g_{p_\iota}^{r_\iota,\ell_\iota}$ as described in \eqref{Elizabeth}\em{)}, then $\CDF(f^u_\iota,g^u_\iota)$ has the same limit as $\CDF(f_\iota,g_\iota)$.
\end{theorem}
\begin{proof}
We apply Theorem \ref{Mary} with the parameters from \eqref{Judith}, and note that $\Omega(1,0)=1$.
\end{proof}

\section{Asymptotic Crosscorrelation Merit Factor is Unbounded}\label{James}
  
Appending character sequences to obtain lengths far beyond their natural length causes large autocorrelation sidelobe values at shifts that are multiples of the period of the original sequence, and as such is not recommended for most applications.
Nonetheless, it is interesting to see what happens to the mean-square crosscorrelation of the sequences $\fpslu$ and $\gpslu$ defined in \eqref{Elizabeth} in the Introduction when we append by more and more.

In the Introduction we appended these sequences to lengths $\ell$ equal to about $1.057827$ times the their natural period $p$.
Recall the definition of $\gamma_p$ in \eqref{Dorothy}, and recall from Lemma \ref{Theresa} that for any $\gamma \in [0,\pi/2]$, there is an increasing sequence of primes $p$ congruent to $1$ modulo $4$ such that $\gamma_p \to \gamma$ as $p\to\infty$.

If one further appends our sequences $\fpslu$ and $\gpslu$, the autocorrelation merit factor will decrease, but if we use sequences based on primes $p$ with $\cos(2 \gamma_p) \to -1$, then we can make the crosscorrelation merit factor increase without bound (i.e., demerit factor tends to $0$).
This is proved in Theorem \ref{Irene} below, and is illustrated in Figure \ref{Cecilia}.
In that figure we consider $100$ sequence pairs based on the smallest $100$ primes $p_1 < p_2 < \cdots < p_{100}$ of the form $1+(2 c)^2$ with $c \in \Z$ (so $\cos(2\gamma_{p_k})$ approaches $-1$ as $k$ increases).  The $k$th sequence pair is $(f_{p_k}^{s_k,\ell_k},g_{p_k}^{s_k,\ell_k})$ where $\ell_k$ is as close as possible to $k/10$ periods (that is, $\ell_k$ is $p_k\cdot k/10$, rounded to the nearest integer) and $s_k$ is chosen to keep autocorrelation relatively low  given the length $\ell_k$ (we let $s_k$ be $p_k \cdot (3-2(k/10))/4$, rounded to the nearest integer).  We say that the {\it fractional length} of $k$th sequence is $\ell_k/p_k=k/10$.  In Figure \ref{Cecilia}, the dots show the crosscorrelation demerit factors for these $100$ sequence pairs as a function of fractional length.  The curve in our figure indicates the asymptotic crosscorrelation demerit factor that we would obtain as a function of fractional length in the limit as $\cos(\gamma_p) \to -1$.
One can see that all but the shortest sequences have performance very close to asymptotic.

\begin{center}
\begin{figure}
\begin{center}
\caption{Crosscorrelation demerit factors as a function of fractional length for sequences derived from linear combinations of quartic characters (with $\cos(2 \gamma_p)$ approaching $-1$)}\label{Cecilia}
\begin{tikzpicture}[gnuplot]
\path (0.000,0.000) rectangle (12.446,9.398);
\gpcolor{color=gp lt color border}
\gpsetlinetype{gp lt border}
\gpsetlinewidth{1.00}
\draw[gp path] (1.504,0.985)--(1.684,0.985);
\draw[gp path] (11.893,0.985)--(11.713,0.985);
\node[gp node right] at (1.320,0.985) { 0};
\draw[gp path] (1.504,2.594)--(1.684,2.594);
\draw[gp path] (11.893,2.594)--(11.713,2.594);
\node[gp node right] at (1.320,2.594) { 0.2};
\draw[gp path] (1.504,4.203)--(1.684,4.203);
\draw[gp path] (11.893,4.203)--(11.713,4.203);
\node[gp node right] at (1.320,4.203) { 0.4};
\draw[gp path] (1.504,5.811)--(1.684,5.811);
\draw[gp path] (11.893,5.811)--(11.713,5.811);
\node[gp node right] at (1.320,5.811) { 0.6};
\draw[gp path] (1.504,7.420)--(1.684,7.420);
\draw[gp path] (11.893,7.420)--(11.713,7.420);
\node[gp node right] at (1.320,7.420) { 0.8};
\draw[gp path] (1.504,9.029)--(1.684,9.029);
\draw[gp path] (11.893,9.029)--(11.713,9.029);
\node[gp node right] at (1.320,9.029) { 1};
\draw[gp path] (1.504,0.985)--(1.504,1.165);
\draw[gp path] (1.504,9.029)--(1.504,8.849);
\node[gp node center] at (1.504,0.677) { 0};
\draw[gp path] (3.582,0.985)--(3.582,1.165);
\draw[gp path] (3.582,9.029)--(3.582,8.849);
\node[gp node center] at (3.582,0.677) { 2};
\draw[gp path] (5.660,0.985)--(5.660,1.165);
\draw[gp path] (5.660,9.029)--(5.660,8.849);
\node[gp node center] at (5.660,0.677) { 4};
\draw[gp path] (7.737,0.985)--(7.737,1.165);
\draw[gp path] (7.737,9.029)--(7.737,8.849);
\node[gp node center] at (7.737,0.677) { 6};
\draw[gp path] (9.815,0.985)--(9.815,1.165);
\draw[gp path] (9.815,9.029)--(9.815,8.849);
\node[gp node center] at (9.815,0.677) { 8};
\draw[gp path] (11.893,0.985)--(11.893,1.165);
\draw[gp path] (11.893,9.029)--(11.893,8.849);
\node[gp node center] at (11.893,0.677) { 10};
\draw[gp path] (1.504,9.029)--(1.504,0.985)--(11.893,0.985)--(11.893,9.029)--cycle;
\node[gp node center,rotate=-270] at (0.246,5.007) {Demerit Factor};
\node[gp node center] at (6.698,0.215) {Fractional Length};
\gpcolor{rgb color={0.000,0.000,0.000}}
\gpsetlinetype{gp lt plot 0}
\gpsetlinewidth{2.00}
\draw[gp path] (1.609,8.487)--(1.714,7.946)--(1.819,7.404)--(1.924,6.862)--(2.029,6.321)%
  --(2.134,5.779)--(2.239,5.237)--(2.344,4.696)--(2.448,4.154)--(2.553,3.614)--(2.658,3.231)%
  --(2.763,3.021)--(2.868,2.902)--(2.973,2.825)--(3.078,2.764)--(3.183,2.700)--(3.288,2.627)%
  --(3.393,2.537)--(3.498,2.427)--(3.603,2.300)--(3.708,2.201)--(3.813,2.139)--(3.918,2.101)%
  --(4.023,2.074)--(4.127,2.052)--(4.232,2.029)--(4.337,2.001)--(4.442,1.964)--(4.547,1.918)%
  --(4.652,1.862)--(4.757,1.818)--(4.862,1.790)--(4.967,1.772)--(5.072,1.758)--(5.177,1.747)%
  --(5.282,1.735)--(5.387,1.720)--(5.492,1.700)--(5.597,1.674)--(5.702,1.643)--(5.807,1.619)%
  --(5.911,1.603)--(6.016,1.592)--(6.121,1.585)--(6.226,1.578)--(6.331,1.570)--(6.436,1.561)%
  --(6.541,1.548)--(6.646,1.531)--(6.751,1.511)--(6.856,1.496)--(6.961,1.486)--(7.066,1.480)%
  --(7.171,1.475)--(7.276,1.470)--(7.381,1.465)--(7.486,1.458)--(7.590,1.449)--(7.695,1.438)%
  --(7.800,1.424)--(7.905,1.414)--(8.010,1.407)--(8.115,1.402)--(8.220,1.399)--(8.325,1.395)%
  --(8.430,1.392)--(8.535,1.387)--(8.640,1.380)--(8.745,1.371)--(8.850,1.361)--(8.955,1.354)%
  --(9.060,1.349)--(9.165,1.346)--(9.270,1.343)--(9.374,1.341)--(9.479,1.338)--(9.584,1.334)%
  --(9.689,1.329)--(9.794,1.322)--(9.899,1.314)--(10.004,1.309)--(10.109,1.305)--(10.214,1.303)%
  --(10.319,1.301)--(10.424,1.299)--(10.529,1.297)--(10.634,1.293)--(10.739,1.289)--(10.844,1.284)%
  --(10.949,1.278)--(11.053,1.274)--(11.158,1.271)--(11.263,1.269)--(11.368,1.267)--(11.473,1.266)%
  --(11.578,1.264)--(11.683,1.261)--(11.788,1.258)--(11.893,1.253);
\gpsetlinewidth{1.00}
\gpsetpointsize{2.00}
\gppoint{gp mark 7}{(1.712,9.029)}
\gppoint{gp mark 7}{(1.687,7.241)}
\gppoint{gp mark 7}{(1.813,6.769)}
\gppoint{gp mark 7}{(1.915,6.777)}
\gppoint{gp mark 7}{(2.026,6.414)}
\gppoint{gp mark 7}{(2.127,5.809)}
\gppoint{gp mark 7}{(2.232,5.468)}
\gppoint{gp mark 7}{(2.336,4.898)}
\gppoint{gp mark 7}{(2.439,4.304)}
\gppoint{gp mark 7}{(2.543,3.705)}
\gppoint{gp mark 7}{(2.647,3.295)}
\gppoint{gp mark 7}{(2.751,3.056)}
\gppoint{gp mark 7}{(2.855,2.921)}
\gppoint{gp mark 7}{(2.959,2.834)}
\gppoint{gp mark 7}{(3.062,2.798)}
\gppoint{gp mark 7}{(3.166,2.700)}
\gppoint{gp mark 7}{(3.270,2.636)}
\gppoint{gp mark 7}{(3.374,2.559)}
\gppoint{gp mark 7}{(3.478,2.455)}
\gppoint{gp mark 7}{(3.582,2.318)}
\gppoint{gp mark 7}{(3.686,2.221)}
\gppoint{gp mark 7}{(3.790,2.152)}
\gppoint{gp mark 7}{(3.893,2.113)}
\gppoint{gp mark 7}{(3.997,2.082)}
\gppoint{gp mark 7}{(4.101,2.070)}
\gppoint{gp mark 7}{(4.205,2.033)}
\gppoint{gp mark 7}{(4.309,2.011)}
\gppoint{gp mark 7}{(4.413,1.974)}
\gppoint{gp mark 7}{(4.517,1.939)}
\gppoint{gp mark 7}{(4.621,1.888)}
\gppoint{gp mark 7}{(4.725,1.829)}
\gppoint{gp mark 7}{(4.828,1.800)}
\gppoint{gp mark 7}{(4.932,1.781)}
\gppoint{gp mark 7}{(5.036,1.764)}
\gppoint{gp mark 7}{(5.140,1.750)}
\gppoint{gp mark 7}{(5.244,1.741)}
\gppoint{gp mark 7}{(5.348,1.729)}
\gppoint{gp mark 7}{(5.452,1.707)}
\gppoint{gp mark 7}{(5.556,1.685)}
\gppoint{gp mark 7}{(5.660,1.658)}
\gppoint{gp mark 7}{(5.763,1.629)}
\gppoint{gp mark 7}{(5.867,1.609)}
\gppoint{gp mark 7}{(5.971,1.597)}
\gppoint{gp mark 7}{(6.075,1.588)}
\gppoint{gp mark 7}{(6.179,1.581)}
\gppoint{gp mark 7}{(6.283,1.574)}
\gppoint{gp mark 7}{(6.387,1.568)}
\gppoint{gp mark 7}{(6.491,1.556)}
\gppoint{gp mark 7}{(6.595,1.541)}
\gppoint{gp mark 7}{(6.699,1.524)}
\gppoint{gp mark 7}{(6.802,1.504)}
\gppoint{gp mark 7}{(6.906,1.491)}
\gppoint{gp mark 7}{(7.010,1.484)}
\gppoint{gp mark 7}{(7.114,1.477)}
\gppoint{gp mark 7}{(7.218,1.473)}
\gppoint{gp mark 7}{(7.322,1.469)}
\gppoint{gp mark 7}{(7.426,1.463)}
\gppoint{gp mark 7}{(7.530,1.455)}
\gppoint{gp mark 7}{(7.634,1.446)}
\gppoint{gp mark 7}{(7.737,1.432)}
\gppoint{gp mark 7}{(7.841,1.420)}
\gppoint{gp mark 7}{(7.945,1.411)}
\gppoint{gp mark 7}{(8.049,1.406)}
\gppoint{gp mark 7}{(8.153,1.402)}
\gppoint{gp mark 7}{(8.257,1.397)}
\gppoint{gp mark 7}{(8.361,1.395)}
\gppoint{gp mark 7}{(8.465,1.391)}
\gppoint{gp mark 7}{(8.569,1.386)}
\gppoint{gp mark 7}{(8.672,1.378)}
\gppoint{gp mark 7}{(8.776,1.367)}
\gppoint{gp mark 7}{(8.880,1.359)}
\gppoint{gp mark 7}{(8.984,1.353)}
\gppoint{gp mark 7}{(9.088,1.349)}
\gppoint{gp mark 7}{(9.192,1.345)}
\gppoint{gp mark 7}{(9.296,1.343)}
\gppoint{gp mark 7}{(9.400,1.340)}
\gppoint{gp mark 7}{(9.504,1.337)}
\gppoint{gp mark 7}{(9.607,1.333)}
\gppoint{gp mark 7}{(9.711,1.328)}
\gppoint{gp mark 7}{(9.815,1.320)}
\gppoint{gp mark 7}{(9.919,1.313)}
\gppoint{gp mark 7}{(10.023,1.308)}
\gppoint{gp mark 7}{(10.127,1.305)}
\gppoint{gp mark 7}{(10.231,1.303)}
\gppoint{gp mark 7}{(10.335,1.301)}
\gppoint{gp mark 7}{(10.439,1.299)}
\gppoint{gp mark 7}{(10.542,1.296)}
\gppoint{gp mark 7}{(10.646,1.293)}
\gppoint{gp mark 7}{(10.750,1.289)}
\gppoint{gp mark 7}{(10.854,1.283)}
\gppoint{gp mark 7}{(10.958,1.277)}
\gppoint{gp mark 7}{(11.062,1.274)}
\gppoint{gp mark 7}{(11.166,1.271)}
\gppoint{gp mark 7}{(11.270,1.269)}
\gppoint{gp mark 7}{(11.374,1.267)}
\gppoint{gp mark 7}{(11.477,1.266)}
\gppoint{gp mark 7}{(11.581,1.264)}
\gppoint{gp mark 7}{(11.685,1.261)}
\gppoint{gp mark 7}{(11.789,1.258)}
\gppoint{gp mark 7}{(11.893,1.254)}
\gpcolor{color=gp lt color border}
\gpsetlinetype{gp lt border}
\draw[gp path] (1.504,9.029)--(1.504,0.985)--(11.893,0.985)--(11.893,9.029)--cycle;
\gpdefrectangularnode{gp plot 1}{\pgfpoint{1.504cm}{0.985cm}}{\pgfpoint{11.893cm}{9.029cm}}
\end{tikzpicture}
\end{center}
\end{figure}
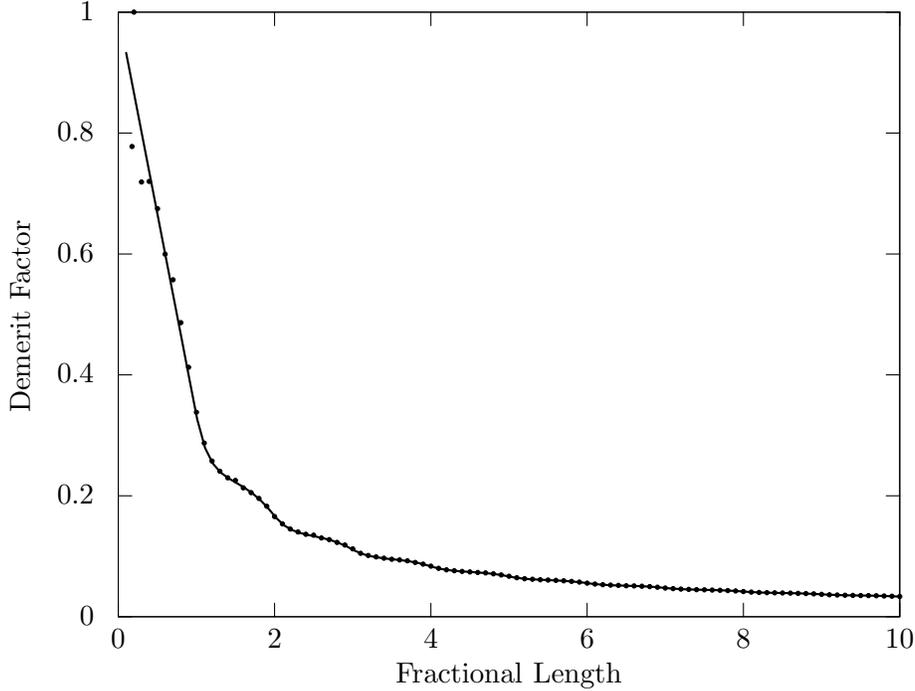
\end{center}

We conclude with a proof that the limiting crosscorrelation demerit factor tends to $0$ in this scenario.
\begin{theorem}\label{Irene}
There exists an infinite family $\{(f_\iota,g_\iota)\}_{\iota \in I}$ of pairs of quartic residue sequences such that for each $\iota \in I$, both $f_\iota$ and $g_\iota$ have prime $p_\iota$, and the set $\{p_\iota : \iota \in I\}$ is infinite, and $\CDF(f_\iota,g_\iota) \to 0$ as $p_\iota\to \infty$.
If $f^u_\iota$ and $g^u_\iota$ are unimodularizations of $f_\iota$ and $g_\iota$, respectively, for each $\iota \in I$, then $\CDF(f^u_\iota,g^u_\iota) \to 0$ as $p_\iota \to\infty$.
\end{theorem}
\begin{proof}
Let $\Lambda_1,\Lambda_2,\ldots$ be a sequence of positive real numbers chosen so that $\lim_{n\to\infty} \Lambda_n=\infty$.
For each positive integer $n$, let ${\mathcal F}_n$ be a family of pairs of quartic residue sequences that meets the hypotheses of Theorem \ref{Christopher} with the parameters $\Lambda$ and $\gamma$ in that theorem being equal to $\Lambda_n$ and $\pi/2$, respectively.  Then the same theorem tells us that the limiting value of crosscorrelation demerit factor for the sequences in $F_n$ as their prime tends to infinity is $-2\Lambda_n/3+\Omega(1/\Lambda_n,0)$.  For each $n$, select a sequence pair $(f_n,g_n)$ from ${\mathcal F}_n$ with prime $p_n$ and length $\ell_n$ such that
\begin{align*}
\left|\frac{\ell_n}{p_n} -\Lambda_n\right| & < 1 \\
\left|\CDF(f_n,g_n) -\left(-\frac{2\Lambda_n}{3}+\Omega\Big(\frac{1}{\Lambda_n},0\Big)\right)\right| & < \frac{1}{n},
\end{align*}
and make sure that $p_n > p_{n-1}$ and $p_n \geq n(\Lambda_n+1)$.
Since Lemma \ref{Samantha} below shows that $\lim_{x \to \infty} -2x/3+\Omega(1/x,0) = 0$, we see that $\CDF(f_n,g_n) \to 0$ as $n \to \infty$ (equivalently, as $p_n \to \infty$).
We then let $f^u_n$ and $g^u_n$ be unimodularizations of $f_n$ and $g_n$ for each $n$.  Then apply Lemma \ref{Nancy} to see that $\CDF(f^u_n,g^u_n) \to 0 $ as $p_n \to \infty$.  To see that Lemma \ref{Nancy} applies, note that $\ell_n \geq p_n(\Lambda_n-1)$ and that $\Lambda_n\to\infty$ and $p_n \to \infty$ as $n\to\infty$.  Thus $\ell_n\to\infty$ as $n\to\infty$ (equivalently, as $p_n \to \infty$).  Furthermore $\frac{\ell_n}{p_n^2} < \frac{\Lambda_n+1}{p_n} \leq 1/n$, so that $\ell_n/p_n^2 \to 0$ as $n\to \infty$ (equivalently, as $p_n \to \infty$).
\end{proof}
We close this section with a technical lemma needed in the proof of Theorem \ref{Irene}.
\begin{lemma}\label{Samantha}
The function $f(x)=-2 x/3+\Omega(1/x,0)$ tends to $0$ as $x \to \infty$.
\end{lemma}
\begin{proof}
In \cite[Lemma 24]{Katz-2016-Aperiodic}, it is shown that if $m$ is a positive integer and $x \in [m,m+1]$, then
\[
\Omega\left(\frac{1}{x},0\right)= 2 m+1 -\frac{2 m(m+1)}{x} + \frac{m(m+1)(2 m+1)}{3 x^2},
\]
so that
\[
f(x) = \frac{-2 x^3+ 3(2 m+1) x^2 - 6 m(m+1) x + m(m+1)(2 m+1)}{3 x^2}
\]
for $x \in [m,m+1]$.  If we write $x=m+y$ with $y \in [0,1]$, then one can substitute $x-y$ for $m$ in the previous equation to obtain $f(x) = (x-y+3 y^2-2 y^3)/(3 x^2)$.  Then it is clear that $f(x) \to 0$ as $x \to \infty$.
\end{proof}

\section*{Acknowledgements}
The authors thank Jonathan Jedwab and Amy Wiebe for helpful comments about this work.  They also thank the anonymous reviewers for their corrections, which improved this paper.  The second author also thanks Jacob King for valuable corrections to the manuscript.


\begin{thebibliography}{10}

\bibitem{Berndt-Evans-1981-Determination}
B.~C. Berndt and R.~J. Evans.
\newblock The determination of {G}auss sums.
\newblock {\em Bull. Amer. Math. Soc. (N.S.)}, 5(2):107--129, 1981.
\newblock Correction in {\it Bull. Amer. Math. Soc. (N.S.)}, 5(2): 107--129,
  1982.

\bibitem{Boehmer-1967-Binary}
A.~Boehmer.
\newblock Binary pulse compression codes.
\newblock {\em IEEE Trans. Inform. Theory}, 13(2):156--167, 1967.

\bibitem{Borwein-Choi-2000-Merit}
P.~Borwein and K.-K.~S. Choi.
\newblock Merit factors of character polynomials.
\newblock {\em J. London Math. Soc. (2)}, 61(3):706--720, 2000.

\bibitem{Borwein-Choi-2002-Explicit}
P.~Borwein and K.-K.~S. Choi.
\newblock Explicit merit factor formulae for {F}ekete and {T}uryn polynomials.
\newblock {\em Trans. Amer. Math. Soc.}, 354(1):219--234, 2002.

\bibitem{Borwein-Choi-Jedwab-2004-Binary}
P.~Borwein, K.-K.~S. Choi, and J.~Jedwab.
\newblock Binary sequences with merit factor greater than 6.34.
\newblock {\em IEEE Trans. Inform. Theory}, 50(12):3234--3249, 2004.

\bibitem{Ding-Helleseth-Lam-2000-Duadic}
C.~Ding, T.~Helleseth, and K.~Y. Lam.
\newblock Duadic sequences of prime lengths.
\newblock {\em Discrete Math.}, 218(1-3):33--49, 2000.

\bibitem{Golay-1972-Class}
M.~J.~E. Golay.
\newblock A class of finite binary sequences with alternate auto-correlation
  values equal to zero.
\newblock {\em IEEE Trans. Inform. Theory}, 18(3):449--450, 1972.

\bibitem{Golay-1983-Merit}
M.~J.~E. Golay.
\newblock The merit factor of {L}egendre sequences.
\newblock {\em IEEE Trans. Inform. Theory}, 29(6):934--936, 1983.

\bibitem{Golomb-Gong-2005-Signal}
S.~W. Golomb and G.~Gong.
\newblock {\em Signal design for good correlation}.
\newblock Cambridge University Press, Cambridge, 2005.

\bibitem{Gunther-Schmidt-2015-Merit}
C.~G\"unther and K.-U. Schmidt.
\newblock Merit factors of polynomials derived from difference sets.
\newblock {\em J. Combin. Theory Ser. A}, 145:340--363, 2017.

\bibitem{Heath-Brown-Patterson-1979-Distribution}
D.~R. Heath-Brown and S.~J. Patterson.
\newblock The distribution of {K}ummer sums at prime arguments.
\newblock {\em J. Reine Angew. Math.}, 310:111--130, 1979.

\bibitem{Hoholdt-Jensen-1988-Determination}
T.~H{\o}holdt and H.~E. Jensen.
\newblock Determination of the merit factor of {L}egendre sequences.
\newblock {\em IEEE Trans. Inform. Theory}, 34(1):161--164, 1988.

\bibitem{Jedwab-Katz-Schmidt-2013-Advances}
J.~Jedwab, D.~J. Katz, and K.-U. Schmidt.
\newblock Advances in the merit factor problem for binary sequences.
\newblock {\em J. Combin. Theory Ser. A}, 120(4):882--906, 2013.

\bibitem{Jedwab-Katz-Schmidt-2013-Littlewood}
J.~Jedwab, D.~J. Katz, and K.-U. Schmidt.
\newblock Littlewood polynomials with small {$L^4$} norm.
\newblock {\em Adv. Math.}, 241:127--136, 2013.

\bibitem{Jedwab-Schmidt-2010-Appended}
J.~Jedwab and K.-U. Schmidt.
\newblock Appended {$m$}-sequences with merit factor greater than 3.34.
\newblock In C.~Carlet and A.~Pott, editors, {\em Sequences and Their
  Applications – SETA 2010}, volume 6338 of {\em Lecture Notes in Computer
  Science}, pages 204--216. Springer Berlin / Heidelberg, 2010.

\bibitem{Jensen-Hoholdt-1989-Binary}
H.~E. Jensen and T.~H{\o}holdt.
\newblock Binary sequences with good correlation properties.
\newblock In {\em Applied Algebra, Algebraic Algorithms and Error-Correcting
  Codes}, volume 356 of {\em Lecture Notes in Comput. Sci.}, pages 306--320.
  Springer, Berlin, 1989.

\bibitem{Jensen-Jensen-Hoholdt-1991-Merit}
J.~M. Jensen, H.~E. Jensen, and T.~H{\o}holdt.
\newblock The merit factor of binary sequences related to difference sets.
\newblock {\em IEEE Trans. Inform. Theory}, 37(3):617--626, 1991.

\bibitem{Karkkainen-1992-Mean}
K.~K{\"a}rkk{\"a}inen.
\newblock Mean-square cross-correlation as a performance measure for department
  of spreading code families.
\newblock In {\em IEEE Second International Symposium on Spread Spectrum
  Techniques and Applications}, pages 147--150, 1992.

\bibitem{Katz-2013-Asymptotic}
D.~J. Katz.
\newblock Asymptotic {$L^4$} norm of polynomials derived from characters.
\newblock {\em Pacific J. Math.}, 263(2):373--398, 2013.

\bibitem{Katz-2016-Aperiodic}
D.~J. Katz.
\newblock Aperiodic crosscorrelation of sequences derived from characters.
\newblock {\em IEEE Transactions on Information Theory}, 62(9):5237--5259,
  2016.

\bibitem{Kirilusha-Narayanaswamy-1999-Construction}
A.~Kirilusha and G.~Narayanaswamy.
\newblock Construction of new asymptotic classes of binary sequences based on
  existing asymptotic classes.
\newblock Summer {S}cience {T}ech. {R}ep., Dept. Math. Comput. Sci., Univ.
  Richmond, VA, 1999.

\bibitem{Lidl-Niederreiter-1997-Finite}
R.~Lidl and H.~Niederreiter.
\newblock {\em Finite {F}ields}, volume~20 of {\em Encyclopedia of Mathematics
  and its Applications}.
\newblock Cambridge University Press, Cambridge, second edition, 1997.

\bibitem{Patterson-1981-Distribution}
S.~J. Patterson.
\newblock The distribution of general gauss sums at prime arguments.
\newblock In H.~Halberstam and C.~Hooley, editors, {\em Recent progress in
  analytic number theory. {V}ol. 2}, pages 171--182. Academic Press, New York,
  1981.

\bibitem{Patterson-1987-Distribution}
S.~J. Patterson.
\newblock The distribution of general {G}auss sums and similar arithmetic
  functions at prime arguments.
\newblock {\em Proc. London Math. Soc. (3)}, 54(2):193--215, 1987.

\bibitem{Pursley-Sarwate-1976-Bounds}
M.~B. Pursley and D.~V. Sarwate.
\newblock Bounds on aperiodic cross-correlation for binary sequences.
\newblock {\em Electronics Letters}, 12(12):304--305, 1976.

\bibitem{Sarwate-1984-Mean}
D.~Sarwate.
\newblock Mean-square correlation of shift-register sequences.
\newblock {\em Communications, Radar and Signal Processing, IEE Proceedings F},
  131(2):101--106, 1984.

\bibitem{Sarwate-1984-Upper}
D.~Sarwate.
\newblock An upper bound on the aperiodic autocorrelation function for a
  maximal-length sequence (corresp.).
\newblock {\em IEEE Trans. Inform. Theory}, 30(4):685--687, 1984.

\bibitem{Sarwate-Pursley-1980-Crosscorrelation}
D.~V. Sarwate and M.~B. Pursley.
\newblock Crosscorrelation properties of pseudorandom and related sequences.
\newblock {\em IEEE Trans. Inform. Theory}, 68(5):593--619, 1980.
\newblock Correction in IEEE Trans. Inform. Theory 68(12):1554, 1980.

\bibitem{Schroeder-2006-Number}
M.~R. Schroeder.
\newblock {\em Number {T}heory in {S}cience and {C}ommunication}, volume~7 of
  {\em Springer Series in Information Sciences}.
\newblock Springer-Verlag, Berlin, fourth edition, 2006.

\bibitem{Turyn-1960-Optimum}
R.~Turyn.
\newblock Optimum codes study, final report.
\newblock Technical report, Sylvania Electronic Products Inc., Waltham, MA,
  1960.

\end{thebibliography}
\end{document}